\documentclass[a4paper]{article}

\usepackage[margin=1in]{geometry}  
\usepackage[english]{babel}
\usepackage[utf8x]{inputenc}
\usepackage[T1]{fontenc}
\usepackage{parskip}
\usepackage[export]{adjustbox}
 
\usepackage[dvipsnames]{xcolor}
\usepackage{graphicx}
\usepackage[font = small]{caption}
\usepackage{subcaption}
\usepackage{longtable} 
\usepackage{multirow}
\usepackage{listings}
\usepackage{makecell}
\usepackage{array}
\usepackage{float}
\usepackage{dsfont}
\usepackage{rotating}
\usepackage{booktabs}
\usepackage{enumerate}
\usepackage{tikz}
\usetikzlibrary{chains}
\usetikzlibrary{arrows}
\usepackage{pgf}
\usepackage{amsmath}
\usepackage{amssymb}
\usepackage{amsthm}
\usepackage{bm}
\usepackage{algorithm,algpseudocodex}

\usepackage{mathtools}
\usepackage{mathrsfs} 

\usepackage{natbib}
\usepackage[
  colorlinks,
  citecolor=blue,
  linkcolor=red,
  anchorcolor=red,
  urlcolor=blue
]{hyperref}
\mathtoolsset{showonlyrefs}
\usepackage{authblk}
\usepackage{todonotes}
\theoremstyle{plain}

\newtheorem{theorem}{Theorem}
\newtheorem{lemma}{Lemma}

\newtheorem{remark}{Remark}

\newtheorem{proposition}{Proposition}

\newcommand{\RR}{\mathbb{R}}
\newcommand{\EE}{\mathbb{E}}
\newcommand{\PP}{\mathbb{P}}
\newcommand{\ind}{\mathds{1}}

\usepackage{xspace}

\newcommand{\stepa}[1]{\overset{\rm (a)}{#1}}
\newcommand{\stepb}[1]{\overset{\rm (b)}{#1}}
\newcommand{\stepc}[1]{\overset{\rm (c)}{#1}}

\newcommand{\step}[2]{\overset{\rm (#1)}{#2}}

\newcommand{\given}{{\,|\,}}
\newcommand{\biggiven}{\,\big{|}\,}
\newcommand{\Biggiven}{\,\Big{|}\,}

\def\@#1\@{\begin{align}#1\end{align}}
\def\$#1\${\begin{align*}#1\end{align*}}

\definecolor{myblue}{rgb}{.8, .8, 1}
\definecolor{mathblue}{rgb}{0.2472, 0.24, 0.6} 
\definecolor{mathred}{rgb}{0.6, 0.24, 0.442893}
\definecolor{mathyellow}{rgb}{0.6, 0.547014, 0.24}

\newcommand{\tR}{{\tilde{R}}}
\newcommand{\tT}{{\tilde{T}}}

\newcommand{\cD}{{\mathcal{D}}}

\newcommand{\cI}{{\mathcal{I}}}

\newcommand{\cK}{{\mathcal{K}}}
\newcommand{\cL}{{\mathcal{L}}}

\newcommand{\cN}{{\mathcal{N}}}

\newcommand{\cR}{{\mathcal{R}}}
\newcommand{\cS}{{\mathcal{S}}}
\newcommand{\cT}{{\mathcal{T}}}

\newcommand{\cX}{{\mathcal{X}}}
\newcommand{\cY}{{\mathcal{Y}}}

\def\##1\#{\begin{align}#1\end{align}}
\def\$#1\${\begin{align*}#1\end{align*}}

\renewcommand{\hat}{\widehat}
\renewcommand{\tilde}{\widetilde}

\newcommand{\hcS}{\hat{\cS}}

\newcommand{\hC}{\hat{C}}
\newcommand{\hcR}{\hat{\cR}}
\newcommand{\hS}{\hat{S}}

\newcommand{\cdc}{\cD_{\textnormal{calib}}}
\newcommand{\cdt}{\cD_{\textnormal{train}}}
\newcommand{\cdte}{\cD_{\textnormal{test}}}
\newcommand{\cic}{[n]}

\newcommand{\citest}{[m]}
\newcommand{\quant}{{\textnormal{Quantile}}}
\newcommand{\test}{{\textnormal{test}}}
\newcommand{\calib}{{\textnormal{calib}}}
\newcommand{\train}{{\textnormal{train}}}

\newcommand{\rand}{{\textnormal{rand}}}
\newcommand{\swap}[2]{{\textnormal{swap}(#1,#2)}}
\newcommand{\mfS}{\mathfrak{S}}
\newcommand{\topk}{\textnormal{topk}}
\newcommand{\cq}{\textnormal{cq}}
\newcommand{\jq}{\textnormal{jq}}
\newcommand{\mname}{{\textnormal{JOMI}}}
\newcommand{\fixed}{\textnormal{fixed}}
\newcommand{\bh}{\textnormal{BH}}
\newcommand{\ps}{\textnormal{ps}}

\newcommand{\etas}{\eta^{\swap{i}{j}}(y)}
\newcommand{\cp}{\textnormal{cp}}

\newcommand{\tS}{{\tilde{\mathcal{S}}}}

\long\def\comment#1{}

\makeatletter
\newcommand{\printfnsymbol}[1]{%
  \textsuperscript{\@fnsymbol{#1}}%
}
\makeatother
\title{Confidence on the Focal: Conformal Prediction \\ with 
Selection-Conditional Coverage}
\author[1]{Ying Jin\thanks{Author names listed alphabetically.}} 
\author[2]{Zhimei Ren\printfnsymbol{1}}
\affil[1]{Data Science Initiative and Department of Health Care Policy, Harvard University}
\affil[2]{Department of Statistics and Data Science, University of Pennsylvania}
\date{}

\begin{document}
\maketitle
\begin{abstract}
  Conformal prediction builds marginally valid prediction intervals that 
  cover the unknown outcome of a randomly drawn test point  
  with a prescribed probability. 
  However, in practice, data-driven methods are often used to identify specific 
  test unit(s) of interest, requiring uncertainty quantification tailored to 
  these focal units.
  In such cases,  marginally valid conformal prediction intervals 
  may fail to provide valid coverage
  for the  focal unit(s)  
  due to selection bias.
  This paper presents a general framework for constructing a prediction set 
  with finite-sample exact coverage,  
  conditional on the unit being selected by a given procedure. 
  The general form of our method accommodates arbitrary selection rules that are invariant to the 
  permutation of the calibration units, 
  and generalizes Mondrian Conformal Prediction to multiple test units and 
  non-equivariant classifiers.
  We also work out computationally efficient 
  implementation of our framework for a number of realistic selection rules, 
  including top-K selection, optimization-based selection, selection based on conformal p-values,
  and selection based on properties of preliminary conformal prediction sets.
  The performance of our methods is demonstrated via applications in drug discovery 
  and health risk prediction. 
\end{abstract}
\section{Introduction}

Conformal prediction is a versatile framework for quantifying the uncertainty of any black-box prediction model, by issuing a prediction set that covers the unknown outcome with a prescribed probability. 
Formally, suppose the task is to predict an outcome $Y\in \cY$ 
based on features $X\in \cX$. 
Given a set of calibration data $\{(X_i,Y_i)\}_{i=1}^{n}$ 
and the features of a new test point $X_{n+1}$, 
conformal prediction builds upon a given prediction model and delivers a 
prediction set $\hat{C}_{\alpha, n+1}\subseteq \cY$ at level $\alpha\in (0,1)$, which obeys 
\#\label{eq:def_mgn_coverage}
\PP\big( Y_{n+1}\in \hat{C}_{\alpha,n+1} \big) \geq 1-\alpha,
\#
as long as $\{(X_i,Y_i)\}_{i=1}^{n+1}$ are exchangeable 
(e.g., when they are i.i.d.~samples).
The probability in~\eqref{eq:def_mgn_coverage} is over both the 
calibration data and the test point~\citep{vovk2005algorithmic,lei2018distribution}.  

With this finite-sample, distribution-free guarantee, 
the conformal prediction set $\hat{C}_{\alpha,n+1}$ describes a range of plausible values the unknown outcome $Y_{n+1}$ may take, thereby expected to inform downstream decision-making based on the black-box prediction model. 
With such a promise, methods for constructing \emph{marginally valid} -- in the sense of~\eqref{eq:def_mgn_coverage} --
prediction sets have been developed 
for various problems; see e.g.,~\cite{angelopoulos2021gentle} for a review. 


In many downstream applications, however, 
people are often only interested in a \emph{selective} subset of units. 
For example, practitioners may only act upon a unit if it 
exhibits an interesting property~\citep{olsson2022estimating,sokol2024conformalized,business_conformal},  
or focus only on a subset of test units picked by a
complicated data-dependent process  such as 
resource optimization~\citep{svensson2018maximizing,castro2012combined,kemper2014optimized,gocgun2014dynamic}.
It would be misleading for the practitioners if the prediction sets fail to deliver
the promised coverage guarantee for the selected unit(s).
Let us discuss a few applications where such cases may arise.  

\begin{itemize}
  \item In \emph{drug discovery}, an important task is to predict the binding affinity 
  of a drug candidate to a disease target, which informs
  subsequent  drug prioritization~\citep{laghuvarapu2024codrug}.
  Among many drug candidates, 
  scientists may only focus on those with highest predicted affinities, 
  or those selected by a false discovery rate (FDR)-controlling procedure~\citep{jin2022selection}, 
  or whose prediction sets only cover high values~\citep{svensson2017improving}, 
  or those optimizing resource usage~\citep{svensson2018maximizing}. 
  It may lead to a waste of resources if the 
  prediction sets for the selected drugs fail to cover the actual binding affinities with an exceedingly high chance.
  \item In \emph{business decision-making}, companies may take
  different inventory decisions based on whether a conformal prediction set
  suggests a strong demand or a weak demand~\citep{business_conformal}. 
  Similarly, it will be problematic if a strong-demand prediction cannot cover with at least  $1-\alpha$ of the time. 
  \item In \emph{disease diagnosis},~\cite{olsson2022estimating} suggest 
  human intervention if the prediction set for a disease status 
  is too large 
  (this implicitly declares confidence in small prediction sets; 
  see similar ideas in~\cite{ren2023robots,sokol2024conformalized}). 
  However, it would be concerning if, with more than a chance $\alpha$, the 
  small-sized prediction sets -- ``approved'' as confident -- miss the true disease status. 
  \item In \emph{healthcare management}, patients may be sent to different healthcare categories
  based a program that optimizes some performance measure 
  (such as waiting time) subject to certain constraints, such as budget, 
  capacity, or fairness~\citep{castro2012combined,kemper2014optimized,gocgun2014dynamic}. 
  The subset of patients in each category is therefore data-dependent.
\end{itemize}

In all these examples,  
it is highly desirable that a prediction set should cover 
the unknown outcome for a unit \emph{of interest} with a prescribed probability. 
This motivates 
a stronger, selection-conditional guarantee. 
Supposing there are $m\geq 1$ test units $\cD_\test = \{X_{n+j}\}_{j=1}^m$, 
we  aim for
\#\label{eq:def_sel_coverage}
\PP \big(Y_{n+j}\in \hat{C}_{\alpha,n+j}\given  j\in \cS(\cD_\calib,\cD_\test)\big) \geq 1-\alpha,
\#
where $\cD_\calib = \{(X_i,Y_i)\}_{i=1}^n$ is the calibration data, and 
 $\cS(\cdot,\cdot)$ is 
 a data-driven process 
 to decide the units of interest, which maps 
 the calibration and test data to a subset of $[m]:=\{1,\dots,m\}$. 
 Our target is similar in spirit to post-selection 
 inference~\citep{lee2016exact,markovic2017unifying,tibshirani2018uniform}, 
 but we consider predictive inference settings and develop distinct techniques; 
 see Supplementary Section~\ref{app:overview} for more discussion. 
 Throughout, we focus on settings where the prediction sets are 
 constructed after $\cS(\cD_\calib,\cD_\test)$ is determined.  
 By separating  the selection process from the (post hoc) uncertainty quantification step, 
 we leave the freedom of defining selection rules to the practitioners.

A prediction set with marginal validity~\eqref{eq:def_mgn_coverage} 
does not necessarily cover an unknown outcome \emph{conditional on} 
being of interest as in~\eqref{eq:def_sel_coverage}. 
Such a selection issue
has recently been raised in the literature of predictive inference: 
through analysis of a real drug 
discovery dataset, 
\citet[Section 1]{jin2022selection} demonstrate 
that more than $30\%$ 
of seemingly promising prediction sets (in the sense that 
$\hat{C}_{\alpha,n+j}=\{\text{active}\}$) miss the actual outcomes 
when the nominal marginal miscoverage rate 
is $\alpha=0.01$.
We shall see more examples of this issue in our numerical experiments 
with various selection processes.  
 New techniques are therefore needed for constructing prediction sets achieving~\eqref{eq:def_sel_coverage}. 

\subsection{Exchangeability via reference sets}

Accounting for  selection in conformal prediction 
is a delicate task since it breaks exchangeability. 
The core of conformal prediction is to leverage the exchangeability 
among the calibration and test data, 
such that their ``prediction residuals'', referred to as \emph{nonconformity scores}, 
are comparable in distribution (see Section~\ref{subsec:SCP} for more details).  
The  calibration scores therefore
inform the magnitude of uncertainty in a new test point~\citep{vovk2005algorithmic}. 
However, given a selection event, the calibration data are no longer exchangeable 
with the test point, leading to violation of the coverage guarantee mentioned above. 
In general, the selection-conditional distributions of these scores are complex 
since the selection event 
can be highly data-dependent.

Such a challenge motivates   
our new framework, 
named \underline{JO}int \underline{M}ondrian Conformal \underline{I}nference (JOMI), 
which builds prediction sets that achieve selection-conditional coverage, 
with~\eqref{eq:def_sel_coverage} as a special case. 
As visualized in Figure~\ref{fig:exch}, our key idea is to find a ``reference set'' ---  
a data-dependent subset of calibration data that 
remain exchangeable with respect to the new test point conditional on the selection event.   
This reference set thus provides calibrated quantification of 
uncertainty for a selected unit. 
The mechanism we devise accommodates arbitrary selection rules $\cS(\cD_\calib,\cD_\test)$ 
that are invariant to permutations of data in $\cD_\calib$, and  
can be computed efficiently for a wide range of commonly used selection rules.

\begin{figure}
  \centering 
  \includegraphics[width=0.85\linewidth]{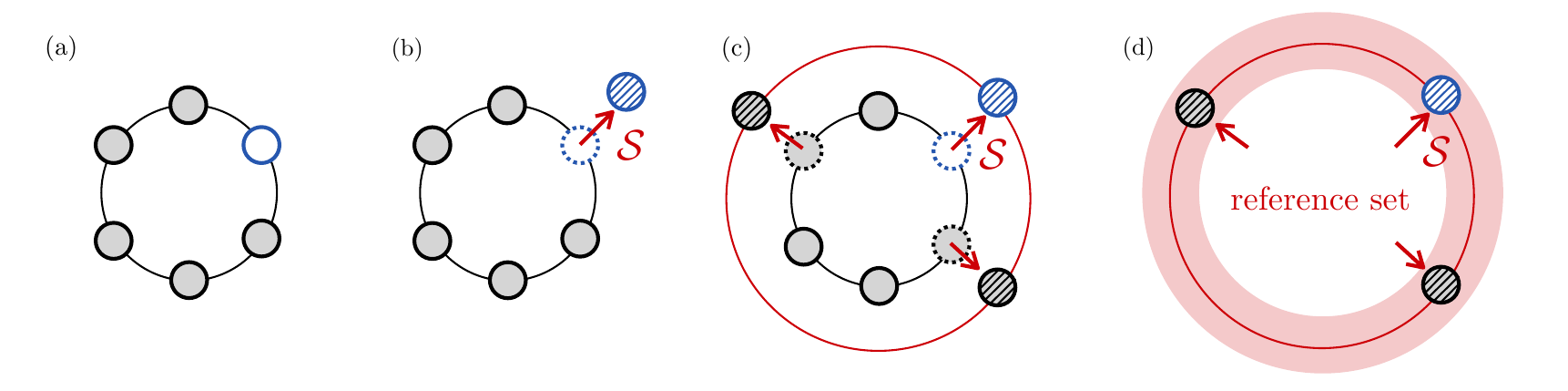}
  \caption{Visualization of the intuition behind the reference set. 
  (a) Marginally, the calibration data (black) are exchangeable with respect to the test point (blue). 
  (b) The calibration data are not exchangeable with respect to the test point (shaded  blue) given a selection event.
  (c) We find calibration data which, when posited as a ``test point'', would lead to the 
  same selection event. 
  (d) The reference set consists of calibration data that are exchangeable with respect to 
  the test point given selection, and we use them to construct JOMI prediction sets.}
  \label{fig:exch}
\end{figure}


We also note that the selection-conditional coverage~\eqref{eq:def_sel_coverage} may be implied by other 
stronger notions  such as 
 conditional coverage: $\PP(Y_{n+1}\in \hat{C}_{\alpha,n+1}\given X_{n+1}=x) \geq 1-\alpha$ for $\PP$-almost all $x\in \cX$. 
 Conditional coverage implies~\eqref{eq:def_sel_coverage} 
 if $\{(X_i,Y_i)\}_{i=1}^{n+1}$ are mutually independent 
 and the selection rule only depends on test features. 
 However, it is not achievable 
 by finite-length prediction intervals without distributional 
 assumptions~\citep{vovk2005algorithmic,foygel2021limits}, and  
practical selection rules may also depend on other information. 
Given these considerations, 
we may view~\eqref{eq:def_sel_coverage} 
 as a ``relaxed'' version of conditional coverage that is both achievable 
 and relevant for practical use.

\subsection{Preview of contributions}

In Section~\ref{sec:unified}, 
we introduce the general formulation of JOMI, 
including the construction of  reference set  and 
its use in deriving a selection-conditionally valid prediction set.  
We prove that under exchangeability, 
the prediction set $\hat{C}_{\alpha,n+j}$ produced by JOMI covers $Y_{n+j}$ 
with probability at least $1-\alpha$ conditional on a selection event.
The selection event can be ``test unit $j$ is selected'', 
which leads to~\eqref{eq:def_sel_coverage}; 
it can also be more granular, such as ``test unit $j$ is selected, 
and there are $k$ selected test units''. 
This framework is valid for any selection rule $\cS(\cD_\calib,\cD_\test)$ that 
is permutation-invariant to $\cD_\calib$, 
without any modeling assumptions on the data generating process.  

In Section~\ref{sec:compute}, we study the computational aspect of 
JOMI. We show that when $|\cY| < \infty$, our method can be computed with 
a worst-case complexity of $O(|\cY|mn)$ 
times that of the selection rule. Moreover, for general continuous outcome space $\cY$, we work out 
efficient implementation of JOMI for a number of 
selection processes that may be of practical interest, including:
\begin{itemize}
  \item \emph{Covariate-dependent selection.} When the selection rule 
  does not involve $\{Y_i\}_{i=1}^n$, the computation complexity of our generic method 
  is (at most) $O(mn)$ times that of the selection rule. 
  This implies efficient implementation for a wide range of problems, 
  including various forms of top-K selection (for which the computation can be further reduced to $O(m+n)$) and selection based on 
  complicated constrained optimization programs. 
  As a special case, we recover the method of \cite{bao2024selective} 
  for top-K selection among test data, and provide valid solutions 
  to other ranking-based selection methods attempted in their work. 
  \item \emph{Conformal p-value-based selection.} 
  We derive efficient 
  implementation for a class of selection rules based on thresholding 
  conformal p-values with ``stopping time-type'' cutoffs. 
  For instance, our prediction set is finite-sample exactly valid for 
  units selected by the Conformal Selection method~\citep{jin2022selection}, 
  which is studied in~\cite{bao2024selective} 
  with approximate FCR guarantee. 
  \item \emph{Selection based on preliminary conformal prediction sets.} 
  In addition, we derive a general, efficient 
  instantiation when selecting units whose marginal 
  prediction sets demonstrate certain interesting properties, such as 
  being of a short length, or having a lower bound above some threshold. 
  Our method can be useful for re-calibrating the uncertainty quantification 
  of such seemingly promising units. 
\end{itemize}

Finally, we demonstrate the application of our methods 
via several realistic selection rules that may occur 
in drug discovery (Section~\ref{sec:drug}) 
and health risk prediction (Section~\ref{sec:icu}). 
Our results show that marginal prediction sets may undercover or overcover the selected units, 
while our methods always achieve the promised coverage guarantees. 
 
Due to space constraints, we delegate a comprehensive literature review to Supplementary 
Section~\ref{app:overview}. We end this section with some useful notations.

\paragraph{Notations.} 
For a positive integer $n\in \mathbb{N}^+$, write $[n]:=\{1,\dots,n\}$.    
We write the data pair as $Z=(X,Y)$, so that  
the calibration data is 
$\cD_\calib = \{Z_i\}_{i=1}^n$.
For any $j\in[m]$, 
we define the augmented calibration set as $\cD_j = \cdc \cup \{Z_{n+j}\}$ and 
the remaining test set as $\cD^c_j = \cdte \backslash 
\{X_{n+j}\}$. 
The unordered set of $\cD_j$ is denoted  $[\cD_j] = [Z_1,\dots,Z_n,Z_{n+j}]$, 
which provides the order statistics. 
 
\section{Problem setup}


Following the split conformal prediction 
framework~\citep{vovk2005algorithmic,lei2018distribution}, 
we build our prediction sets based on a prediction model 
fitted on a training set $\cD_\train$,  
assuming $\cD_\train$ is independent of the calibration and test data. 
In what follows, we always 
condition on $\cD_\train$, thereby treating the fitted models as fixed. 
In this section, we formally introduce the selection-conditional coverage 
guarantees  and compare them to other related notions in the literature.

\subsection{Selection-conditional coverage}

Recall that $j\in[m]$  is 
a focal unit if $j\in \hat\cS$, where 
$\hat\cS = \cS(\cD_\calib,\cD_\test) \subseteq [m]$ is obtained from a 
selection rule $\cS$ that depends on the observed data. 
Its potential dependence on $\cD_\train$ is clear 
since we treat $\cD_\train$ as fixed. 
Without loss of generality, we posit that $\cS$ is a deterministic function; 
when the selection rule is randomized, one can condition on its randomness 
and follow our framework. 

For each test unit $j$, we wish to 
construct a prediction set $\hC_{\alpha,n+j}\subseteq \cY$ 
such that  
\@\label{eq:def_cond_cov_set}
\PP\big(Y_{n+j} \in \hC_{\alpha,n+j} \biggiven 
j\in \hat\cS,~\hat\cS \in \mathfrak{S}\big) \geq 1-\alpha,
\@
where $\alpha\in(0,1)$ is the confidence level,
and  $\mathfrak{S} \subseteq 2^{[m]}$ is some pre-specified collection of subsets 
of $[m]$. We call $\mfS$ the {\em selection taxonomy} in what follows. 
Different choices of $\mfS$ lead to a spectrum of granularity in 
the conditioning event.  
For instance, taking $\mfS = 2^{[m]}$ puts no restrictions on the selection set, 
giving rise to the guarantee~\eqref{eq:def_sel_coverage} 
introduced in the beginning (coverage conditional on a unit being selected). 
Taking $\mathfrak{S} = \{S\subseteq [m] \colon |S|=r\}$ for some $r\leq m$ 
achieves coverage conditional on selecting a specific number of units. 
Finally, 
taking $\mathfrak{S} = \{S_0\}$ for some $S_0\subseteq [m]$ 
achieves coverage conditional on selecting a specific set; 
this is similar to the coverage guarantee conditional on a selected model in~\citet{lee2016exact}
for high-dimensional parameter inference.  

In words, the selection-conditional coverage guarantee~\eqref{eq:def_cond_cov_set} 
can be interpreted as follows: 
imagine there are infinitely many independent realizations of $\cD_\calib$ and $\cD_\test$; 
among those realizations where the selection event of interest happens, the prediction set 
$\hC_{\alpha,n+j}$ will cover the true outcome for at least $1-\alpha$ fraction of times.
Here, $\mathfrak{S}$ is introduced to provide practitioners with the 
language to specify the granularity of the conditioning event, allowing 
them to tailor the guarantee to their specific needs. Moreover, it 
allows us to properly describe the relationship between selection-conditional coverage
and FCR control, as we will see below.

\subsection{Relations among notions of selective coverage}
\label{subsec:compare}

Before introducing our methods, 
we take a moment to compare different notions of selective coverage. 
Readers who are more interested in our methodology 
may skip the remaining of this section. 

Our first observation is that~\eqref{eq:def_cond_cov_set} implies~\eqref{eq:def_sel_coverage} under appropriate conditions. 
The proof of the next proposition is in Supplementary Section~\ref{app:strong_weak}. 


\begin{proposition}\label{prop:strong_weak}
Suppose a family of prediction sets $\{\hC_{\alpha,n+j}^{(\ell)}\}_{\ell\in \cL}$  
satisfy $\PP (Y_{n+j} \in \hC_{\alpha,n+j}^{(\ell)} \given 
j\in \hat\cS,~\hat\cS \in \mathfrak{S}_\ell ) \geq 1-\alpha$
for a set of  disjoint 
taxonomies $\{\mfS_\ell\}_{\ell\in\cL}$ such that $\cup_{\ell\in \cL}\mfS_\ell = 2^{[m]}$.
Define the prediction set $\hat{C}_{\alpha,n+j} = \hat{C}_{\alpha,n+j}^{(\ell)}$ when 
$j\in \hat\cS$ and $\hat\cS\in \mfS_\ell$. Then $\PP(Y_{n+j} \in \hC_{\alpha,n+j}  \given 
j\in \hat\cS ) \geq 1-\alpha$.
\end{proposition}

To distinguish the selection-conditional coverage in~\eqref{eq:def_cond_cov_set}
and that in~\eqref{eq:def_sel_coverage}, we refer to~\eqref{eq:def_cond_cov_set} as 
{\em strong} selection-conditional coverage and~\eqref{eq:def_sel_coverage} as {\em weak} selection-conditional coverage.

Another widely used post-selection guarantee is the 
false coverage rate (FCR)~\citep{benjamini2005false}, 
defined as the expected proportion of selected units missed by the
prediction set:
\@ \label{eq:def_fcr}
\text{FCR} \,:= \, \EE\bigg[\frac{\sum_{j\in [m]}\ind\{j\in \hcS, Y_{n+j} 
\notin \hC_{\alpha,n+j}\}}{|\hcS| \vee 1}\bigg],  
\@
where $a \vee b = \max\{a,b\}$ for any $a,b\in \RR$. 
Previous works~\citep{weinstein2020online,bao2024selective,gazin2024selecting}
mainly focus on constructing prediction sets with FCR control.
However, with $m=1$, it holds that 
$\text{FCR} = \PP(j\in \hat{\cS},Y_{n+j}\notin \hat{C}_\alpha(X_{n+j}))\leq \alpha$ 
for any marginally valid prediction set. 
Thus, FCR does not always address the selection issue.

The following proposition shows that strong selection-conditional coverage 
implies FCR control 
with a proper choice of selection taxonomy, whose
proof is in Supplementary Section~\ref{appd:proof_prop_notions}.

\begin{proposition}
\label{prop:notions}
Suppose a family of prediction sets $\{\hC_{\alpha,n+j}^{(\ell)}\}_{\ell\in \cL}$ satisfy $\PP (Y_{n+j} \in \hC_{\alpha,n+j}^{(\ell)} \given 
j\in \hat\cS,~\hat\cS \in \mathfrak{S}_\ell ) \geq 1-\alpha$
for a set of disjoint selection 
taxonomies $\{\mfS_\ell\}_{\ell\in\cL}$ such that $\cup_{\ell\in \cL}\mfS_\ell = 2^{[m]}$, 
and for each $\ell\in \cL$,  
$\mfS_\ell \subseteq \{S\subseteq[m]\colon |S|=r(\ell)\}$ for some $0\leq r(\ell) \leq m$. 
Define the prediction set $\hat{C}_{\alpha,n+j} = \hat{C}_{\alpha,n+j}^{(\ell)}$ when $j\in \hat{\cS}$ and $\hat\cS\in \mfS_\ell$. 
Then its FCR~\eqref{eq:def_fcr} is upper bounded by 
$\alpha \cdot \PP(\hat{\cS}\neq \varnothing)\leq \alpha$.
\end{proposition} 


The selection taxonomies in Proposition~\ref{prop:notions} require the selection set to be of a specific size, without other conditions on its form. This can be automatically satisfied by some selection rules such as top-K selection. 
Extending Proposition~\ref{prop:notions} to non-exact scenarios yields useful insights in practice. For instance, consider the family of taxonomies with $\mfS_\ell\subseteq \{S\subseteq[m]\colon r(\ell)-\delta \le |S| \leq \delta + r(\ell)\}$ for some small $\delta\in \mathbb{N}^+$. When the selection size is stable, weak selection-conditional coverage implies  $\PP (Y_{n+j} \in \hC_{\alpha,n+j}  \given 
j\in \hat\cS,~\hat\cS \in \mathfrak{S}_\ell ) \approx 1-\alpha$, which leads to  $\textrm{FCR} \approx \alpha\cdot \PP(\hat{\cS}\neq \varnothing)$. This approximate result helps explain several phenomena in our numerical experiments: first, the FCR is usually controlled empirically by our method; second, we sometimes observe significantly lower FCR than the selection-conditional coverage when $\PP(\hat{\cS} = \varnothing)$ is moderately large. 
As we shall see later, a version of our proposed method achieves coverage guarantees conditional on being selected and the selection size,
which requires checking a slightly more complex condition in the reference set construction when the selection size varies.
The computation of such prediction sets can be done efficiently for {\em all} the instances provided 
in the paper, with the corresponding worst-case computational complexity explicitly stated in each section. 
One potential concern, though, is that the additional condition on the selection size may reduce 
the number of calibration data points in the reference set, especially when the selection step is highly variable. 
This could potentially lead to wider and/or unstable prediction intervals.

As a side note,  
the weak selection-conditional coverage does not necessarily imply FCR control, 
although in some special cases both can be true (see, e.g., \cite{bao2024selective} 
for such examples).  
We put this as a proposition below, with a counterexample given 
in Supplementary Section~\ref{app:proof_weak_FCR}.

\begin{proposition}\label{prop:weak_FCR}
There exists an instance 
and prediction sets $\{\hC_{\alpha,n+j}:j \in \hcS\}$ that satisfy the weak selection-conditional coverage 
at level $\alpha$ but violate the FCR control at level $\alpha$.
\end{proposition}

We end this section with a remark on the interpretation of selection-conditional coverage and FCR control.  
\begin{remark}[Interpretation of selection-conditional coverage and FCR control]
As shown by Proposition~\ref{prop:notions}, the guarantee in~\eqref{eq:def_cond_cov_set} 
is stronger than FCR control (for a proper choice of the selection taxonomy).
In fact, the former has often been used as a device 
to derive the latter in the literature~\citep[e.g.,][]{weinstein2013selection,bao2024selective}. 

In terms of interpretation, there are two main differences between selection-conditional 
coverage and FCR. 
First, the guarantee of FCR is averaged over all the selection events, 
including the case of empty selection set where the false coverage proportion is by definition  
zero. Therefore, even with FCR control, one could still suffer from a high false 
coverage proportion when the selection set is non-empty as long as this is compensated 
by the cases where the selection set is empty. On the other hand, selection-conditional 
coverage delivers guarantees conditioning on selection events of interest, which prevents 
the aforementioned undesired situation.
Second, FCR control provides a guarantee that is averaged over 
all the selected units: it could be possible that for some selected units, the coverage
is much lower than the nominal level, and for others the coverage is much higher, 
so that the average coverage over all the selected units is controlled at the nominal level. 
In contrast, selection-conditional coverage provides guarantees specific to
each selected units.
        
Finally, we note that the distinction between the two concepts can be  asymptotically 
negligible in special cases. 
We refer to Lemma~\ref{lem:scc_fcr} in the supplementary material for such an instance.
\end{remark}
\section{JOMI: a unified framework}
\label{sec:unified}

\subsection{Warm-up: split conformal prediction}
\label{subsec:SCP}
To warm up, we briefly summarize the split conformal prediction (SCP) 
method~\citep{vovk2005algorithmic,lei2018distribution} and 
how it achieves finite-sample coverage under exchangeability.

SCP starts with a nonconformity score function
$V\colon \cX \times \cY \mapsto \RR$ determined by $\cdt$,  
so that $V(x,y)$ informs how well a hypothetical value $y\in\cY$ conforms 
to a machine prediction. For instance, one may set 
$V(x,y) = |y - \hat{\mu}(x)|$, where $\hat{\mu}(x)$ is regression function 
fitted on $\cdt$. 
Other popular choices  
include {\em conformalized quantile regression}~\citep[CQR,][]{romano2019conformalized} 
for regression and {\em adaptive prediction sets}~\citep[APS,][]{romano2020classification} 
for classification.

Compute  
$V_i =V(X_{i},Y_i)$ for $i \in[n]$. 
The split conformal prediction set for test unit $j\in[m]$ is 
\#\label{eq:vanilla_conformal}
\hat{C}_{\alpha, n+j}^{\text{SCP}} = 
\Big\{ y\colon V(X_{n+j},y)\leq \text{Quantile}\big(1-\alpha; \{V_i\}_{i=1}^n \cup\{+\infty\}\big) \Big\},
\#
where $\text{Quantile}(1-\alpha;\cdot)$ is the  $(1-\alpha)$-th empirical
quantile of the set in the second argument. 
When $(Z_1,\dots,Z_n,Z_{n+j})$ are exchangeable, $\hat{C}_{\alpha, n+j}^{\text{SCP}}$
achieves~\eqref{eq:def_mgn_coverage}~\citep{vovk2005algorithmic}.

It helps motivate our approach  
to see 
the ideas behind the validity of $\hat{C}_{\alpha, n+j}^{\text{SCP}}$. 
In words, $\hat{C}_{\alpha,n+j}^{\text{SCP}}$ 
finds hypothesized values of $y$ that make $V(X_{n+j},y)$ look 
similar to calibration scores. 
Note that
\#\label{eq:proof_vanilla_CP}
\PP\big(Y_{n+j}\in \hat{C}_{\alpha, n+j}^{\text{SCP}}\big)
= \PP\Big( V_{n+j}\leq \text{Quantile}\big(1-\alpha; \{V_i\}_{i=1}^n \cup\{V_{n+j}\}\big)  \Big),
\#
where $V_{n+j}=V(X_{n+j},Y_{n+j})$. 
Recall the unordered set  
 $[\cD_j] = [Z_1,\dots,Z_n,Z_{n+j}]$ where $Z_i=(X_i,Y_i)$. 
Conditional on the event $\{[\cD_j] = [z_1,\dots,z_{n+j}]\}$, 
the only randomness is in the ordering of $(Z_1,\dots,Z_n,Z_{n+j})$; 
due to exchangeability, 
the probability of $Z_{n+j}$ taking on each value in $z_1,\dots,z_n,z_{n+j}$ is equal. 
Therefore, conditional on $[\cD_j]= [z_1,\dots,z_{n+j}]$, 
the chance of $V_{n+j}$ being no greater than the $(1-\alpha)$-th quantile 
of $[v_1,\dots,v_n,v_{n+j}]$, where $v_i = v(z_i)$, is 
at least $1-\alpha$, i.e.,  
\#\label{eq:proof_vanilla_CP_2}
\PP\Big( V_{n+j} \leq \text{Quantile}\big(1-\alpha; \{v_i\}_{i=1}^n \cup\{v_{n+j}\}\big) \Biggiven [\cD_j] = [z_1,\dots,z_{n+j}] \Big) \geq 1-\alpha.
\# 
This leads to~\eqref{eq:proof_vanilla_CP} by the tower property. 
Therefore, inverting the criterion on the right-hand side of~\eqref{eq:proof_vanilla_CP} 
gives a valid prediction set for $Y_{n+j}$. 

The reason why vanilla SCP may fail to achieve selection-conditional
coverage like~\eqref{eq:def_cond_cov_set} is that, 
conditional on the selection event, 
the data points $\{Z_1,\dots,Z_n,Z_{n+j}\}$ are no longer exchangeable.  
In other words, in~\eqref{eq:proof_vanilla_CP_2}, 
it is unclear how  $V_{n+j}$ is distributed over $v_1,\dots,v_{n+j}$ 
if we additionally condition on the selection event. 
As such, a natural remedy is to  
find a subset of calibration data 
that are ``exchangeable'' with respect to the test point 
conditioning on the selection event, and leverage 
their scores to calibrate 
the prediction of the test unit. 
We introduce our methods below.  

\subsection{Conformal inference via a reference set}

Fix a test unit $j\in[m]$. 
Recall that $Z_{n+j}(y) = (X_{n+j},y)$ is the imputed test point
with a hypothesized response $y\in \cY$. 
The core of our method is to find calibration points  
that are exchangeable with respect to the test point 
conditional on the selection set. 
At a high level, 
they are calibration units $i\in [n]$ that are ``indistinguishable'' with $j\in [m]$, 
in the sense that treating $Z_{n+j}(y)$ as  a 
calibration point and $Z_i$ as a test point  results in the same selection event. 

We formalize this idea via a ``swap'' operation. 
For any calibration unit $i\in \cic$,  
we define the ``swapped'' calibration data  $\cD_{\calib}^{\swap{i}{j}}(y)$ 
and the swapped test data $\cD_{\test}^{\swap{i}{j}}$ as follows: 
\$
& \cD_{\calib}^{\swap{i}{j}}(y) = (Z_1^{\swap{i}{j}}(y), Z_2^{\swap{i}{j}}(y), \dots, Z_n^{\swap{i}{j}}(y)), \\ 
& \cD_{\test}^{\swap{i}{j}} = (X_{n+1}^{\swap{i}{j}}, X_{n+2}^{\swap{i}{j}}, \dots, X_{n+m}^{\swap{i}{j}}),
\$
where for $k \in [n]$ and $\ell \in[m]$,  
\$
Z_k^{\swap{i}{j}}(y) = 
\begin{cases}
Z_{n+j}(y) & k = i, \\ 
Z_k & k \neq i .
\end{cases}
\quad 
X_{n+\ell}^{\swap{i}{j}} =
\begin{cases}
X_i & \ell = j, \\
X_{n+\ell} & \ell \neq j.
\end{cases}
\$ 
That is, $\cdc^\swap{i}{j}(y)$ and $\cdte^\swap{i}{j}$ 
are the calibration and test data 
if we treat $Z_i$ as the $j$-th test point, 
and $Z_{n+j}(y)$ as the $i$-th calibration point. 
Figure~\ref{fig:swap} is an illustration of the swap operation.

\begin{figure}[htbp]
    \scalebox{0.9}{%
    \begin{tikzpicture}
        \coordinate (p);
        \foreach \n/\w/\c in {{\footnotesize$Z_1$}/0.8/blue!10,
        {\footnotesize$Z_2$}/0.8/blue!10,
        $\cdots$/0.8/blue!10,
        {\footnotesize$Z_{i}$}/0.8/red!10,
        $\cdots$/0.8/blue!10,
        {\footnotesize$Z_n$}/0.8/blue!10}{
        \node[draw,minimum height=0.8cm,minimum width=\w cm,anchor=west,outer sep=0pt, fill = \c]
        (n) at (p) {\n};
        \coordinate (p) at (n.east);
        }
        \node at (-0.8,0) {$\cdc = $};
        \node at (7.1,0) {$\cdte = $};
    
        \coordinate (p) at (7.8,0);
        \foreach \n/\w/\c in {{\footnotesize$X_{n+1}$}/0.8/blue!10,
        {\footnotesize$X_{n+2}$}/0.8/blue!10,
        {\footnotesize$\cdots$}/0.8/blue!10,
        {\footnotesize$X_{n+j}$}/0.8/red!10,
        $\cdots$/0.8/blue!10,
        {\footnotesize$X_{n+m}$}/0.8/blue!10}{
        \node[draw,minimum height=0.8cm,minimum width=\w cm,anchor=west,outer sep=0pt, fill = \c]
        (n) at (p) {\n};
        \coordinate (p) at (n.east);
        }
        
        \coordinate (p) at (0,-1.5);
        \foreach \n/\w/\c in {{\footnotesize$Z_1$}/0.8/blue!10,
        {\footnotesize$Z_2$}/0.8/blue!10,
        $\cdots$/0.8/blue!10,
        {\footnotesize$Z_{n+j}(y)$}/0.8/red!10,
        $\cdots$/0.8/blue!10,
        {\footnotesize$Z_n$}/0.8/blue!10}{
        \node[draw,minimum height=0.8cm,minimum width=\w cm,anchor=west,outer sep=0pt, fill = \c]
        (n) at (p) {\n};
        \coordinate (p) at (n.east);
        }
        \node at (-1.3,-1.5) {$\cdc^{\swap{i}{j}}(y)=$};
    
        \coordinate (p) at (7.8,-1.5);
        \foreach \n/\w/\c in {{\footnotesize$X_{n+1}$}/0.8/blue!10,
        {\footnotesize$X_{n+2}$}/0.8/blue!10,
        {\footnotesize$\cdots$}/0.8/blue!10,
        {\footnotesize$X_{i}$}/0.8/red!10,
        $\cdots$/0.8/blue!10,
        {\footnotesize$X_{n+m}$}/0.8/blue!10}{
        \node[draw,minimum height=0.8cm,minimum width=\w cm,anchor=west,outer sep=0pt, fill = \c]
        (n) at (p) {\n};
        \coordinate (p) at (n.east);
        }
    \node at (6.7,-1.5) {$\cdte^{\swap{i}{j}}=$};
      \end{tikzpicture} 
    }
      \caption{Graphical illustration of $\cdc^\swap{i}{j}(y)$ and $\cdte^{\swap{i}{j}}$.}
      \label{fig:swap}
    \end{figure}

Applying the same selection rule $\cS$ to the swapped data, 
we define the swapped  selection set  
with the hypothesized $y$ as
\$
\hat\cS^{\swap{i}{j}}(y) =  
\cS \big(\cD_{\calib}^{\swap{i}{j}}(y),\cD_{\test}^{\swap{i}{j}} \big).
\$  
Then, we define the ``reference set'' for achieving~\eqref{eq:def_cond_cov_set}
with taxonomy $\mfS$ as  
\@\label{eq:def_hat_Rj}
\hcR_{n+j}(y) = \big\{  i \in [n] \colon j \in \hat\cS^{\swap{i}{j}}(y), 
\text{ and } \hcS^{\swap{i}{j}}(y) \in \mfS \big\}.
\@
In words, the reference set is the collection of calibration points $i\in[n]$ such that, 
after swapping unit $i$ and $n+j$, 
the (posited) test point $j$ (which is the original unit $i$)  remains in the focal set and the 
focal set remains in $\mfS$. 
We use the notation $\hcR_{n+j}(y)$ to emphasize that $\hat\cR_{n+j}(\cdot)$ is a 
data-dependent mapping from $\cY$ to the power set of $[n]$. 

With all the preparation, we define our prediction set for $Y_{n+j}$ as
\@\label{eq:PI_cond}
\hat{C}_{\alpha,n+j} = \Big\{ y \in \cY \colon V(X_{n+j},y) \leq \quant 
\big(1-\alpha; \{V_i\}_{i\in \hat\cR_{n+j}(y)} \cup \{+\infty\} \big)  \Big\}.
\@ 

As we shall show shortly, 
our prediction set~\eqref{eq:PI_cond} achieves near-exact 
coverage when the nonconformity score is continuous and the reference set is of moderate 
size. 
In addition, we can achieve exact coverage by introducing extra randomness: 
\@\label{eq:PI_cond_rand}
&\hat{C}_{\alpha,n+j}^\rand \nonumber\\ 
=& \bigg\{  y\colon 
\frac{\sum_{i \in \hat\cR_{n+j}(y)} \ind\{V(X_{n+j},y) < V_i\} + 
U_j \cdot (1+  \sum_{i\in \hat\cR_{n+j}(y)}
\ind\{V(X_{n+j},y) = V_i\}) }{1+ | \hat\cR_{n+j}(y)|}
\leq 1-\alpha \bigg\},
\@
where $U_1,\ldots,U_m$ are i.i.d.~random variables drawn from 
$\textrm{Unif}[0,1]$  independent of the data. 
In Supplementary Sections~\ref{app:MCP} 
and~\ref{app:BY}, we discuss in detail the connection of our method to Mondrian conformal prediction (MCP)
and comparison with the BY procedure~\citep{benjamini2005false}. 
While the BY procedure is a natural and heuristic solution to post-selection inference, it suffers from over-conservativeness and non-adaptivity to the selection event.

\subsection{Theoretical guarantees}
Theorem~\ref{thm:PI_cond_set} confirms the conditional validity of our prediction 
sets $\hC_{\alpha,n+j}$ in~\eqref{eq:PI_cond} and 
$\hC^\rand_{\alpha,n+j}$ in~\eqref{eq:PI_cond_rand}. 

\begin{theorem}
    \label{thm:PI_cond_set}
    Suppose $\cS(\cD_\calib,\cD_\test)$ is invariant to permutations  of $\cD_\calib$, 
    and that $\{Z_i\}_{i=1}^n\cup\{Z_{n+j}\}$ are exchangeable conditional on 
    $\{X_{n+\ell}\}_{\ell\in[m]\backslash \{j\}}$ for any $j\in[m]$. 
    Then, for any selection taxonomy $\mfS$, 
    the following statements hold.
    \begin{enumerate} 
    \item[(a) ] $\hat{C}_{\alpha,n+j}$ defined in~\eqref{eq:PI_cond} obeys 
    \@\label{eq:def_cond_cov_set_again}
    \PP\big( Y_{n+j} \in \hat{C}_{\alpha, n+j}
    \biggiven j\in \hat\cS, ~\hat\cS\in \mathfrak{S}\big) \geq 1-\alpha. 
    \@
    Furthermore, if ties among $V_1,\ldots,V_n,V_{n+j}$ occur with probability 
    zero, then 
    \@\label{eq:def_cond_cov_set_upper}
    \PP\big( Y_{n+j} \in \hat{C}_{\alpha, n+j} 
    \biggiven j\in \hat\cS, ~\hat\cS\in \mathfrak{S}\big) \le 1-\alpha
    + \EE\bigg[\frac{1}{1+|\hcR_{n+j}(Y_{n+j})|}\bigg]. 
    \@
    \item[(b) ] The randomized prediction set $\hat{C}_{\alpha,n+j}^\rand$ 
    defined in~\eqref{eq:PI_cond_rand}
    satisfies 
    \$
    \PP(Y_{n+j}\in \hat{C}_{\alpha,n+j}^\rand\given j\in \hat\cS,~ \hcS \in \mfS) = 1-\alpha. 
    \$
    \end{enumerate}
\end{theorem}

We defer the detailed proof of Theorem~\ref{thm:PI_cond_set} 
to Supplementary Section~\ref{appd:proof_PI_cond_set} and
provide some intuition here. 
Similar to the ideas of SCP in Section~\ref{subsec:SCP}, 
the key fact we rely on is that, 
plugging in the true value $y=Y_{n+j}$, 
data in the reference set and the new test point 
are still exchangeable conditional on the selection event. 
To be specific, to prove~\eqref{eq:def_cond_cov_set_again}, we are to show a stronger result:
\#\label{eq:intuition_1}
\PP\Big(  V_{n+j} \leq \quant 
    \big(1-\alpha; \{V_i\}_{i\in \hcR_{n+j}(Y_{n+j})} \cup \{V_{n+j}\} \big) 
    \Biggiven j\in \hat\cS, \, \hcS \in \mfS, \, [\cD_j],\, \cD_j^c 
    \Big)\geq 1-\alpha,
\#
where we recall $[\cD_j]= [Z_1,\dots,Z_n,Z_{n+j}]$ and 
$\cD_j^c = \cdte\backslash\{X_{n+j}\}$.
For any fixed values $z_1,\dots,z_n,z_{n+j}$, 
once given the unordered values $[\cD_j] =[d_j] = [z_1,\dots,z_n,z_{n+j}]$  and 
the values of other test points $\{X_{n+\ell}\}_{\ell\neq j}$,
the only randomness is in the ordering of $Z_1,\dots,Z_{n+j}$ 
among $[z_1,\dots,z_n,z_{n+j}]$. 
Meanwhile, we show that our reference set $\hat\cR_{n+j}(\cdot)$ is 
constructed in a delicate way such that the 
(unordered set of) scores 
$\{V_i\}_{i\in \hcR_{n+j}(Y_{n+j})} \cup \{V_{n+j}\}$ 
is fully determined by  $[z_1,\dots,z_n,z_{n+j}]$. 
That is,  $\hat{R}_{n+j}^+:= [V_i\colon i \in \hat\cR_{n+j}(Y_{n+j}) \cup\{n+j\}]$ 
is fully determined given $[d_j]$. 
Then, \eqref{eq:intuition_1} reduces to 
\#
\PP\Big(  V_{n+j} \leq \quant 
    \big(1-\alpha; \hat{R}_{n+j}^+ \big) 
    \Biggiven j\in \hat\cS, \, \hcS \in \mfS, \, [\cD_j],\, \cD_j^c 
    \Big)\geq 1-\alpha,
\#
Finally, by the exchangeability of $Z_1,\dots,Z_n,Z_{n+j}$,  
the probability of $V_{n+j}$ taking on any value in $\hat{R}_{n+j}^+$ is 
equal given $[\cD_j]$ and $\cD_j^c$, leading to the validity in Theorem~\ref{thm:PI_cond_set} 
via Bayes' rule.

\section{Computationally tractable instances}
\label{sec:compute}

So far, we have presented a general framework for constructing 
valid prediction sets conditional on general selection events. 
However, computing the prediction sets 
according to their definition 
requires looping over all possible values of $y\in \cY$, 
which can be computationally intractable. 
When $|\cY|$ is finite (and relatively small), our proposed method 
can be efficiently implemented according to its definition; the corresponding 
computational complexity is at most $O(|\cY|mn)$ times 
the complexity of the selection rule.

In this section, we instantiate our general procedure beyond the small $|\cY|$ setting 
with concrete examples where special 
structures enable efficient computation. 
We focus on three classes of selection rules that can be of practical interest: 
selection using only the covariates, 
selection based on conformal p-values, and selection based on conformal prediction sets.
When practitioners are willing to slightly modify the selection rule to improve computation efficiency, they may refer to Supplementary Section~\ref{app:splitting} where we discuss extensions to simplify the construction of prediction sets 
by further splitting the calibration data.

\subsection{Covariate-dependent selection rules}
\label{sec:covariate}

We first consider covariate-dependent selection rules, i.e., 
when $\cS(\cD_\calib,\cD_\test)$ is only a function of $\{X_{i}\}_{i=1}^{n+m}$.  
Under such rules, the reference set no longer 
depends on $y$; we shall suppress the dependence on $y$ 
and write $\hcR_{n+j}(y)\equiv \hcR_{n+j}$ throughout this subsection.  

Here, $\hcR_{n+j}$  
can be efficiently computed by looping over $i\in [n]$. 
The complete procedure for constructing $\hC_{\alpha,n+j}$  
and $\hC^\rand_{\alpha,n+j}$  with an arbitrary covariate-dependent selection rule 
is summarized in Algorithm~\ref{alg:PI_cond_cov}. 
Its overall computational complexity is $O(mn\cdot C_\cS)$, 
with $C_\cS$ being the complexity of the selection process. 

\begin{algorithm}[htbp]
\caption{JOMI for arbitrary covariate-dependent selection rules}
\label{alg:PI_cond_cov}
\begin{algorithmic}
\Require Calibration data $\cdc$; test data $\cD_{\test}$; miscoverage level $\alpha$;
nonconformity score $V(\cdot, \cdot)$;
selection rule $\cS$; 
selection taxonomy $\mfS$;
form of prediction set $\in$ \{\texttt{dtm}, \texttt{rand}\}.
\State Compute $\hcS = \cS(\cdt,\cdte)$.
\For{$j \in \hcS$}

\State Initialize $\hcR_{n+j} = \emptyset$.
\For{$i = 1,\ldots,n$}
    
 \State   $\hcR_{n+j} = \hcR_{n+j} \cup \{i\}$ ~ if $j \in \hcS^{\swap{i}{j}}$ and 
    $\hcS^{\swap{i}{j}} \in \mfS$.

\EndFor

\If{\textnormal{form = \texttt{dtm}}}
    \State $\hC_{\alpha,n+j} = 
     \Big\{ y \in \cY \colon V(X_{n+j},y) \leq \quant \big(1-\alpha; 
     \{V_i\}_{i\in \hat\cR_{n+j}} \cup \{+\infty\} \big)  \Big\}.
    $
\EndIf
\If{\textnormal{form = \texttt{rand}}}

\State Sample $U_{j} \sim \textnormal{Unif}[0,1]$.
\State $\hC_{\alpha,n+j} \leftarrow \Big\{  y\colon 
\frac{\sum_{i \in \hat\cR_{n+j}} \ind\{V(X_{n+j},y) < V_i\} + U_{j} \cdot (1+  \sum_{i\in \hat\cR_{n+j}} 
\ind\{V(X_{n+j},y) = V_i\}) }{1+ | \hat\cR_{n+j}|}
\leq 1-\alpha \Big\}.$
\EndIf
\EndFor

\Ensure {$\{\hC_{\alpha,n+j}\}_{j\in \hcS}$.}
\end{algorithmic}
\end{algorithm}

By Theorem~\ref{alg:PI_cond_cov}, 
the output of Algorithm~\ref{alg:PI_cond_cov} is valid as long as 
the selection rule does not rely on the ordering   
of the calibration points. This includes many commonly-used selection rules: 
\begin{enumerate}
    \item[(1)] {\em Top-K selection.}  
    The $K$ test units with the highest scores $S(X_i)$ are selected, 
    where $\cS\colon \cX\to \RR$ is a pre-trained score function. 
    For example, $S(X_i)$ may be the predicted  
    binding affinity for a drug with chemical structure $X_i$ or 
    the predicted health risk of a patient with features $X_i$, 
    and the drug discovery process or clinical system admits a fixed number of new units. 
    \item[(2)] {\em Selection based on joint quantiles.}  
     A unit $j$ is selected if its score $S(X_{n+j})$ surpasses 
     the $q$-th quantile of both 
     the calibration {\em and}  test scores $\{S(X_i)\}_{i=1}^{n+m}$. 
     For instance, a scientist may be interested in toxicities of 
     drug candidates in $\cD_\test$ where those in $\cD_\calib$ have been tested, 
     but they only focus on drugs 
     with highest predicted activities $S(X_{i})$ in the entire library.
    \item[(3)] {\em Selection based on calibration quantiles.} A test unit $j$ is 
    selected if its score $S(X_{n+j})$ surpasses the 
    $q$-th quantile of calibration scores $\{S(X_i)\}_{i=1}^n$. 
    This may happen when a doctor uses the 
    predicted health risks of existing patients to determine a normal range,  
    and picks test units anticipated to have a relatively extreme health risk.  
    \item[(4)] {\em Selection by black-box optimization procedures.} 
    A test unit $j$ is selected by running an arbitrary (even black-box) optimization program that 
    does not involve $\{Y_i\}_{i\in \cI_\calib}$. 
    One such example is optimization under constraints: apart from the score $S(X_{n+j})$, each unit $j$ is associated with a 
    cost $C_{n+j}$, and $\hcS$ is the subset of test units 
    that maximizes $\sum_{j \in \hcS} S(X_{n+j})$ subject to $\sum_{j\in \hcS} C_{n+j} \leq c$, 
    where $c$ reflects the total budget. This may happen in healthcare management systems 
    that optimize resources subject to constraints and send patients to different care categories, or 
    in candidate screening for job interviews subject to budget and diversity constraints. 
    More generally, in drug discovery, scientists may run a complex Bayesian optimization algorithm to select the next batch of drugs to evaluate~\citep{pyzer2018bayesian},
    which may be viewed as a black box process. No matter how complicated these optimization programs are, as long as they do not involve calibration labels, JOMI supports efficient uncertainty quantification afterwards. We will see some stylized examples in our numerical experiments.
\end{enumerate}

We additionally show in Supplementary Section~\ref{app:cov_dep} that 
we can further improve the computational efficiency of JOMI for rules (1)-(3) 
by deriving exact forms of the reference sets, where in each case the computation complexity is $O(\max\{m,n\})$. 
We also note that rules (1)-(2) were considered in~\citet{bao2024selective}, 
where they propose methods that achieve~\eqref{eq:def_sel_coverage} and FCR control. 
The prediction sets proposed therein  
coincide with ours, and 
our results imply that they in fact achieve 
the strong selection-conditional coverage in~\eqref{eq:def_cond_cov_set} 
for free. 

\subsection{Selection based on conformal p-values}
\label{sec:conf_pval}

The second class of selection rules we study 
concern  selecting units 
whose outcomes satisfy certain conditions 
while controlling some type-I error. 
To this end, the test units are selected by thresholding 
a class of conformal p-values, 
where each p-value is computed via contrasting 
a test point with the calibration data. 
As such, the selection rule can be  complicated and asymmetric. 

We will follow the framework of~\citet{jin2022selection} 
to define the p-values and selection rules, who 
study the problem of discovering test units with large outcomes. 
Examples include selecting drugs with sufficiently high binding affinities, 
finding highly competent job candidates, and identifying patients who benefit from a treatment, etc. 
In these problems, 
predictions from machine learning models 
serve as proxies for the true outcomes of interest 
that are too expensive or impossible to evaluate, 
and the selection procedure leverages the power of 
predictions to select units with large outcomes while ensuring error control. 

Given test points $\{X_{n+j}\}_{j\in[m]}$ 
and (potentially random) thresholds $c_{n+j}\in\RR$, 
the goal is to select those $Y_{n+j}> c_{n+j}$ while controlling 
the number/fraction of
false positives. 
%
%
The statistical evidence for 
detecting a large outcome is quantified by the so-called
``conformal p-values''. 

Suppose the calibration data are $\{(X_i,Y_i,c_i)\}_{i=1}^n$, 
such that the tuples $\{X_{i},Y_i,c_i\}_{i=1}^{n+m}$ are exchangeable.  
Assume access to a score function $S\colon \cX\times\RR\to \RR$
such that  $S(x,y)$ is 
non-increasing in $y$ for any $x\in \cX$. An example is 
$S(x,y) = \hat{\mu}(x) - y$ where $\hat\mu(x)$ is a 
point predictor trained on $\cD_\train$. 
We then compute $\hat{S}_i= S(X_i,c_i)$ for $i\in[n+m]$, and 
define the conformal p-values\footnote{
Our p-value slightly modify the definition
in~\citet{jin2022selection} for the ease of describing our prediction sets. 
Similar to the original ones, our p-values control the type-I error 
in detecting one large outcome. 
In addition, using our p-values 
in their original procedures maintains error control with improved power; 
we discuss these results 
in Supplementary Section~\ref{subsec:conf_pval} for completeness. 
}
\@\label{eq:iid_conf_pval}
& p_j = \frac{1 + \sum_{i \in \cic}\ind\{\hS_i \ge \hS_{n+j}, Y_i \le c_i\}}
{n+1},\quad j\in [m].
\@
Hereafter, we call $S$ the \emph{selection} score function. 
We can show that $p_j$ is valid in the sense that   
$
\PP(p_j\leq t,~ Y_{n+j}\leq c_{n+j}) \leq t, 
$
$\forall t \in [0,1]$.
That is, testing with $p_j$ controls the type-I error in 
finding one large outcome, accounting for the randomness in the outcomes as well. 
We then select test units whose conformal p-values are below a threshold,
the choice of which determines the type of error control guarantee. Some examples are given below.
\begin{enumerate}
\item {\em Fixed threshold.} We select test units whose conformal p-values~\eqref{eq:iid_conf_pval}
are below a fixed threshold $q\in(0,1)$. This could happen 
when testing a single hypothesis, 
or testing multiple hypotheses with Bonferroni correction. 
The latter  controls the family-wise error rate in finding large outcomes, 
which is useful in highly risk-sensitive settings such as disease diagnosis. 
\item {\em BH threshold.} We select test units whose conformal p-values~\eqref{eq:iid_conf_pval} 
are below  a data-dependent threshold given by the Benjamini-Hochberg 
(BH) procedure~\citep{benjamini1995controlling}. This selection procedure is shown in~\cite{jin2022selection} to control the FDR in detecting large outcomes, 
which is useful in exploratory screening such as drug discovery for ensuring efficient resource use in follow-up investigations. 
\end{enumerate}

For generality, we consider the selection rule in the form of  
$\hcS_\cp = \{j\in[m]: \hS_{n+j} \ge \tau\}$, where 
$\tau$ is a stopping time adapted to the filtration 
$\{\sigma(\{A_s\}_{s \le t}, \{B_s\}_{s\le t})\}_{t \in \RR}$:  
\begin{align}\label{eq:partial_sum}
A_s = 1 + \sum_{i \in [n]} \ind\{\hS_i \ge s, Y_i \le c_i\},
~B_s = \sum_{j\in[m]} \ind\{\hat S_{n+j} \ge s\}.
\end{align}
Equivalently, for any $t \in \RR$, we can write 
$ \ind\{\tau \le t\} = f_t(\{A_s\}_{s \le t}, \{B_s\}_{s \le t})$ 
for some function $f_t$.
It can be shown that the two examples above are special cases of this general rule 
(these results can be found in  Section~\ref{app:equiv} in the supplementary material).

We now provide a general solution for constructing selection-conditional 
prediction sets corresponding to such rules.
The challenge here is that all the p-values depend on each other, 
as they leverage the same set of calibration data, and/or 
the data-driven threshold determined by all p-values would further add to the intricacy. 
We note that the BH-based rule is studied in~\cite{bao2024selective} with 
approximate FCR control, while we are to provide an efficient solution with 
exact coverage guarantee. The following proposition lays out the form of the reference set 
and its validity, with its proof delegated to Section~\ref{app:proof_general_conf_pval} 
in the supplementary material.

\begin{proposition}\label{prop:general_conf_pval}
Suppose the selection set is $\hcS_\cp = \{j\in[m]: \hS_{n+j} \ge \tau\}$, 
where $\tau$ is determined by   
$\ind\{\tau \le t\} = f_t(\{A_s\}_{s \le t}, \{B_s\}_{s\le t})$, 
for some $f_t$ and $(A_s,B_s)$ defined in~\eqref{eq:partial_sum}, $\forall t \in \RR$.
For any $\mfS \in 2^{[m]}$ and $j\in[m]$ such that $j\in \hcS_\cp$
and $\hcS_\cp\in \mfS$, the reference set can be simplified as  
\#
& \hcR_{n+j}^\cp(y) = \ind\{y \le c_{n+j}\} \cdot \hcR_{n+j}^{\cp,1} + 
\ind\{y > c_{n+j}\} \cdot \hcR_{n+j}^{\cp,0}, \quad \text{where} \\ 
 \label{eq:conf_pval_ref}
\hcR_{n+j}^{\cp,k} = & \{i \in \cic: Y_i \le c_i, \hat S_i \ge  
\tau^{(j)}(k,0), \{\ell \in [m]: \hS_{n+\ell}^\swap{i}{j} \ge \tau^{(j)}(k,0)\} \in \mfS \} \\ 
& \cup \{i \in \cic: Y_i > c_i, \hat S_i \ge \tau^{(j)}(k,1), 
\{\ell \in [m]: \hS_{n+\ell}^\swap{i}{j} \ge \tau^{(j)}(k,1)\} \in \mfS\}.
\#
For $k,\ell \in \{0,1\}$, the adjusted threshold $\tau^{(j)}(k,\ell)$ is given by 
\begin{align}\label{eq:conf_pval_T}
& \ind\{\tau^{(j)}(k,\ell) \le t\} = f_t(\{A_s^{(j)}(k,\ell)\}_{s\le t}, \{B_s^{(j)}\}_{s\le t}), \text{where}\\
& A_s^{(j)}(k,\ell) = \ell + \sum_{i\in[n]} \ind\{\hS_i \ge s, Y_i \le c_i\} 
+ k\cdot \ind\{\hS_{n+j} \ge s\},~
B^{(j)}_s = 1 + \sum_{\ell \neq j} \ind\{\hS_{n+\ell} \ge s\}. \notag 
\end{align}
\end{proposition}
Based on Proposition~\ref{prop:general_conf_pval}, the JOMI prediction set is given by 
\$ 
\hC_{\alpha,n+j} = \big\{y \in \cY: y > c_{n+j}, V(X_{n+j},y)\le \hat{q}_0\big\} 
\cup \big\{y \in \cY: y \le c_{n+j}, V(X_{n+j},y)\le \hat{q}_1\big\},
\$
where $\hat{q}_k = \quant(1-\alpha;\{V_i:i\in \hcR^{\cp,k}_{n+j}\} \cup \{\infty\})$
for $k=0,1$. 
See Algorithm~\ref{alg:pval_fix} for a summary of the complete procedure, where we only present the 
deterministic version for simplicity. The overall computation complexity is 
at most $O(m(m+n)|\hat\cS|)$. 


\begin{algorithm}[htbp]
    \caption{JOMI for selection based on conformal p-values}
    \label{alg:pval_fix}
    \begin{algorithmic}
        
    \Require{Calibration data $\cdc$; test data $\cD_{\test}$; miscoverage level $\alpha$; selection taxonomy $\mfS$;
    selection rule $\cS$; nonconformity score function $V(\cdot, \cdot)$.}
    \State Compute $\hcS = \cS(\cdc,\cdte)$.
    
    \For{$j \in \hat\cS$}
        \State Compute $\tau^{(j)}(k,\ell)$ as in~\eqref{eq:conf_pval_T}, for $(k,\ell)\in \{0,1\}$. 
    
        \State Compute $\hat{\cR}_{n+j}^{\text{cp},0}$ and $\hat{\cR}_{n+j}^{\text{cp},1}$ as~\eqref{eq:conf_pval_ref}.

        \State Compute $\hat{q}_k =\text{Quantile}(1-\alpha;\{V_i: i\in \hat\cR_{n+j}^{\text{cp},k}\}\cup \{\infty\})$, for $k=0,1$.
        \State Compute $\hat{C}_{\alpha,n+j} = \{y\in \cY\colon y>c_{n+j},V(X_{n+j},y) \leq \hat{q}_0\}
        \cup \{y\in \cY\colon y\leq c_{n+j},  V(X_{n+j},y)\leq \hat{q}_1\}$.
    
    \EndFor
    \Ensure {$\{\hC_{\alpha,n+j}\}_{j\in \hcS}$.}
    \end{algorithmic}
    \end{algorithm}

\subsection{Selection based on conformal prediction sets}
\label{sec:conf_pred}

The final class of selection rules 
we study are based on the properties of 
(preliminary) prediction sets, usually constructed by 
running the vanilla SCP. 
Such use cases have appeared implicitly in many heuristic applications of conformal prediction. 
For example, practitioners may select
units whose prediction intervals are shorter/longer than a 
threshold, which roughly indicates enough confidence~\citep{sokol2024conformalized}. 
People may also select units 
whose prediction sets entirely lie above a threshold, which roughly 
indicates a desired outcome~\citep{svensson2017improving}. 
Note that 
the original prediction intervals are no longer valid 
 {\em conditional on being selected}~\citep{jin2022selection}, 
 and thus  using them for interpreting downstream uncertainty 
 can be misleading. 
 In this section, we apply our general framework to re-calibrate 
prediction sets for the units selected in such a way. 

Formally, we consider two stages of prediction set construction. 
The one constructed in the first stage, called 
the \emph{preliminary} prediction set, is used for determining the selection set $\hat\cS$.
The one in the second stage, which we call 
the \emph{selective} prediction set, is the one we are to build with JOMI.  
Following SCP in Section~\ref{subsec:SCP},
we let $S(x,y)$ and $V(x,y)$ be the nonconformity score 
functions for the two stages, respectively. 
The ($1-\beta$)-level preliminary 
prediction set for the $j$-th test unit is  
\#\label{eq:CP_prelim}
\hC^{\text{prelim}}_{\beta,n+j} 
= \big\{y\in \cY: S(X_{n+j},y) \le \eta \big\},
\#
where $\eta$ is the $K \,:=\, \lceil (1-\beta)(n+1)\rceil$-th smallest element in 
$\{S(X_i,Y_i)\}_{i=1}^n$.

We consider any selection rule based on the preliminary prediction set 
$\hat{C}_{\beta,n+j}^{\text{prelim}}$. 
Note that by~\eqref{eq:CP_prelim}, given the first-stage score function $S(\cdot,\cdot)$, 
the form of $\hat{C}_{\beta,n+j}^{\text{prelim}}$
is fully determined by $X_{n+j}$ and $\eta$. 
We can thus express any selection rule through $\cL: \cX \times \RR \mapsto \{0,1\}$,
where $\cL(X_{n+j},\eta) = 1$ means selecting the unit and $\cL(X_{n+j},\eta) = 0$ otherwise. 
An example is selecting based on prediction interval lengths: 
following~\cite{sokol2024conformalized}, 
suppose that we use CQR~\citep{romano2019conformalized} in the first stage, 
i.e., 
$S(x,y) = \max\{\hat{q}_L(x) - y, y - \hat{q}_U(x)\}$, where $\hat{q}_L(x)$
and $\hat{q}_U(x)$ are estimates of some lower and upper conditional quantiles. 
Selecting prediction intervals shorter than a threshold $\lambda$ gives 
$\cL(x,\eta) = \ind\{\hat{q}_U(x) 
- \hat{q}_L(x) + 2\eta \le \lambda\}$. 
As another example, 
for a binary outcome $Y$, 
we might want to select units whose prediction set is a singleton, leading to 
$\cL(x,\eta) = \ind\big\{ S(x,1)\le \eta <S(x,0) \text{ or } S(x,0)\le \eta <S(x,1) \big\}$. 

Having determined the selection rule $\cL$, the selection set is thus  
$
\cS(\cD_\calib,\cD_\test) = \hcS_{\ps} = \big\{j\in [m]: \cL(X_{n+j},\eta) =1 \big\}.
$
We are to derive a computationally efficient but slightly conservative version 
of the JOMI prediction set, which nevertheless has tight coverage 
in all our numerical experiments (see Section~\ref{sec:icu}).  
Define
\begin{equation} 
\begin{aligned}
\label{eq:conf_pred_set}
\hC^\ps_{\alpha,n+j}
:= ~& \{y: \eta^- \le S(X_{n+j},y) \le \eta^+\} \cup \{y: V(X_{n+j},y)\le q_{j,1}~~\text{and}~~ S(X_{n+j},y) < \eta^-\}\\
& \cup \{y: V(X_{n+j},y)\le q_{j,2}~~\text{and}~~ S(X_{n+j},y) > \eta^+\}
\end{aligned}
\end{equation}
where $\eta^+$ and $\eta^-$ are the $(K+1)$-th and $(K-1)$-th 
smallest element in $\{S_i\}_{i=1}^n$, respectively, and  
\@\label{def_ps_q}
& q_{j,1} := \quant\Big(1-\alpha; 
 \big\{V_i\colon i\in [n], S_i \le \eta^-, \cL(X_i,\eta)=1, 
\{\ell \in [m]: \cL(X_{n+\ell}^\swap{i}{j}, \eta) = 1\} \in \mfS \big\}\notag\\
& \qquad \qquad \qquad \qquad \qquad 
 \cup \big\{V_i: S_i > \eta^-, ~\cL(X_i,\eta^-) = 1 , 
\{\ell \in [m]: \cL(X_{n+\ell}^\swap{i}{j}, \eta^-) = 1\} \in \mfS\big\}\Big);\notag \\
& q_{j,2} := \quant\Big(1-\alpha;\big\{V_i\colon i\in [n],  S_i \le \eta, \cL(X_i,\eta^+) = 1,
\{ \ell \in [m]: \cL(X_{n+\ell}^\swap{i}{j}, \eta^+) = 1\} \in \mfS\big\}\notag\\
& \qquad \qquad \qquad \qquad \qquad 
\cup \big\{V_i: S_i > \eta, ~\cL(X_i,\eta)=1, 
\{\ell \in [m]: \cL(X_{n+\ell}^\swap{i}{j}, \eta) = 1\} \in \mfS\big\}\Big).
\@

We prove the validity of $\hC_{\alpha,n+j}^\ps$ in~\eqref{eq:conf_pred_set} 
below, whose proof is in Supplementary Section~\ref{appd:proof_conf_pred}.
\begin{proposition}
\label{prop:conf_pred}
For any selection rule $\cL$, any $j\in [m]$ 
and any selection taxonomy $\mfS$ such that 
$j\in \hcS_\ps$ and $\hcS_\ps \in \mfS$, 
 $\hC_{\alpha,n+j}^\ps$  
is a superset of the \textnormal{JOMI} prediction set $\hC_{\alpha,n+j}$ 
defined in~\eqref{eq:PI_cond}, and 
\$
\hat{C}_{\alpha,n+j}^\ps \backslash \hC_{\alpha,n+j} 
\subseteq \big\{y\in \cY\colon \eta^- \leq S(X_{n+j},y) \leq \eta^+ \big\}.
\$
\end{proposition}

By Proposition~\ref{prop:conf_pred}, 
the conservativeness of $\hat{C}_{\alpha,n+j}^\ps$ is quite limited, as 
$\eta^-$ and $\eta^+$ are usually very close to each other. We also verify 
its tight empirical coverage in Section~\ref{sec:icu}. 

The procedure is summarized in Algorithm~\ref{alg:ps}. 
For each $j$, the computation cost 
of $\hat{C}_{\alpha,n+j}^\ps$ is $O(m + n)$, and 
therefore the overall computation cost is $O(m(m+n))$.  

\begin{algorithm}[htbp]
    \caption{JOMI for selection based on preliminary prediction sets}
    \label{alg:ps}
    \begin{algorithmic}
        
    \Require {Calibration data $\cdc$; test data $\cD_{\test}$; 
    selection taxonomy $\mfS$;
    selective miscoverage level $\alpha$; 
    first-stage miscoverage level $\beta$;
    selection rule $\cS$; first-stage score function $S(\cdot,\cdot)$; 
    nonconformity score function $V(\cdot, \cdot)$.}
    \State $\hcS = \cS(\cdc,\cdte)$.

    \State Compute $\eta^+$ as the $(K+1)$-th order statistic of $\{S_i\}_{i=1}^n$, where $K=\lceil (1-\beta)(n+1)\rceil$. 

    \State Compute $\eta^+$ as the $(K-1)$-th order statistic of $\{S_i\}_{i=1}^n$.

    \For{$j \in \hat\cS$}
        
    \State    Compute $q_{j,1}$ and $q_{j,2}$ as in~\eqref{def_ps_q}.  
    
    \State    Compute $\hC_{\alpha,n+j}^{\ps}$ as in~\eqref{eq:conf_pred_set}.
    
    \EndFor
    \Ensure {$\{\hC_{\alpha,n+j}\}_{j\in \hcS}$.}
    \end{algorithmic}
    \end{algorithm}

\begin{remark}
We note that it is possible that the prediction interval produced by Algorithm~\ref{alg:ps} 
does not satisfy the constraints on the preliminary prediction sets.   
This is, however, a desired feature in our setting where the selection rule is 
given and the inference step is decoupled from the selection step; the selection-conditional 
prediction sets should be used as tools for informing further decisions/analysis. Otherwise, we may need to take an orthogonal strategy: change the selection algorithm, for which ideas from~\cite{jin2022selection,gazin2024selecting} may be useful.
\end{remark}

\section{Application to drug discovery}
\label{sec:drug}
In drug discovery, powerful prediction machines are increasingly used to 
guide the search of promising drug candidates.  
For such high-stakes decisions,  
it is important 
to quantify the uncertainty in 
the predictions~\citep{svensson2017improving,jin2022selection,laghuvarapu2024codrug}. 
Meanwhile, selection issues naturally arise 
as scientists may only focus on seemingly promising drugs.

In this section, we apply $\mname$ to several application scenarios in drug discovery  
with a selective nature.  
In some cases, $\mname$ yields shorter prediction intervals than vanilla conformal prediction 
when the latter is under-confident; in others, it makes the just right inflation of the prediction interval to provide exact selection-conditional coverage.

\vspace{0.5em}
\noindent\textbf{Application scenarios.}
In the main text, we focus on \emph{drug property prediction} (DPP),  a classification problem 
where the binary outcome indicates whether a drug candidate binds to a pre-specified disease target, and the covariates are the (encoded) chemical structure of the drug compound. Due to limited space, we defer the results for several selection scenarios in drug-target-interaction prediction (DTI) to Supplementary Section~\ref{subsec:DTI}. DTI is a regression problem 
where each sample is a pair of drug and disease target. The outcome of interest 
is a real-valued variable indicating the binding affinity of that pair. 
The covariates are the (concatenated) encoded structure of both.


\vspace{0.5em}
\noindent\textbf{Selection rules.}
We consider three types of realistic selection rules $\cS$: 
\begin{enumerate}[(1)]
    \item \emph{Covariate-dependent top-K selection}: selecting drugs with highest predicted binding affinities. 
    \begin{enumerate}[(i)]
        \item Top-K among test data. When the scientist has a fixed budget of investigating $K$ drug candidates in the next phase, one may select  $K$ test samples with the largest  $\hat\mu(X_{n+j})$. 
        \item Top-K among mixed data. When the scientist is to investigate other properties for promisingly active drug candidates in the next phase, they may select $K$ units in $\cD_\calib\cup\cD_\test$ with the highest predicted affinities. 
        \item Calibration-referenced selection. The scientist may use the $\cD_\calib$ as reference and select test samples whose predicted activities are greater than the $K$-th highest in $\{\hat\mu(X_i)\}_{i\in \cI_\calib}$. 
    \end{enumerate} 
    \item \emph{Conformal selection}. The scientist might also obtain a subset of active drugs while controlling the FDR below some $q\in (0,1)$. In this case, $\cS(\cD_\calib,\cD_\test)$ is the set of test drugs picked by Conformal Selection in~\cite{jin2022selection} at FDR level $q\in(0,1)$. 
    \item \emph{Selection with constraints}. The cost of developing a drug is a variable $C_{n+j}$, and one wants to select as many drugs with the highest predictions $\hat\mu(X_{n+j})$ as possible while ensuring the total cost is below some constant $C$, 
    or maximize the total reward while controlling the costs.  
\end{enumerate} 

\noindent\textbf{Evaluation metrics.}
We evaluate the selection-conditional (mis)coverage 
via the consistent estimator
$
\hat{\textrm{Miscov}} = 
\frac{\sum_{j=1}^{m}\hat{P}(j\in \hat\cS, Y_{n+j}\notin \hat{C}_{\alpha,n+j})}
{\sum_{j=1}^{m}\hat{P}(j\in \hat\cS)},
$
where $\hat{P}$ is the empirical probability over all repeats. 
The size of prediction sets is evaluated by averaging the 
cardinality of $\hat{C}_{\alpha,n+j}$ for classification or length of $\hat{C}_{\alpha,n+j}$ 
for regression over selected test units in all repeats. 
We also evaluate (but do not show in figures for brevity) the false coverage rate via 
$
\hat{\textrm{FCR}} = \hat{E}\Big[ 
\frac{\sum_{j=1}^{m}\ind\{ j\in \hat\cS, Y_{n+j}\notin \hat{C}_{\alpha,n+j}\}}{1\vee|\hat\cS|}  \Big],
$
where $\hat{E}$ is the empirical mean over all repeats. 

\vspace{0.5em}
\noindent\textbf{Data and prediction models.}
We use the HIV screening data in the DeepPurpose library~\citep{huang2020deeppurpose} with 
a total sample size of $n_{\textnormal{tot}} = 41127$. 
The numerical features $X\in \cX$ are  encoded by Extended-Connectivity FingerPrints~\citep[ECFP]{rogers2010extended} which
characterize topological properties of the drug candidates.  
A small-scale neural network $\hat\mu\colon \cX\to [0,1]$ 
is trained on randomly sampled $20\%$ of the entire dataset. 
Then, the remaining is randomly 
split into two equally-sized folds $\cD_\calib$ 
and $\cD_\test$. 
The exchangeability among $\cD_\calib$ 
and $\cD_\test$ is thus satisfied. 
Sampling of the training data and model training is independently repeated  $10$ times;
within each, we randomly split calibration/test data for $100$ times. This leads 
to $N=1000$ runs in total.

\subsection{Top-K selection}
\label{subsubsec:dpp-topK-test}

We first consider the selection rule (i) top-K with test data. 
For a trained predictor $\hat\mu$, we select (a)  the highest $K$ values of $\hat\mu(X_j)$ and (b) the lowest $K$ values of $\hat\mu(X_j)$ among all test data. 
The results for (ii) top-K selection with both test and calibration data, 
as well as (iii) calibration-referenced selection 
are deferred to Supplementary Section~\ref{app:subsec_topK_dpp},
showing similar patterns.

For a confidence level $\alpha\in(0,1)$, 
we apply the vanilla conformal prediction and  
JOMI (both deterministic and randomized) 
to the selected units at level $\alpha$, 
with $\mathfrak{S} = 2^{[m]}$.  
Since $|\hat\cS|$ is fixed, 
selection-conditional coverage implies FCR control at level $\alpha$ due to Proposition~\ref{prop:notions}. 
In addition, we include the BY procedure~\citep{benjamini2005false} as a heuristic baseline.
We vary $\alpha \in \{0.1,0.2,\dots,0.9\}$ 
and $K\in\{20, 100, 1000, 2000, 5000, 10000, 15000\}$, 
and set $V$ as
the APS score~\citep{romano2020classification}. The results with another binary score is in Supplementary Section~\ref{app:subsec_dpp_binary}. 

Figure~\ref{fig:dpp_topK_test} shows the empirical selection-conditional miscoverage 
in the left panel  
and the average prediction set size in the right panel, 
both at nominal levels $\alpha\in\{0.1, 0.8\}$. 
The solid (resp. dashed) lines show the results when we select units with 
the highest (resp. lowest) predicted affinities.

\begin{figure}[htbp]
    \includegraphics[width=\textwidth]{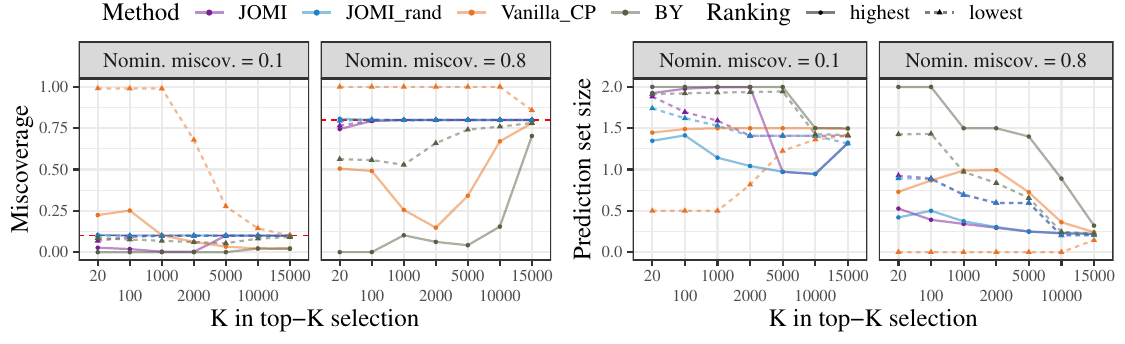}
    \caption{Empirical selection-conditional miscoverage (left) and prediction set size (right) in drug property prediction for vanilla conformal prediction (\texttt{Vanilla\_CP}), BY correction (\texttt{BY}), JOMI and randomized JOMI (\texttt{JOMI\_rand}) for test units whose $\hat\mu(X_{n+j})$ are top-$K$ (\texttt{highest}, solid line) or bottom-$K$ (\texttt{lowest}, dashed line) among test units. The reds dashed line is the nominal miscoverage level $\alpha\in\{0.1,0.8\}$. The results are averaged over $N=1000$ runs.}
    \label{fig:dpp_topK_test}
\end{figure}

The orange curves (\texttt{Vanilla\_CP}) show that 
vanilla conformal prediction with APS scores 
is over-confident for units with the lowest predicted affinities 
while under-confident for units with highest predicted affinities. 
In contrast, the purple (\texttt{JOMI}) and blue (\texttt{JOMI\_rand}) curves
both show valid coverage for our proposed methods. 
Also, using the APS score introduces visible gap between the actual coverage and $1-\alpha$ 
for $\mname$ due to discretization, which is made exact by its randomized version. 
Finally, while \texttt{BY} has lower-than-nominal miscoverage which shows its validity as a heuristic method, it is overly conservative, leading to much larger prediction sets.

Interestingly, vanilla conformal prediction with the APS score 
yields zero-cardinality prediction sets 
for units with the lowest predicted affinities 
with $\alpha=0.8$ (the orange dashed line in the right panel of Figure~\ref{fig:dpp_topK_test}). This is because vanilla CP covers other test units with very high rate, 
and thus marginal coverage is guaranteed even with empty prediction sets for 
the selected units. 
Of course, this is worrying if one cares more about these units at the bottom.

The behavior of these methods also depends on the choice of 
the nonconformity score $V$. 
Supplementary Section~\ref{app:subsec_dpp_binary} 
shows the results for selection rules (i)-(iii) with the binary score
$V(x,y) = y(1-\hat\mu(x))+(1-y)\hat\mu(x)$. 
With the binary score, 
$\mname$ consistently achieves valid coverage for selected units, 
and the coverage gap of $\mname$ 
due to discretization is less visible. 
For vanilla CP, opposite to the situations here, 
it is over-confident 
for units with the highest predicted affinities yet under-confident for units with 
the lowest predicted affinities.

\subsection{Conformal selection}

We then consider conformal selection, where the focal test units 
are those believed 
to obey $Y=1$ with false discovery rate control at level $q\in(0,1)$. 
This problem was investigated in~\cite{bao2024selective}  
with no exact finite-sample coverage guarantees in theory (though their heuristic method performs reasonably in their empirical studies). 
We apply conformal  selection~\citep{jin2022selection} at 
FDR level $q\in\{0.2,\dots,0.9\}$ and $c_{n+j}\equiv 0.5$ to determine $\hat\cS$, 
and construct prediction intervals for units in $\hat\cS$ using 
both the APS score and the binary score. 
Experiments are repeated for $N=1000$ independent runs.

Figure~\ref{fig:dpp_csel} depicts  $\hat{\textnormal{Miscov}}$ 
and prediction interval lengths 
with nominal miscoverage level $\alpha\in\{0.1,0.8\}$ under various 
FDR levels $q$ and two choices of nonconformity score $V$.

\begin{figure}[htbp]
    \includegraphics[width=\textwidth]{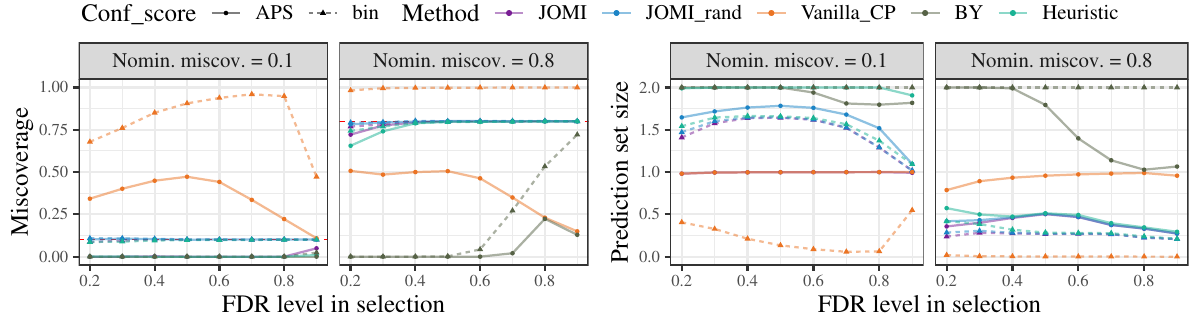}
    \caption{Empirical selection-conditional miscoverage (left) and prediction set size (right)
     in drug property prediction for test units selected by conformal selection 
      at target FDR levels $q\in \{0.2,\dots,0.9\}$ when applying 
     vanilla conformal prediction (\texttt{Vanilla\_CP}), BY (\texttt{BY}), JOMI (\texttt{JOMI}),
      randomized JOMI (\texttt{JOMI\_rand}), and the method of~\cite{bao2024selective} 
      (\texttt{Heuristic}).
       Results with the APS score are in solid lines, while those with the 
       binary score are in dashed lines. Details are otherwise the same as Figure~\ref{fig:dpp_topK_test}.}
    \label{fig:dpp_csel}
\end{figure}

Vanilla CP is not calibrated for 
selected units: it is over-confident with both scores 
for $\alpha=0.1$, while being over-confident with binary score 
and under-confident with APS score at $\alpha=0.8$.  
In contrast, \texttt{JOMI} and 
\texttt{JOMI\_rand} achieves valid selection-conditional coverage in all scenarios. 
There is some gap for \texttt{JOMI} with the APS score due to discretization, 
but not for \texttt{JOMI\_rand} or the binary score. 
We observe an even lower empirical FCR than the conditional miscoverage for our methods
(so they achieve valid FCR control);
this is because the selection set can sometimes be empty for small values of FDR level $q$ 
(recall Proposition~\ref{prop:weak_FCR}).  
Compared with the heuristic methods of~\cite{bao2024selective}, 
our method usually achieves smaller prediction set sizes whereas their method 
seems overly conservative. We conjecture that this is due to a more delicate choice of 
the reference set. 
Finally, \texttt{BY} is also overly conservative despite valid empirical coverage.

\subsection{Selection with constraints}

We now consider selecting units with the highest predicted binding affinities 
within a total budget of subsequent development. In this case, 
$\cS(\cD_\calib,\cD_\test) = \{ j\in [m]\colon \hat{\mu}(X_{n+j})\geq \bar\mu\}$, 
where $\bar\mu = \max\{\mu \colon \sum_{j=1}^m L_{n+j}\ind\{\hat\mu(X_{n+j})\geq \mu\}\leq C\}$, 
and $\{L_{n+j}\}_{j\in [m]}$ are the costs.

We create semi-synthetic datasets since the original HIV data does not contain the cost information. 
Specifically, for each $i\in [n+m]$, 
we generate $L_i = \exp(3\hat\mu(X_i)) + 2 |\sin(\hat\mu(X_i))| +\epsilon_i$, 
where $\hat\mu(X_i)$ is the predicted binding affinity, and $\epsilon_i \sim \text{Exp}(1)$ are i.i.d.~random variables that capture other cost-related information. 
Setting $20\%$ of the data aside 
as the training set, 
we randomly sample the data without replacement so that $n=2500$ and $m=2500$. 

The average miscoverage, prediction set size, and reference set size 
are reported in Figure~\ref{fig:dpp_cons}. 
Interestingly, after adding cost constraints, vanilla conformal prediction 
is over-confident with the binary score
and under-confident with the APS score. 
In contrast, our methods always yield near-exact coverage. 
From the right-most plot, we see that $|\hat\cR_j|$ is positively correlated 
with the number of selected test units. 
As usual, \texttt{BY} is overly conservative.

\begin{figure}[htbp]
    \centering
    \includegraphics[width=0.85\textwidth]{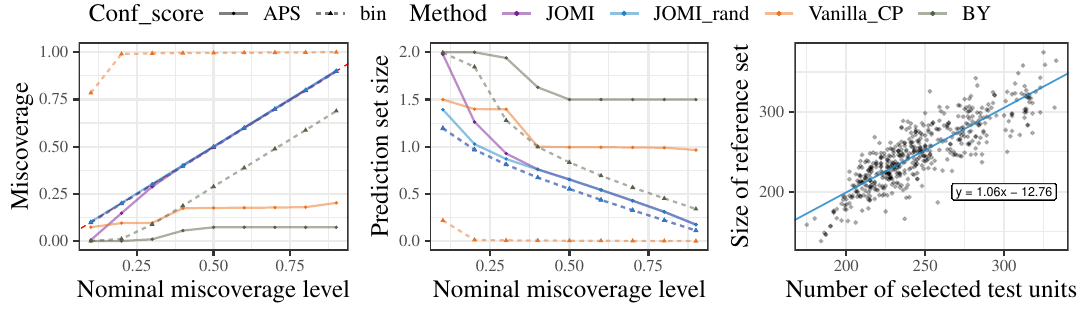}
    \caption{Empirical miscoverage rate (left), average length of prediction interval (middle), 
    and scatter plots for averaged reference set size 
    $\{|\hat\cR_j|\}_{j\in \hcS}$ versus $|\hcS|$  (right), across  $N=500$ independent runs of \texttt{Vanilla\_CP}, \texttt{BY}, \texttt{JOMI}, and \texttt{JOMI\_rand}. The $x$-axis of left and middle plots is the nominal miscoverage level $\alpha\in\{0.1,0.2,\dots,0.9\}$.}
    \label{fig:dpp_cons}
\end{figure}

    Finally, we report the empirical FCR in the three tasks in Supplementary Section~\ref{app:subsec_DPP_FCR}. 
    Consistent with our theory in Section~\ref{subsec:compare}, selection-conditional coverage implies FCR control in top-K selection, and the empirical FCR is also close to the nominal coverage level under conformal selection where the selection size is stable. However, in the constrained optimization task, we find that the empirical FCR of vanilla CP is much lower than the selection-conditional coverage. Similar to the ideas in Proposition~\ref{prop:strong_weak}, this is because the FCR is defined as zero when the selection set is empty. This shows the limitation of FCR as an error metric: even if it is controlled, on the ``interesting'' event that the selection set is non-empty, the actual coverage can be lower than anticipated.

\section{Application to health risk prediction}
\label{sec:icu}
Prediction machines are also widely used in healthcare for guiding clinical decision-making. 
These decisions may come from complicated underlying processes 
such as clinical resource optimization~\citep{ahmadi2017outpatient,master2017improving} or 
preliminary uncertainty quantification~\citep{olsson2022estimating}. 
Accounting for selections from complicated decision processes 
is necessary for reliable and informative uncertainty quantification. 
In this section, we demonstrate 
the application of our framework to health risk prediction 
settings.
We will consider three selection rules:

\begin{enumerate}[(1)]
    \item \emph{Selection with constraints}. Clinical decision makers may optimize a performance measure  subject to certain constraints such as budget, capacity, or fairness~\citep{castro2012combined,kemper2014optimized,gocgun2014dynamic}, and send patients to different care categories. In Section~\ref{subsec:icu_cons}, we consider a stylized example where we minimize the total predicted ICU stay subject to total cost budget, where the selected units can be viewed as patients sent to a certain category. 
    \item \emph{Selecting small conformal prediction sets}. Based on preliminary conformal prediction, one may  suggest human intervention for units with large set sizes~\citep{olsson2022estimating,sokol2024conformalized} while leaving those with small prediction sets unattended. In Section~\ref{subsec:icu_size}, we re-calibrate predictive inference for those whose preliminary prediction sets are shorter than a threshold. 
    \item \emph{Selection based on upper prediction bounds.} Practitioners may also focus on units whose preliminary 
    prediction intervals lie below a threshold~\citep{business_conformal}. In Section~\ref{subsec:icu_high}, we study selected test units whose preliminary prediction intervals 
    have upper bounds below a threshold.
\end{enumerate}

Among the above, rule (1) is covariate-dependent, 
while rules (2) and (3) can depend on the outcomes.
All of them are efficiently tackled by the computation tricks in Section~\ref{sec:compute}. 

Our experiments use the ICU-stay data in the \texttt{MIMIC-IV} dataset~\citep{johnson2023mimic}. 
Data (and features in $X$) are pre-processed using the pipeline provided by~\cite{gupta2022extensive}, 
with the outcome $Y$ being the length of ICU stay. 
We use random forests in the \texttt{scikit-learn} Python package 
to train a point prediction model $\hat\mu(\cdot)$ 
and two quantile regression models $\hat{q}_{1-\beta/2}(\cdot)$ and $\hat{q}_{\beta/2}(\cdot)$ 
using a holdout training set. Calibration and test data are randomly split with $n=3000$ and $m=2000$. As the discretization issue is minimal in this regression problem, we only present results for JOMI without randomization. 

\subsection{Selection with constraints}
\label{subsec:icu_cons}

We first study the case where test units are
selected by minimizing the total predicted ICU stay time subject to a budget constraint. 
Formally, for each patient $i\in [m+n]$, 
we let  $\hat\mu(X_i)$ be its predicted ICU stay length, 
and $L_i>0$ be the budget needed for them. 
Then, $\cS(\cD_\calib,\cD_\test)$ aims to solve the following optimization problem: 
\#\label{eq:knapsack_icu}
\mathop{\text{maximize}}_{S\subseteq [m]}&\quad  
\sum_{j\in S}  \hat\mu(X_{n+j})   \\ 
\text{subject to}&\quad  \sum_{j\in S} L_{n+j} \leq \bar{L}, \notag
\#
where $\bar{L} = 200$ is a budget limit. Again, as the dataset does not come with 
drug development costs, we generate $L_i = \lceil\exp(3\hat\mu(X_i)/\bar\mu) +   |\sin(\hat\mu(X_i))| +\epsilon_i-1+\varepsilon_i \rceil$, 
where $\bar{\mu} = \max_{i \in \cdt}|\hat\mu(X_i)|$, and
$\epsilon_i \sim \text{Exp}(1)$, $\varepsilon_i\sim \text{Unif}([0,1])$ are independent random variables.

The optimization problem~\eqref{eq:knapsack_icu} is known as the Knapsack problem which is 
NP-hard. Nevertheless, there are efficient approximate solvers and 
we note that the validity of JOMI does not rely on exactness of the results; in our experiments, 
we use the Python package \texttt{mknapsack}~\citep{mknapsack}. 
Existing methods such as~\cite{bao2024selective} cannot 
deal with such a complicated selection process. 
In contrast, our framework tackles this problem with 
a computation complexity that is polynomial in $m$, $n$, 
and the complexity of the subroutine $\cS(\cdot,\cdot)$. 


Figure~\ref{fig:icu_cons} shows the empirical miscoverage, 
length of prediction interval, and sizes of the selection set and reference sets. 
While vanilla conformal prediction is over-confident and BY is overly conservative, 
our method achieves exact coverage for selected test units despite the complexity of the selection process. 
We also see a slightly positive correlation of 
$|\hat\cR_{n+j}|$ and $|\hcS|$. 

\begin{figure}[htbp]
    \centering
    \includegraphics[width=0.85\textwidth]{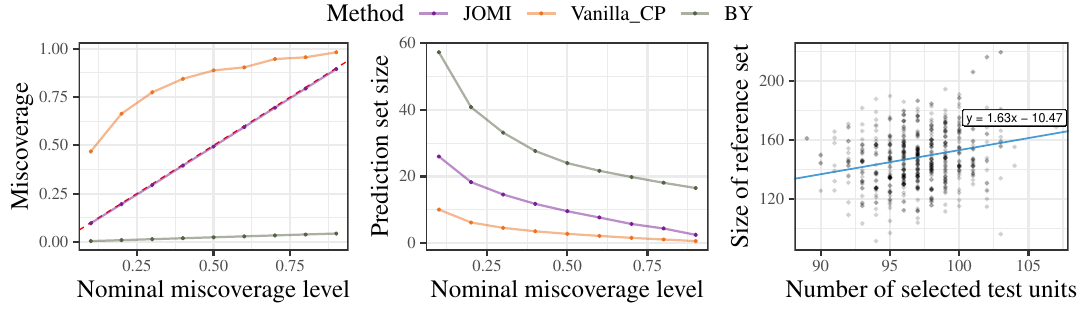}
    \caption{Empirical miscoverage rate (left), average length of prediction interval (middle), 
    and scatter plots for the averaged reference set size versus $|\hcS|$  (right), across  $N=500$ independent runs of \texttt{Vanilla\_CP}, \texttt{BY}, \texttt{JOMI}, and \texttt{JOMI\_rand}. The $x$-axis of left and middle plots is the nominal level $\alpha\in\{0.1,0.2,\dots,0.9\}$.}
    \label{fig:icu_cons}
\end{figure}

\subsection{Selecting small-sized prediction sets}
\label{subsec:icu_size}

We then consider the second selection rule, 
where we first build preliminary conformal prediction 
intervals $\hat{C}_{\alpha,n+j}^{\text{prelim}}$ via 
the score function $S(x,y) = |y-\hat\mu(x)|/\hat\sigma(x)$~\citep{lei2018distribution}; 
both the point prediction function $\hat\mu(\cdot)$ and the conditional 
standard deviation are estimated via random forests. 
We then select those test units with 
$|\hat{C}_{\alpha,n+j}^{\text{prelim}}|\leq 5$, i.e., 
the upper and lower bounds of 
the preliminary prediction intervals are less than $5$ days apart. 
This mimics the ideas in~\citep{ren2023robots,sokol2024conformalized} 
where small-sized prediction sets are ``certified'' as confident. 

After selection, we leverage the method in Section~\ref{sec:conf_pred} to construct 
$\hat{C}^{\textrm{ps}}_{\alpha,n+j}$ for all selected test units.  
Since $\hat{C}^{\textrm{ps}}_{\alpha,n+j}$ 
is a superset of 
the exact output $\hat{C}_{\alpha,n+j}$, we evaluate its empirical coverage 
to investigate whether it is over-conservative. 
Also, note from~\eqref{eq:conf_pred_set} that it 
is the union of three subsets; we also evaluate 
the number of disjoint segments in  $\hat{C}^{\textrm{ps}}_{\alpha,n+j}$. 

The miscoverage and  length of prediction sets are reported in the left and middle 
plots in Figure~\ref{fig:icu_len}. We observe that selectively certifying short prediction intervals 
can lead to under-coverage (orange curve), while JOMI achieves exact coverage (purple curve) 
by inflating the prediction sets, 
meaning that the superset $\hat{C}^{\textrm{ps}}_{\alpha,n+j}$ is
effectively quite tight.   
The average number of disjoint segments in the right plot of Figure~\ref{fig:icu_len}, 
which shows that $\hat{C}^{\textrm{ps}}_{\alpha,n+j}$ is almost always  one single interval. 

\begin{figure}[htbp]
    \centering
    \includegraphics[width=0.85\textwidth]{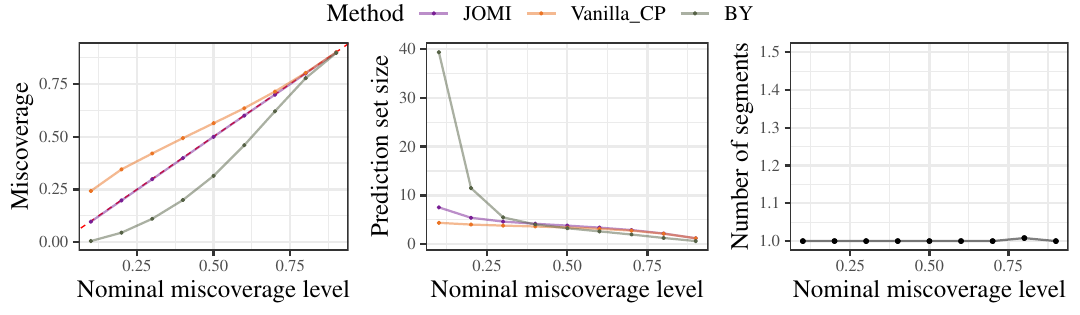}
    \caption{Empirical miscoverage rate (left), average length of prediction interval (middle), 
    and average number of segments in $\hat{C}^{\textrm{ps}}_{\alpha,n+j}$ (right), across  $N=500$ runs of \texttt{Vanilla\_CP} and \texttt{JOMI} when test units with short preliminary prediction sets are selected. The  $x$-axis is nominal levels $\alpha\in\{0.1,0.2,\dots,0.9\}$.}
    \label{fig:icu_len}
\end{figure}

\subsection{Selecting prediction sets below a threshold}
\label{subsec:icu_high}

Finally, we study selection rule (iii) which is also based on a preliminary 
prediction set constructed in the same way as Section~\ref{subsec:icu_size}.
We imagine that practitioners select a test unit $n+j$ 
if the upper bound of $\hat{C}^{\text{pre}}_{\alpha,n+j}$ is below $6$, i.e., 
it appears the patient will stay in ICU less than 6 days. 

The miscoverage rate, length of prediction sets, and number of disjoint segments
averaged over $N=500$ independent runs of JOMI 
(Section~\ref{sec:conf_pred}) 
and vanilla conformal prediction 
are summarized in Figure~\ref{fig:icu_upp}. We see that 
preliminary prediction sets with low upper bounds tend to under-cover for small $\alpha$, 
while JOMI achieves exact coverage despite 
that we construct a superset of $\hat{C}_{\alpha,n+j}$. 
However, JOMI may produce multiple segments for large values of $\alpha$. 

\begin{figure}[htbp]
    \centering
    \includegraphics[width=0.85\textwidth]{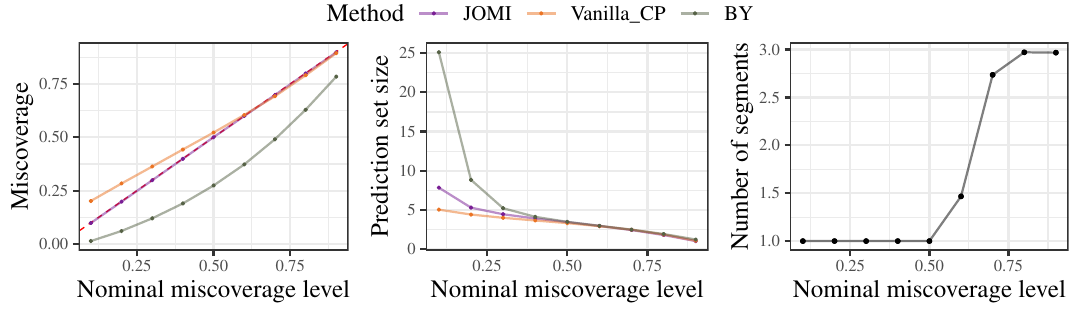}
    \caption{Empirical miscoverage rate (left), average length of prediction interval (middle), 
    and average number of segments in $\hat{C}^{\textrm{ps}}_{\alpha,n+j}$ (right), across  $N=500$ independent runs of \texttt{Vanilla\_CP} and \texttt{JOMI}. The nominal miscoverage levels on the $x$-axis are $\alpha\in\{0.1,0.2,\dots,0.9\}$.}
    \label{fig:icu_upp}
\end{figure}

    Finally, we present in Supplementary Section~\ref{app:subsec_icu_fcr} the tight empirical FCR control of JOMI in the above three tasks, since the size of the selection set is stable and nonzero.

\paragraph{Reproducibility and data availability.}
{The code for reproducing the numerical results in this paper 
is available at \url{https://github.com/ying531/JOMI-paper}.
The data underlying this article are available online, with 
references given at appropriate places in the paper.
}

\paragraph{Conflict of interest.}
{The authors declare no conflict of interest.}

\paragraph{Acknowledgments.}
{The authors thank the anonymous reviewers for their constructive comments.
Z.R.~acknowledges support from the National Science Foundation under grant DMS-2413135.
Y.J. was partially supported by the Harvard Data Science Wojcicki-Troper Postdoctoral Fellowship.

\bibliographystyle{apalike}
\bibliography{ref}

\newpage
\appendix
\section{Related works}
\subsection{Literature overview}
\label{app:overview}
Our method extends conformal prediction, which is a general framework for 
building marginally valid prediction intervals~\citep{vovk2005algorithmic,lei2018distribution,romano2019conformalized,romano2020classification}. 
Compared with vanilla conformal prediction, we aim at 
selection-conditional coverage for a unit of interest, 
instead of marginal coverage for an ``average'' unit. 
Particularly related to our framework is 
Mondrian Conformal Prediction~\citep[MCP]{Vovk2003MondrianCM}, 
which uses a subset of calibration data to build a prediction set
with coverage conditional on a pre-determined or permutation-equivariant 
class membership. 
Achieving selection-conditional coverage
is  a new goal with distinct motivations from MCP. Moreover, our techniques 
and scopes of application are significantly different from those of 
MCP  (see Supplementary Section~\ref{app:MCP}).

The selection-conditional guarantee we seek for 
is closely related to the 
post-selection inference (POSI) literature~\citep{berk2013valid,tibshirani2016exact,chernozhukov2015valid}  and in particular 
conditional selective inference~\citep{lee2016exact,markovic2017unifying,reid2017post,tibshirani2018uniform,kivaranovic2020tight,mccloskey2024hybrid}.
In POSI, one important task is to build confidence intervals 
for selected model parameters conditional on a 
selection event, which is close to our  selection-conditional coverage guarantee. 
Methods for this goal usually leverage specific problem structures such as linearity
and distribution of the estimators~\citep{zhong2008bias,lee2016exact,tian2017asymptotics,andrews2024inference,liu2023exact}.
As we focus on predictive inference (where the inferential targets are random variables 
instead of parameters), 
we leverage the exchangeability across units to achieve exact selection-conditional 
coverage, and thus our method substantially differs from POSI. 

Another related error notion in the selective inference literature  
is the false coverage rate (FCR)~\citep{benjamini2005false}, which is defined as 
the expected fraction of selected objects not covered 
by their confidence/prediction sets. 
Methods for achieving FCR control have been proposed 
for various settings~\citep{benjamini2005false,weinstein2013selection,weinstein2020online,
zhao2020constructing,xu2022post}.
In particular,~\citet{weinstein2020online} introduce
an online FCR-controlling method 
that can be applied to conformal prediction with selection. 
They consider sequentially arriving test units, 
and require the selection of a current unit to only depend on 
its prediction set and 
past selection decisions.  
In contrast,  
we allow the selection rule to simultaneously depend on {\em all} test units, 
such as top-K selection, 
optimization-based selection, and conformal p-value-based selection.  
Our approach to error control differs from theirs as well: 
while they achieve FCR control by adjusting the confidence level 
of marginal prediction sets, 
we leverage a judiciously chosen subset of calibration data 
to construct post-selection prediction sets with
selection-conditional coverage. 
Another closely related work is~\citet{bao2024selective} who  study prediction sets 
for selected units with (approximate) FCR control, yet
they achieve exact guarantee only for a limited class of
selection rules such as top-K selection.
Later on, we will show that JOMI yields the same prediction sets as theirs 
when applied to top-K selection, and leads to exact guarantee 
for  other problems they addressed with approximate FCR control. 
It also applies to many more settings with complicated 
selection rules that cannot be handled by their methods. 
Finally, a concurrent work of \cite{gazin2024selecting} (which was posted shortly after the 
first appearance of this work on arXiv) 
also studies the selective inference problem in conformal prediction. 
It specifically considers selection based on properties of prediction sets, 
but jointly calibrates the selection procedure and the construction of 
prediction sets while achieving FCR control. 
In contrast, we take a complementary perspective, decoupling the 
selection and uncertainty quantification steps, 
which allows the development of a versatile framework that is applicable to 
a broader range of selection rules.

More broadly, many recent works have studied selective inference 
problems in conformal prediction, 
mostly focusing on multiple testing, 
i.e., selecting individuals that obey some properties with 
FDR control.
This includes selecting outliers~\citep{bates2023testing,marandon2022machine,liang2022integrative,bashari2024derandomized} 
and selecting high-quality data~\citep{jin2022selection,jin2023model}. 
The construction of conformal prediction sets is an orthogonal direction, 
and our methods also work for selection based on such procedures; 
see Section~\ref{sec:conf_pval} for details.

Under a similar name, \citet{sarkar2023post} propose ``post-selection inference 
for conformal prediction'', 
but they mainly focus on selecting the confidence level $\alpha\in(0,1)$  
instead of selecting the units. 
In addition, they are closer to the ``simultaneous inference'' strand in POSI 
since they aim for valid inference  simultaneously for all confidence levels, 
other than conditional on a selection event. 
For these reasons, our methods and guarantees are quite different. 

\subsection{Connection to Mondrian Conformal Prediction}
\label{app:MCP}
Our framework is closely related to 
Mondrian Conformal Prediction (\textnormal{MCP})~\citep{Vovk2003MondrianCM} which 
applies to case of one test unit ($m=1$). 
Given $K$ categories $[K]$, MCP assumes access to a classifier 
$\cK$ that gives the class label for all units: 
$ 
(\kappa_1,\kappa_2,\dots,\kappa_{n+1}) = \cK(Z_1,Z_2,\dots,Z_{n+1}),
$ 
where $\kappa_i \in [K]$ is the class label for unit $i$.
The classifier $\cK$ is required to be equivariant to any permutation $\pi$ of $[n+1]$, i.e., 
$ 
(\kappa_{\pi(1)},\kappa_{\pi(2)},\dots,\kappa_{\pi(n+1)})
 = \cK(Z_{\pi(1)},Z_{\pi(2)},\dots,Z_{\pi(n+1)}).
$
The goal of \textnormal{MCP}  is to construct a prediction set for a test unit
with category-wise coverage, i.e., 
$\PP(Y_{n+1}\in \hat{C}^{\textnormal{MCP}}_{\alpha,n+1}\given \kappa_{n+1}=k)$ for each $k\in [K]$. 
This is achieved by finding  a subset of calibration data 
that are categorized as $k$ 
for each hypothesized test outcome $y\in \cY$, i.e.,
\@ \label{eq:mcp_ref_set}
\tilde{\cR}(y) = \{i\in \cic: \cK(\cD_{\calib},(X_{n+1},y))_{[i]} = k\},
\@ 
and then constructing the prediction set as
\$ 
\hat{C}^{\textnormal{MCP}}_{\alpha,n+1}= \Big\{y \in \cY: 
V(X_{n+1},y) \le \quant\big(1-\alpha; \{V_i\}_{i\in \tilde{\cR}(y)}\cup \{\infty\}\big)\Big\}.
\$

While the setups seem different, 
\textnormal{JOMI} effectively reduces to \textnormal{MCP} when 
$\cK$ is a binary classifier
and $\cS$ is equivariant under permutations of $[n+1]$. 
To see this, one can view the selection rule $\cS$ 
as categorizing test points into
two classes (selected or not selected), 
and selection-conditional coverage~\eqref{eq:def_cond_cov_set} as 
coverage given the ``selected'' category. 
In this way, one can check that our reference set $\hcR_{n+1}(y)$ coincides with 
$\tilde{\cR}(y)$ defined in~\eqref{eq:mcp_ref_set}, 
and thus the prediction sets also coincide.

With a distinct goal of selection-conditional coverage, 
we additionally address the challenge that 
realistic selection rules are often asymmetric to calibration 
and (multiple) test points. To this end, \textnormal{JOMI} generalizes \textnormal{MCP} 
via a more delicate framework for finding exchangeable subgroups. 
As we jointly consider the calibration and test points when 
defining the reference set, we call our method 
JOint Mondrian Conformal Inference (\textnormal{JOMI}). 

\subsection{Comparison to the Benjamini-Yekutieli (BY) correction}
\label{app:BY}

For post-selection inference, the Benjamini-Yekutieli (BY) procedure~\citep{benjamini2005false} is a natural solution, which is argued in prior works~\citep{bao2024selective,gazin2024selecting} to heuristically control the FCR in practice. In the conformal prediction context, the BY prediction set is of the form 
\$
\hat{C}_{\alpha,n+j}^{\textrm{BY}} = \Big\{y\in \cY\colon V(X_{n+j},y)\leq \text{Quantile}\big(1-\alpha |\hat\cS|/m; \{V_i\}_{i\in \cI_{\calib}}) \cup\{\infty\} \Big\}.
\$
which essentially produces conformal prediction sets at an adjusted level $\tilde\alpha = \alpha|\hat\cS|/m$. 

As we demonstrate in our numerical experiments, the BY procedure indeed offers valid selection-conditional coverage empirically (though no formal theory is established in general). However, they are overly conservative in many cases, producing non-informative sets $\{0,1\}$ in binary classification problems, and prediction intervals with twice the lengths of JOMI. These results highlight the importance of leveraging the exchangeability structure unique in predictive inference for efficient inference.  

Apart from conservativeness, an issue with the BY prediction set is that it does not adapt to the nature of the selection. Note that given the covariate value of a test unit, $\hat{C}_{\alpha,n+j}^{\textrm{BY}}$ depends on $\hat\cS$ only through its cardinality. Take top-K selection in a regression problem with $V(x,y) = |y-\hat\mu(x)|$ as an example. In this case, no matter whether we select  $K$ units with the highest predictions or the lowest predictions, the prediction set is always of the form $[\hat\mu(x)-\hat\eta,\hat\mu(x)+\hat\eta]$ with $\hat\eta$ being the $(1-\tilde\alpha)$-th quantile of all calibration scores. In contrast, by searching for ``similar'' calibration data in the reference set, our method provides more precise and targeted uncertainty quantification.

\section{Additional results on covariate-dependent selection}
\label{app:cov_dep}
This subsection is devoted to stating the tailored solutions under covariate-dependent rules (1)-(3)
in Section~\ref{sec:covariate} of the main text. 
Throughout, we write $S_i=S(X_i)$ 
and let $\cT = \{S_i\colon i\in [n+m]\}$.

\subsection{Top-K selection} 
In this example, we assume that $S_{n+j}$'s are distinct for all $j\in[m]$ almost 
surely for well-defined choice of top-K units. 
The top-K selection set is equivalently 
\$
\hcS_\topk = \{j\in[m]: S_{n+j} > T_\topk\},\quad \text{where}\quad 
T_\topk = \inf\Big\{t \in \cT: \textstyle{\sum_{j\in[m]}} \ind\{S_{n+j}\le t\}  \ge m-K\Big\}.
\$
In this case, the reference set for each $j \in \hcS_\topk$ takes a unified form:
\@\label{eq:topK_ref_set}
\hcR_\topk \,:=\,  \{i \in \cic: S_i > T_\topk\},
\@
which is also agnostic to the choice of $\mfS$.
Intuitively, $T_\topk$ refers to the $(m-K)$-th smallest score among the test units; 
when $S_{n+j}$ is among the top $K$, replacing $S_{n+j}$ with another 
score above $T_\topk$ does not change the  selection outcome.
This intuition is formalized in the following proposition, with its proof deferred to 
Supplementary Section~\ref{appd:prood_of_topk}.

\begin{proposition}
\label{prop:topK}
With the top-K selection rule, 
for any $j\in[m]$ and any selection taxonomy $\mfS \in 2^{[m]}$ 
such that $j \in \hcS_\topk$ and $\hcS_\topk \in \mfS$, 
the reference set obeys $\hcR_{n+j}(y) = \hcR_\topk$ for all $y\in \cY$.
\end{proposition}

Computing $\hcR_\topk$ as a universal reference set 
reduces the computation complexity of 
the reference construction step in Algorithm~\ref{alg:PI_cond_cov} to $O(\max(m, n))$.  

\subsection{Selection based on joint quantiles}

When selected units are those whose scores are above the 
$q$-th quantile of calibration and test scores, 
the selection set can be written as 
\$ 
\hcS_{\jq} = \{j \in [m]: S_{n+j} > T_\jq\},\quad  \text{where}\quad 
T_\jq = \inf\Big\{t \in \cT: 
\textstyle{\sum_{i=1}^{m+n}} \ind\{S_i \le t\}\ge q(n+m) \Big\} 
\$ 
By definition, the selection threshold 
$T_\jq$ is invariant to the permutation of the calibration and test scores, and 
it is straightforward to check that for all $j\in \hcS_\jq$, the reference set is  
\$ 
\hcR_{\jq} = \{i \in \cic: S_i > T_\jq\}.
\$
Replacing $\hcR_{n+j}$'s with $\hcR_{jq}$ reduces the computation complexity of
the reference set in Algorithm~\ref{alg:PI_cond_cov} to $O(\max(m, n))$.  

\subsubsection{Selection based on calibration quantiles}  
When selected units are those whose scores are greater than the $q$-th quantile 
of the calibration scores, 
the selection set is equivalently
\@ \label{eq:cq_selection}
\hcS_{\cq} = \big\{j\in [m]: S_{n+j} > T_{\cq} \big\},\quad \text{ where}\quad   
T_{\cq} = \inf\Big\{t \in \cT: \textstyle{\sum_{i=1}^n} \ind\{S_i \le t\} \ge qn \Big\}.
\@
For each $j \in \hcS_\cq$, the reference set consists of all calibration units in the top $q$-th quantile, i.e., 
\@ 
\hcR_{\cq} = \{i\in \cic: S_i > T_\cq\}.
\@
The validity of $\hcR_\cq$ as the reference set is established below, 
whose proof is deferred to Supplementary Section~\ref{appd:proof_cq}.

\vspace{0.5em}
\begin{proposition}
\label{prop:cq}
For the selection rule specified in~\eqref{eq:cq_selection}, 
for any $j\in[m]$ and any selection taxonomy $\mfS$ 
such that $j\in \hcS$ and $\hcS \in \mfS$,
it holds that  $\hcR_\cq = \hcR_{n+j}(y)$ defined 
in~\eqref{eq:def_hat_Rj} for all  $y\in \cY$. 
\end{proposition}
\vspace{0.5em}

Similar to the previous two cases, 
directly computing $\hcR_{\cq}$ 
reduces the computation complexity of the reference set to $O(\max(m, n))$.
As a side note,~\citet{bao2024selective} also considered the selection rule based on 
calibration quantiles (as an instance of ``calibration-assisted selection''), for 
which they provide asymptotically FCR control guarantees under additional assumptions.

\section{Additional experiment results}
\label{app:experiments}

This section collects additional experimental results that are 
omitted in Section~\ref{sec:drug} of the main text. 
Supplementary Section~\ref{app:subsec_topK_dpp} shows results for DPP
with selection rules (ii) and (iii) with the APS score.
Supplementary Section~\ref{app:subsec_dpp_binary} shows results for 
DPP with selection rules (i)-(iii) with the binary score. 
Supplementary Section~\ref{app:subsec_DTI} shows results for DTI 
with selection rules (ii) and (iii).

\subsection{Additional results for APS score in drug property prediction}
\label{app:subsec_topK_dpp}
In this part, we present results 
for DPP when test units are selected 
via (ii) top-K in both calibration and test data (Supplementary Section~\ref{subsubsec:topK_dpp_mix}), 
and (iii) only calibration data as reference (Supplementary 
Section~\ref{subsubsec:topK_dpp_calib}). 

\subsubsection{Top-K selection in mixed sample}
\label{subsubsec:topK_dpp_mix}
 
We use the same values of $(\alpha,K)$ and the same scheme of repeated experiments 
as in Section~\ref{subsubsec:dpp-topK-test} of the main text, with $\mathfrak{S} = 2^{[m]}$.
Figure~\ref{fig:dpp_topK_mix} shows that vanilla conformal prediction incurs a 
large coverage gap for selected units, while $\mname$ and its randomized version 
both achieve valid selection-conditional coverage. 
The coverage gap due to discretization is eliminated again by randomization. 
Randomization also reduces the prediction set sizes on the right panel, 
but again, with $\alpha=0.8$, we see the concerning issue that 
it suffices for vanilla CP to construct empty prediction sets 
for selected units to achieve marginal coverage. 

\begin{figure}[htbp]
    \includegraphics[width=\textwidth]{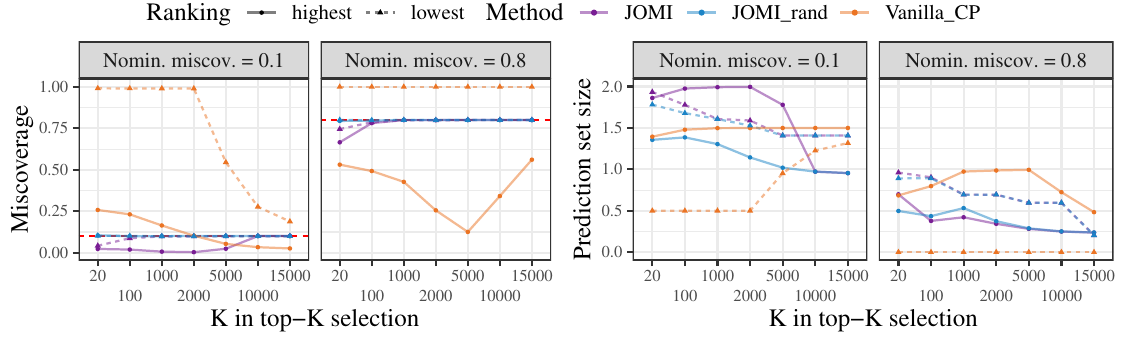}
    \caption{Empirical selection-conditional miscoverage (left) and prediction set size (right) in drug property prediction for test units whose $\hat\mu(X_{n+j})$ are top-K (\texttt{highest}, solid line) or bottom-K (\texttt{lowest}, dashed line) among both calibration and test units. 
    Details are otherwise the same as Figure 3.}
    \label{fig:dpp_topK_mix}
\end{figure}

\vspace{-1em}

\subsubsection{Calibration-referenced selection}
\label{subsubsec:topK_dpp_calib}

Now we select test units whose predicted affinities must be greater (or smaller) 
than the $K$-th largest (or smallest) predictions in the calibration set. 
We use the choices of $(\alpha,K)$ and the scheme of repeated experiments 
as in Section~\ref{subsubsec:dpp-topK-test} of the main text with $\mathfrak{S} = 2^{[m]}$.  

\begin{figure}[H]
    \includegraphics[width=\textwidth]{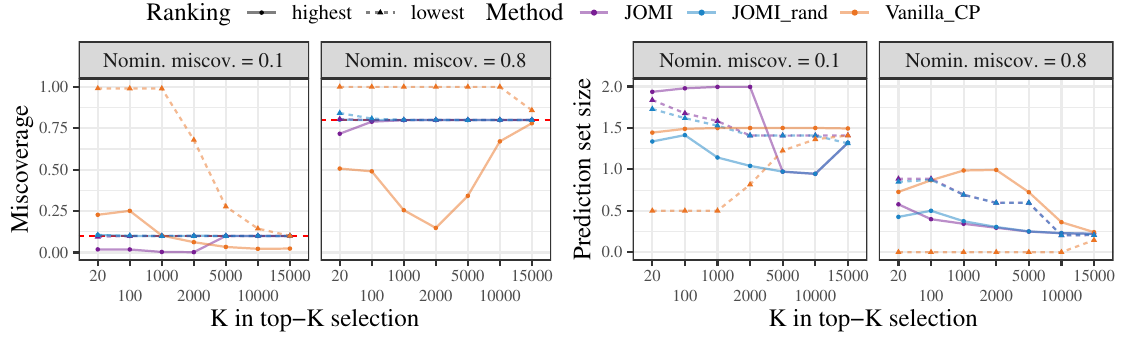}
    \caption{Empirical selection-conditional miscoverage (left) and prediction set size (right) in drug property prediction for test units selected by (iii) whose $\hat\mu(X_{n+j})$ are greater than the $K$-th highest (solid line) or smaller than the $K$-th lowest (dashed line) in calibration units. Details are otherwise 
    as Figure 3.}
    \label{fig:dpp_topK_calib}
\end{figure}

Figure~\ref{fig:dpp_topK_calib} shows that our proposed methods achieve valid 
coverage (exact with \texttt{JOMI\_rand}) while vanilla conformal prediction fails to. 
Again, vanilla CP is either over-confident or under-confident for units with lowest or highest predicted affinities. 

\subsection{Additional results for binary score in drug property prediction}
\label{app:subsec_dpp_binary}

In Figures~\ref{fig:dpp_topK_test_bin},~\ref{fig:dpp_topK_mix_bin}, and~\ref{fig:dpp_topK_calib_bin}, we show the empirical selection-conditional coverage and 
prediction set sizes for selection rules (i)-(iii)
when the prediction sets are constructed with  
the binary score $V(x,y) = y(1-\hat\mu(x))+(1-y)\hat\mu(x)$. 
We observe that while vanilla CP is either over-confident or under-confident, 
our proposed methods achieve valid coverage across all configurations. 
The coverage gap of the non-randomized $\mname$ is smaller with this specific 
score function compared with the APS score.

\begin{figure}[htbp]
    \includegraphics[width=\textwidth]{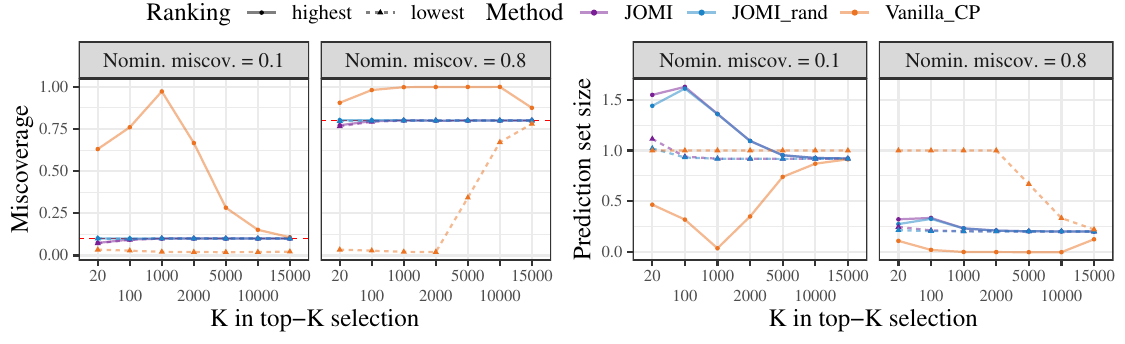}
    \caption{Empirical selection-conditional miscoverage (left) and prediction 
    set size (right) in drug property prediction with the binary score, for test units selected by (i) top-K (or bottom-K) in test data. Details are otherwise as 
    Figure 3.}
    \label{fig:dpp_topK_test_bin}
\end{figure}

\begin{figure}[htbp]
    \includegraphics[width=\textwidth]{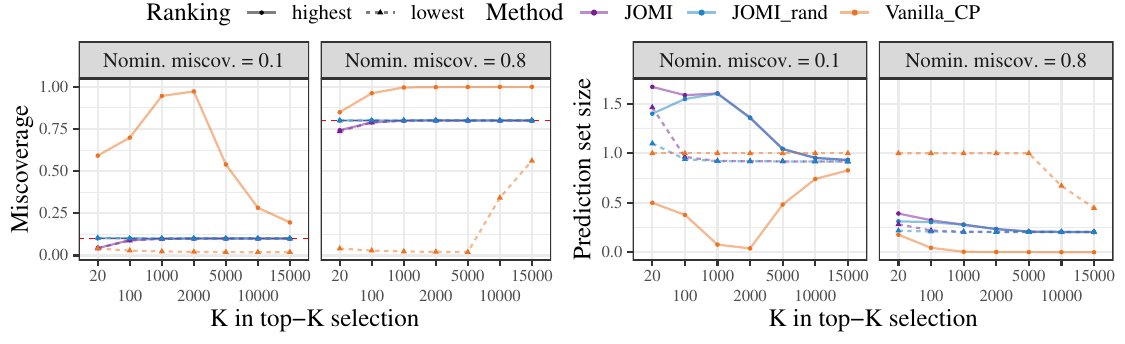}
    \caption{Empirical selection-conditional miscoverage (left) and prediction set size (right) in drug property prediction with the binary score, for test units selected by (ii) top-K (or bottom-K) in both calibration and test data. Details are otherwise as Figure~\ref{fig:dpp_topK_mix}.}
    \label{fig:dpp_topK_mix_bin}
\end{figure}

\begin{figure}[htbp]
    \includegraphics[width=\textwidth]{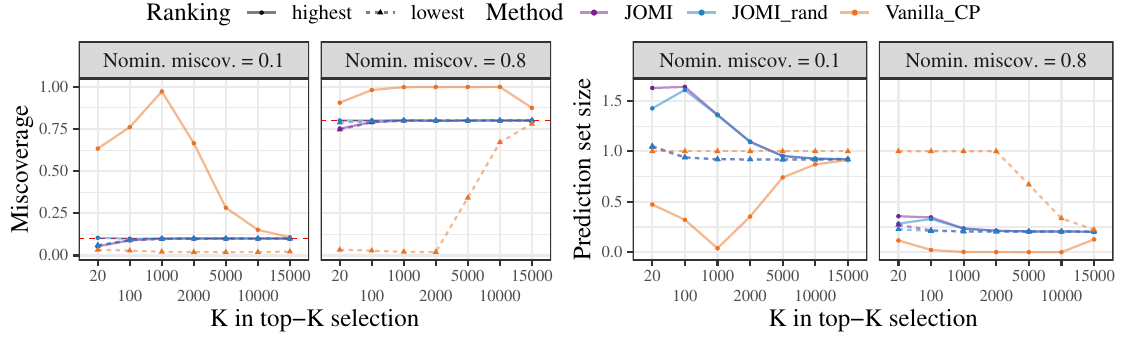}
    \caption{Empirical selection-conditional miscoverage (left) and prediction set size (right) in drug property prediction with the binary score, for test units that are (iii) greater than the $K$-th highest (or smaller than the $K$-th lowest) predicted affinities in the calibration data. Details are as Figure~\ref{fig:dpp_topK_calib}.}
    \label{fig:dpp_topK_calib_bin}
\end{figure}

\subsection{Additional results for empirical FCR in drug property prediction}
\label{app:subsec_DPP_FCR}

In this part, we present empirical FCR results in drug property prediction, 
which complement the three tasks in Section~\ref{sec:drug} in the main text. 

Figure~\ref{fig:dpp_topK_fcr} plots the empirical FCR in DPP in three top-K selection tasks with two conformity scores (APS and binary) at nominal miscoverage level $\alpha=0.1$. Since the selection set size is fixed, by our theory, selection-conditional coverage implies FCR control, as evidence from the results. 

Figure~\ref{fig:dpp_csel_fcr} plots the empirical FCR when the units are selected by the conformal selection algorithm in~\cite{jin2022selection}.
Here, the empirical FCR of JOMI is sill very close to the nominal conditional coverage level, because the selection set size is stable. 

Finally, Figure~\ref{fig:dpp_cons_fcr} plots the empirical FCR when the units are selected by running the constrained optimization program. While the empirical FCR of JOMI is close to the nominal conditional coverage level, the FCR of vanilla CP with both conformity scores is lower than the selection-conditional coverage. This is because the selection set can sometimes be empty, and FCR is considered on this event. This also shows the limitation of FCR as an overly ``optimistic'' error metric: even if the FCR is low, on the ``interesting'' event that the selection set is nonzero, the coverage might still be very low.
    
\begin{figure}[htbp]
    \centering
    \includegraphics[width=0.8\textwidth]{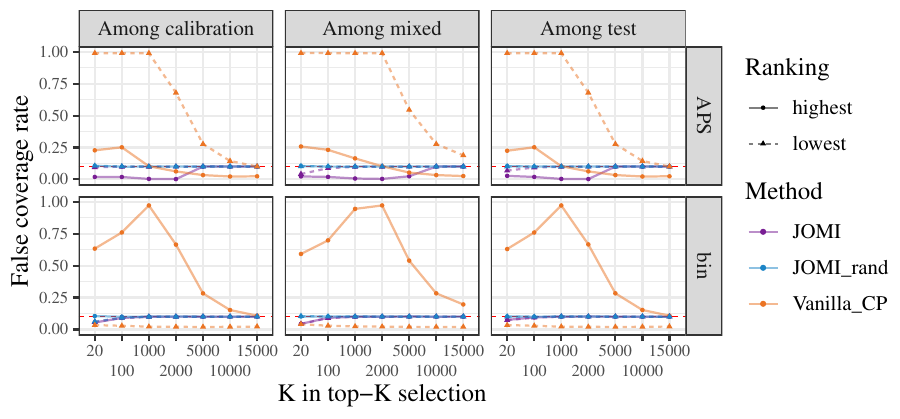}
    \caption{Empirical false coverage rate  across  $N=1000$ independent runs of \texttt{Vanilla\_CP}, \texttt{JOMI}, and \texttt{JOMI\_rand} in top-K selection for drug property prediction. The $x$-axis is the nominal miscoverage level $\alpha=0.1$.}
    \label{fig:dpp_topK_fcr}
\end{figure}

\begin{figure}[htbp]
    \centering
    \includegraphics[width=0.6\textwidth]{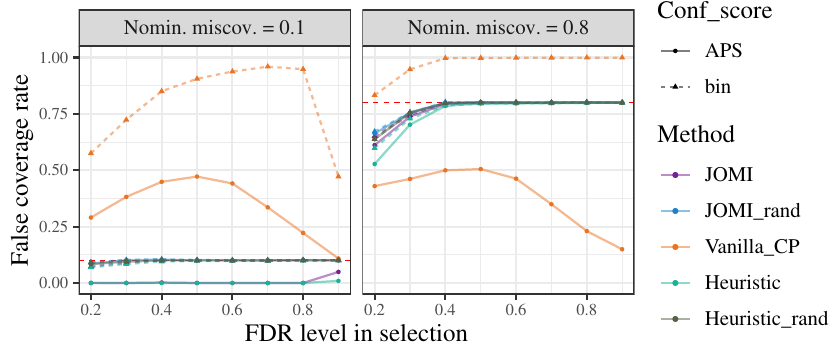}
    \caption{Empirical false coverage rate  across  $N=1000$ independent runs of \texttt{Vanilla\_CP}, \texttt{JOMI}, and \texttt{JOMI\_rand} in drug property prediction when units are selected by FDR-controlled conformal selection. The $x$-axis is the nominal FDR level in the selection algorithm.}
    \label{fig:dpp_csel_fcr}
\end{figure}

\begin{figure}[H]
    \centering
    \includegraphics[width=0.4\textwidth]{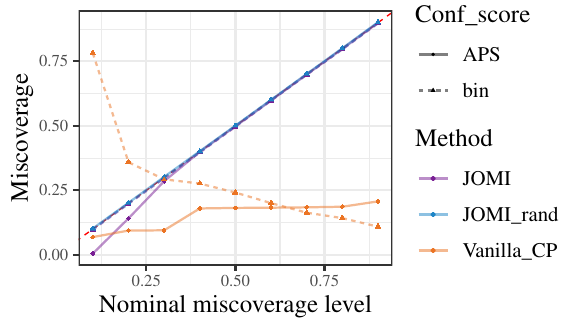}
    \caption{Empirical false coverage rate  across  $N=1000$ independent runs of \texttt{Vanilla\_CP}, \texttt{JOMI}, and \texttt{JOMI\_rand} in drug property prediction when units are selected by running the optimization program. The $x$-axis is the nominal miscoverage level $\alpha\in \{0.1,\dots,0.9\}$.}
    \label{fig:dpp_cons_fcr}
\end{figure}

\subsection{Application to drug-target-interaction prediction}
\label{subsec:DTI}

In this part, we include experiments in drug-target-interaction prediction tasks omitted in the main text.  
We use the DAVIS dataset in the DeepPurpose library~\citep{huang2020deeppurpose}. 
The outcome $Y$ is a continuously-valued variable indicating the 
binding affinity of a drug-target pair. 
The features $X$ combine the encodings of the drug compound 
and the target, where we encode the drug via a convolutional neural network (CNN) structure, and the target via a Transformer encoder. 
Our framework is agnostic to the choice of these encoding methods. 

We train a point predictor $\hat\mu(\cdot)\colon \cX\to \RR$ with $20\%$ randomly selected data in the entire dataset using a small neural network. 
The remaining $80\%$ data are then randomly split into two equally-sized folds 
as $\cD_\calib$ and $\cD_\test$. 
Similar to DPP, the training phase is repeated for $10$ independent runs; 
within each run, we independently run $100$ random split of calibration/test data 
and subsequent selection and construction of prediction sets. 
We use the nonconformity score $V(x,y) = |y-\hat\mu(x)|$ throughout.

\subsubsection{Conformalized selection for heterogeneous $c_{n+j}$'s}

We then study the behavior of vanilla CP and our proposed methods for 
selected units by running conformalized selection~\citep{jin2022selection} to find 
test pairs $j\in [m]$ whose outcomes are greater than 
data-dependent thresholds $c_{n+j}$. Here, we set $c_{n+j}$ as the $0.7$-th quantile 
of activities for training pairs with the same target sequence as the $j$-th data, while controlling the FDR below $q\in \{0.2,0.3,\dots,0.9\}$. 
The method of~\cite{bao2024selective} is not applicable to the heterogeneous thresholds here. 

Figure~\ref{fig:dti_csel} shows the miscoverage (the left panel) 
and average length of prediction intervals (the right panel) using JOMI and vanilla conformal prediction. For those most promising units picked by conformalized selection, 
vanilla CP (orange) tends to be over-confident with 
exceedingly high miscoverage. 
In contrast, 
our proposed methods achieve near-exact coverage across all values of $q$ 
and $\alpha$. Again, we observe that 
the FCR of our methods is even lower than $\hat{\textnormal{Miscov}}$, 
empirically yielding valid FCR control.

\begin{figure}[htbp]
    \includegraphics[width=\textwidth]{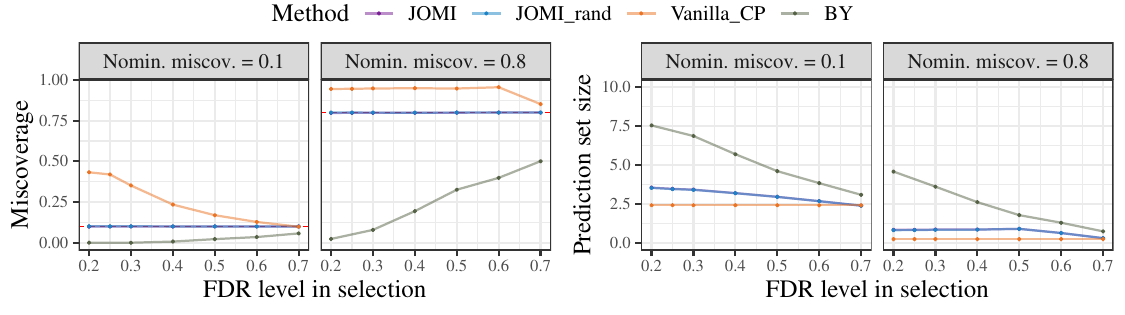}
    \caption{Empirical selection-conditional miscoverage (left) and prediction set size (right) in drug-target-interaction 
    for vanilla conformal prediction (\texttt{Vanilla\_CP}) and JOMI (\texttt{JOMI} and \texttt{JOMI\_rand}), 
    at nominal miscoverage $\alpha\in\{0.1,0.8\}$ for test units selected by conformalized selection at nominal FDR levels $q\in \{0.2,\dots,0.9\}$.}
    \label{fig:dti_csel}
\end{figure}

Our method can sometimes lead to two separate prediction intervals. 
We calculate the 
average number of intervals among selected test units in each experiment 
and report their boxplot in Figure~\ref{fig:dti_csel_interval}.
The frequency of producing two segments is very low across all cases. 

\begin{figure}[htbp] 
    \centering
    \includegraphics[width=0.9\textwidth]{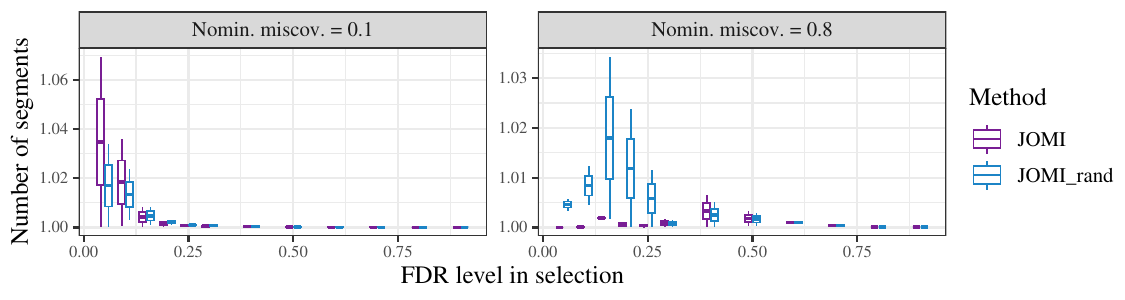}
    \caption{Boxplot of the average number of  intervals produced by \texttt{JOMI} and \texttt{JOMI\_rand} among selected test units, over $N=1000$ independent runs, 
    at nominal miscoverage $\alpha\in\{0.1,0.8\}$ for various experiment configurations.}
    \label{fig:dti_csel_interval}
\end{figure}

\subsubsection{Selection under budget constraint}
\label{subsec:dti_constraint}
Finally, we consider a covariate-dependent selection rule, 
where the scientist aims to maximize the predicted activities of selected drugs 
subject to a fixed budget for 
subsequent development.

Formally, for each 
drug-target pair $i\in [m+n]$, 
we let $c_i$ be the $0.7$-th quantile of true activity scores in the training set 
with the same target, $\hat\mu(X_i)$ be its predicted binding affinity, 
and $L_i>0$ be its development costs. 
Then, $\cS(\cD_\calib,\cD_\test)$ aims to solve the following optimization problem: 
\#\label{eq:knapsack}
\mathop{\text{maximize}}_{S\subseteq [m]}&\quad  
\sum_{j\in S} \big[\hat\mu(X_{n+j}) - c_{n+j}\big] \\ 
\text{subject to}&\quad  \sum_{j\in S} L_{n+j} \leq \bar{L}, \notag
\#
where $\bar{L} = 200$ is a budget limit. Again, as the dataset does not come with 
drug development costs, we generate $L_i = \exp(3\hat\mu(X_i)/\bar\mu) +   |\sin(\hat\mu(X_i))| +\epsilon_i-1+\varepsilon_i$, 
where $\bar{\mu} = \max_{i \in \cdt}|\hat\mu(X_i)|$, and
$\epsilon_i \sim \text{Exp}(1)$, $\varepsilon_i\sim \text{Unif}([0,1])$ are independent random variables.

The optimization problem~\eqref{eq:knapsack} is known as the Knapsack problem which is 
NP-hard. Nevertheless, there are efficient approximate solvers; in our experiments, 
we use the Python package \texttt{mknapsack}~\citep{mknapsack}. 
Existing methods such as~\cite{bao2024selective} cannot 
deal with such a complicated selection process. 
In contrast, our framework tackles this problem with 
a computation complexity that is polynomial in $m$, $n$, 
and the complexity of the subroutine $\cS(\cdot,\cdot)$. 

In Figure~\ref{fig:dti_cons}, 
we report the empirical miscoverage rate, average length of prediction intervals, 
and the sizes of $\cS(\cD_\calib,\cD_\test)$ and the reference set $\hat\cR_j$ 
for all selected units, across $N=1000$ independent runs of our procedures. 
While vanilla conformal prediction is over-confident (with a much higher miscoverage rate), 
our methods yield exact coverage for the focal units. 
From the length of prediction sets, we see that longer prediction intervals are needed to cover 
promising drug-target pairs. 
While the number of selected test units are relatively stable, 
the reference set size may vary a bit. 

\begin{figure}[htbp]
    \centering
    \includegraphics[width=0.85\textwidth]{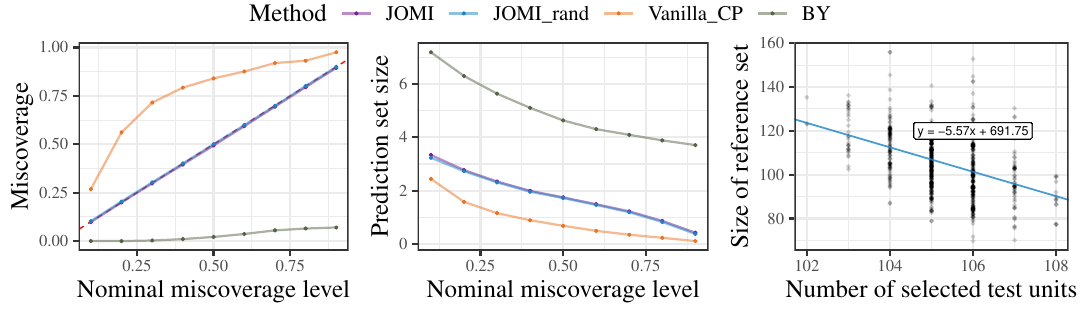}
    \caption{Empirical miscoverage rate (left), average length of prediction interval (middle), 
    and scatter plots for $|\cS(\cD_\calib,\cD_\test)|$ versus $\{|\hat\cR_j|\}_{j\in \cS(\cD_\calib,\cD_\test)}$ (right), across  $N=1000$ independent runs of \texttt{Vanilla\_CP}, \texttt{JOMI}, and \texttt{JOMI\_rand}. The $x$-axis of left and middle plots is the nominal miscoverage level $\alpha\in\{0.1,0.2,\dots,0.9\}$.}
    \label{fig:dti_cons}
\end{figure}

The empirical FCR in the two tasks is reported in Supplementary Section~\ref{app:subsec_DTI_FCR} showing similar messages as in drug property prediction tasks. While JOMI consistently controls the empirical FCR, it can in general be lower than the selection-conditional coverage when the selection set may be empty. 

\subsection{Additional results for drug-target-interaction prediction}
\label{app:subsec_DTI}

In this section, we collect results for DTI 
with (1) \emph{covariate dependent top-K selection}.
We present results under selection rule (1-i) in Supplementary Section~\ref{subsubsec:dti-topK-test},
(1-ii) in Supplementary Section~\ref{app:subsubsec:DIT_topk_mixed}
and (1-iii) in Supplementary Section~\ref{app:subsubsec:DIT_topk_calib}.

\subsubsection{Top-K selection for fixed K}
\label{subsubsec:dti-topK-test}

We consider the case where test units are selected via (i) top-K (or bottom-K) among predicted affinities in the test set. 
We apply the vanilla conformal prediction and our proposed methods 
with $\mathfrak{S} = 2^{[m]}$, 
so that selection-conditional coverage implies FCR control. 
Results under selection rules (ii) and (iii) demonstrate similar patterns, 
which are in Supplementary Section~\ref{app:subsec_DTI}.

Figure~\ref{fig:dti_topK_test} shows $\hat{\textrm{Miscov}}$ (the left panel) and the average length of the prediction intervals (the right panel)
for $K\in \{20,100, 1000, 2000, 5000, 10000\}$ and nominal miscoverage level $\alpha\in\{0.1,0.8\}$. We see that vanilla conformal prediction is over-confident for pairs with the highest predicted affinities, while being under-confident for the lowest. In contrast, our proposed methods achieve near-exact coverage under different nominal levels. 
By comparing the solid and dashed lines in the right panel, 
we observe that in this problem, the intrinsic uncertainty in higher-predicted-affinity pairs is larger, so a longer prediction interval is needed to achieve certain coverage. 

Finally, similar to DPP, BY is overly conservative despite valid coverage. The right panel reveals another drawback of BY: it doesn not adapt to the nature of selection. More specifically, by definition, it only depends on the $\alpha |\hat\cS|/m$-th quantile of calibration scores. As such, it only relies on the size of the selection set, neglecting the nature of the selection (i.e., whether we select highest predictions or lowerst predictions). In both cases, it returns a prediction set of the form $[\hat\mu(X_{n+j})-\hat\eta,\hat\mu(X_{n+j})+\hat\eta]$ where $\hat\eta$ remains the same between two ranking options. 
In contrast, by finding ``similar'' calibration data in the reference set, JOMI provides more precise and targeted uncertainty quantification.

\begin{figure}[htbp]
    \includegraphics[width=\textwidth]{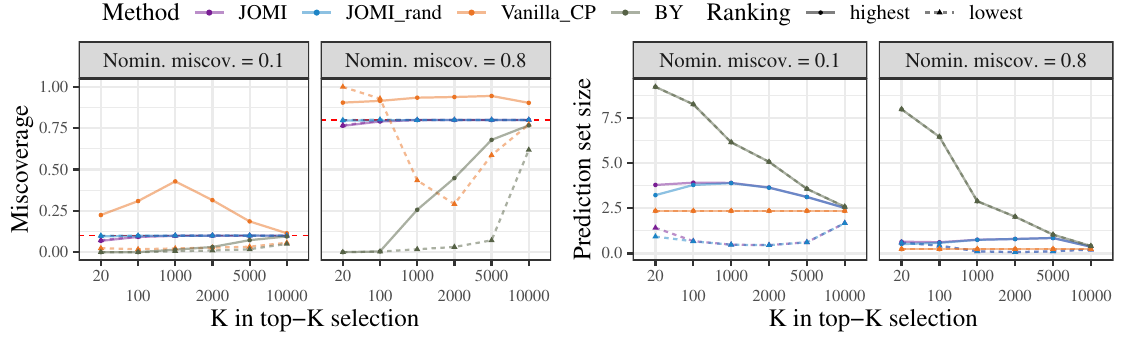}
    \caption{Empirical miscoverage (left) and prediction interval length (right) in drug-target-interaction prediction, for vanilla conformal prediction (\texttt{Vanilla\_CP}), BY (\texttt{BY}), JOMI (\texttt{JOMI}), and randomized JOMI applied to test units whose $\hat\mu(X_{n+j})$ are top-$K$ (\texttt{highest}, solid line) or bottom-$K$ (\texttt{lowest}, dashed line) among test units. The red dashed line is the nominal miscoverage level $\alpha\in\{0.1,0.8\}$.}
    \label{fig:dti_topK_test}
\end{figure}

\subsubsection{Top-K selection in mixed sample}
\label{app:subsubsec:DIT_topk_mixed}

Figure~\ref{fig:dti_topK_mix} shows the miscoverage 
and average length of prediction intervals when (ii) the focal test units 
are those whose predicted affinities are among the $K$ highest or lowest 
in both calibration and test data. We see that vanilla CP 
is over-confident for highest-prediction units, while under-confident for 
lowest-prediction units; BY is overly conservative in both cases. In contrast, our proposed methods always yield valid coverage.  

\begin{figure}[htbp]
    \includegraphics[width=\textwidth]{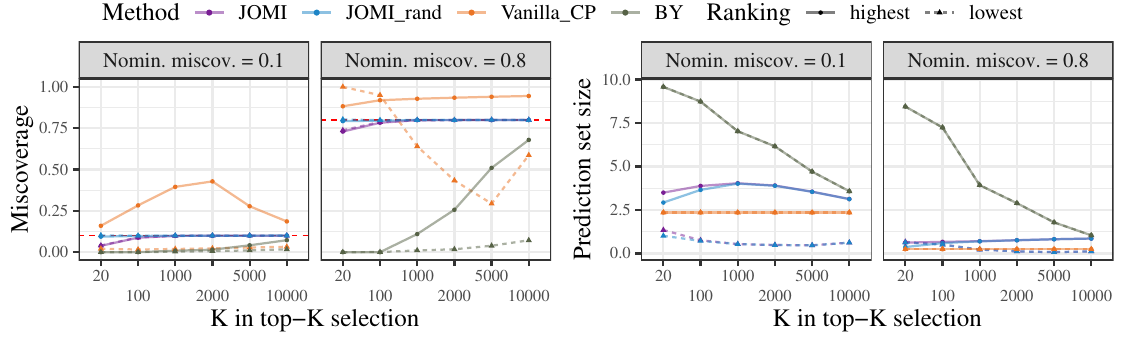} 
    \caption{Empirical miscoverage (left) and prediction interval length (right) in drug-target-interaction prediction, for selected test units whose $\hat\mu(X_{n+j})$ are top-$K$ among {both calibration and test units}. 
    Details are otherwise the same as Figure~\ref{fig:dti_topK_test}.}
    \label{fig:dti_topK_mix}
\end{figure}

\subsubsection{Calibration-referenced selection}
\label{app:subsubsec:DIT_topk_calib}

Figure~\ref{fig:dti_topK_calib} shows the results with (iii) 
calibration-referenced selection, i.e., the selected test units 
are those whose predicted affinities are greater than (or smaller than) 
the $K$-th highest (or lowest) in the calibration data. 
Similar to rules (i) and (ii), vanilla CP is not calibrated for these focal units, 
while our methods yield near-exact coverage across the whole spectrum of $K$. As usual, BY is overly conservative.  

\begin{figure}[H]
    \includegraphics[width=\textwidth]{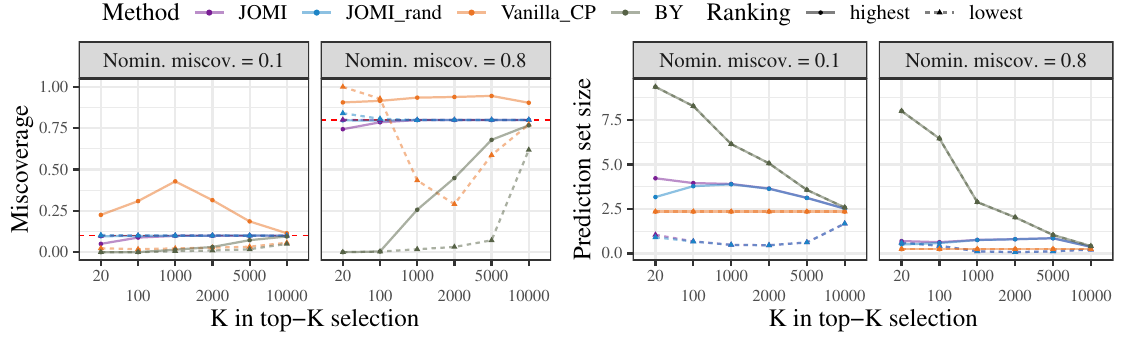}
    \caption{Empirical miscoverage (left) and prediction interval length (right)  in drug-target interaction prediction, for selected test units whose $\hat\mu(X_{n+j})$ surpass the highest $K$ predicted affinities in the calibration data (or smaller than the lowest). Details are otherwise the same as Figure~\ref{fig:dti_topK_test}.}
    \label{fig:dti_topK_calib}
\end{figure}

\subsection{Additional results for empirical FCR in drug target interaction}
\label{app:subsec_DTI_FCR}

    In this part, we collect results for the empirical FCR in drug-target interaction tasks in Section~\ref{subsec:DTI}. 

    Figure~\ref{fig:dti_csel_fcr} plots the empirical FCR when units are selected by running conformalized selection. We observe lower empirical FCR than selection-conditional coverage, demonstrating limitations of FCR as an error metric  similar to the message in Section~\ref{app:subsec_DPP_FCR}. 

    Figure~\ref{fig:dti_cons_fcr} plots the empirical FCR when units are selected by constrained optimization. We observe tight empirical FCR below the nominal level, and in these experiments the selection set is always nonempty.

\begin{figure}[H]
    \centering
    \includegraphics[width=0.5\textwidth]{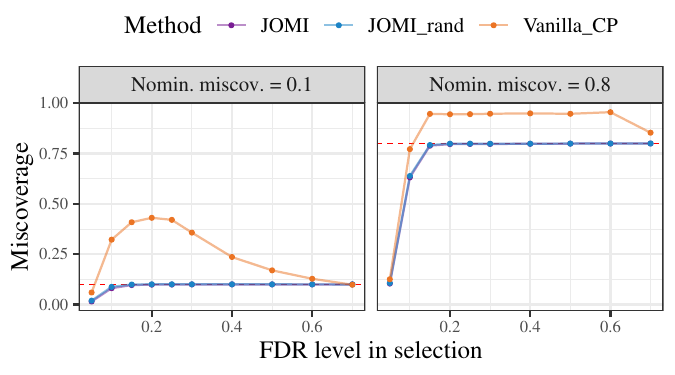}
    \caption{{Empirical FCR in drug-target interaction prediction when units are selected by conformal selection with heterogeneous $c_{n+j}$. The $x$-axis is the nominal FDR level in the selection program.}}
    \label{fig:dti_csel_fcr}
\end{figure}

\begin{figure}[H]
    \centering
    \includegraphics[width=0.45\textwidth]{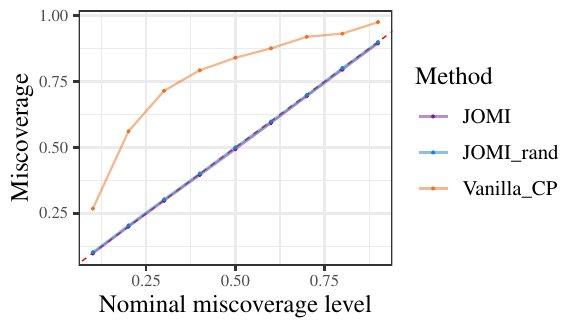}
    \caption{{Empirical FCR in drug-target interaction prediction when units are selected by constrained optimization. The $x$-axis is the nominal miscoverage level $\alpha\in \{0.1,\dots,0.9\}$.}}
    \label{fig:dti_cons_fcr}
\end{figure}

\subsection{Additional results for empirical FCR in health risk prediction}
\label{app:subsec_icu_fcr}

Figure~\ref{fig:icu_fcr} presents the empirical FCR for three tasks in health risk prediction. In these examples, the selection set sizes are stable and nonempty, hence the empirical FCR is also tightly controlled. 

\begin{figure}[H]
    \centering
    \includegraphics[width=0.8\textwidth]{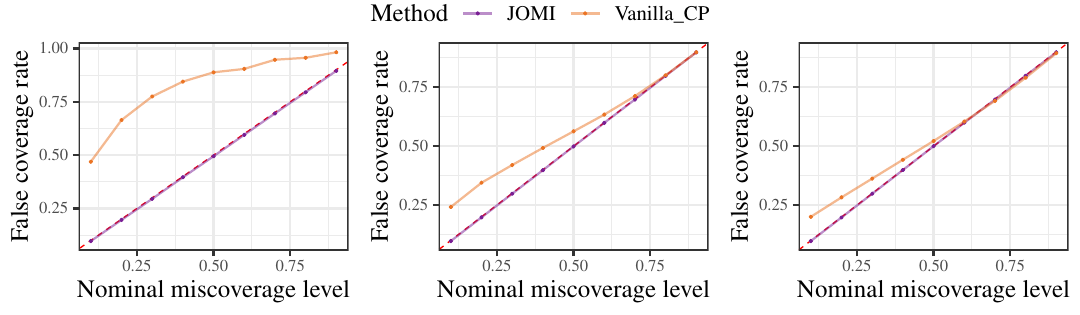}
    \caption{{Empirical FCR in health risk prediction when units are selected by constrained optimization (left), length of preliminary prediction interval (middle), and lower bound of preliminary prediction interval (right). The $x$-axis is the nominal miscoverage level $\alpha\in \{0.1,\dots,0.9\}$.}}
    \label{fig:icu_fcr}
\end{figure}

\section{Technical proofs}

\subsection{Proof of Proposition~\ref{prop:strong_weak}}
\label{app:strong_weak}

Suppose $\hat{C}_{\alpha,n+j}^{(\ell)}$ satisfies $\PP (Y_{n+j} \in \hC_{\alpha,n+j}^{(\ell)} \given 
j\in \hat\cS,~\hat\cS \in \mathfrak{S}_\ell ) \geq 1-\alpha$
for a set of mutually disjoint 
taxonomies $\{\mfS_\ell\}_{\ell\in\cL}$ 
such that $\cup_{\ell\in \cL}\mfS_\ell = 2^{[m]}$, 
and we let $\hat{C}_{\alpha,n+j} = \hat{C}_{\alpha,n+j}^{(\ell)}$ whenever $j\in \hat\cS$ and $\hat\cS\in \mfS_\ell$. This is well-defined because $\mfS_\ell$'s are disjoint. 
Then 
\$
\PP(Y_{n+j}\in \hat{C}_{\alpha,n+j} \given j\in \hat\cS) 
&\stackrel{\textrm{(a)}}{=} 
\sum_{\ell\in \cL}\PP\big(Y_{n+j}\in \hat{C}_{\alpha,n+j}^{(\ell)}, \hcS \in \mfS_\ell \biggiven j\in \hat\cS\big) \\ 
& = \sum_{\ell\in \cL}\PP\big(Y_{n+j}\in \hat{C}_{\alpha,n+j}^{(\ell)} \biggiven j\in \hat\cS,\hat\cS\in \mfS_\ell\big) \cdot \PP(\hat\cS \in \mfS_\ell \given j\in \hat{\cS}) \\ 
& \stackrel{\textrm{(b)}}{\geq} (1-\alpha) \cdot \sum_{\ell\in \cL}\PP(\hat\cS \in \mfS_\ell \given j\in \hat{\cS}) 
= 1-\alpha,
\$
where in the first step we use the fact that  $\{\mfS_\ell\}_{\ell\in\cL}$  are disjoint 
and their union is the full set, and the second step follows from the strong selection-conditional coverage. 
By definition, it satisfies the weak selection-conditional coverage (2).

\subsection{Proof of Proposition~\ref{prop:notions}}
\label{appd:proof_prop_notions} 
Suppose $\hC_{\alpha,n+j}^{(\ell)}$ satisfies the strong 
selection-conditional coverage (3) for a class of mutually disjoint 
taxonomies $\{\mfS_\ell\}_{\ell\in\cL}$ such that $\cup_{\ell\in \cL}\mfS_\ell = 2^{[m]}$, 
and for each $\ell\in \cL$, it holds that 
$\mfS_\ell \subseteq \{S\subseteq[m]\colon |S|=r(\ell)\}$ for some $0\leq r(\ell) \leq m$. 
By definition, and since $\{\mfS_\ell\}_{\ell\in \cL}$ are disjoint, we have 
\@\label{eq:fcr_step1}
\text{FCR} & = 
\EE\bigg[\frac{\sum_{j\in [m]} 
\ind\{j\in \hcS, Y_{n+j}\notin \hC_{\alpha,n+j}\}}{|\hcS|\vee 1} \bigg]\\
&=
\sum_{\ell\in \cL} 
\EE\bigg[\ind\{\hcS\in \mfS_\ell\}\cdot \frac{\sum_{j\in [m]} 
\ind\{j\in \hcS, Y_{n+j}\notin \hC_{\alpha,n+j}^{(\ell)}\}}{|\hcS|\vee 1} \bigg].
\@ 
Noting that $\mfS_\ell \subseteq \{S\subseteq[m]\colon |S|=r(\ell)\}$ hence 
$|\hcS|=r(\ell)$ when $\hcS\in \mfS_\ell$, 
we can rewrite the above as
\@\label{eq:fcr_step2}
& \sum_{\ell\in \cL}  \sum^m_{j=1} \frac{1}{r(\ell) \vee 1}
\EE\big[\ind\big\{\hcS\in \mfS_\ell, j\in \hcS, Y_{n+j} \notin 
\hC_{\alpha,n+j}^{(\ell)} \big\}\big]\notag \\ 
=~&\sum_{\ell\in \cL}\sum^m_{j=1} \frac{1}{r(\ell) \vee 1}
\EE\Big[ \ind\{\hcS\in \mfS_\ell, j\in \hcS\}
\cdot \EE\big[\ind\{Y_{n+j} \notin 
\hC_{\alpha,n+j}^{(\ell)}\} \biggiven \ind\{j\in \hcS, \hcS \in \mfS_\ell\} \big]\Big] \notag\\
=~&\sum_{\ell\in \cL}\sum^m_{j=1} \frac{1}{r(\ell) \vee 1}
\EE\Big[ \ind\{\hcS\in \mfS_\ell, j\in \hcS\}
\cdot \PP\big(Y_{n+j} \notin 
\hC_{\alpha,n+j}^{(\ell)} \biggiven j\in \hcS, \hcS \in \mfS_\ell \big)\Big] \notag\\
\le~& \sum_{\ell\in \cL} \sum^m_{j=1}
 \frac{\alpha}{r(\ell) \vee 1}\EE\big[\ind\{j \in \hcS, \hcS \in \mfS_\ell\}\big], 
\@
where the last inequality follows from the strong selection-conditional coverage.
Finally, the FCR control proof is completed by noting that
\$ 
\eqref{eq:fcr_step2} =  \sum_{\ell\in \cL} \frac{\alpha}{r(\ell) \vee 1} 
\EE\Big[|\hcS| \cdot \ind\{\hcS \in \mfS_\ell\}\Big] 
= \alpha \cdot \PP\big(\hcS\neq \varnothing\big) \le \alpha.
\$

\subsection{Proof of Proposition~\ref{prop:weak_FCR}}
\label{app:proof_weak_FCR}

Consider the following example with $m = 2$ test units and $\cY = \RR$. 
Suppose the target miscoverage level $\alpha \in(0,1)$, and  
the selection rule returns 
\$
\hcS = 
\begin{cases}
    \{1,2\} & \text{with probability } 1-\frac{2\alpha}{1+\alpha},\\
    \{1\} & \text{with probability } \frac{\alpha}{1+\alpha},\\
    \{2\} & \text{with probability } \frac{\alpha}{1+\alpha}.
\end{cases}
\$
Consider prediction sets $\hC_{\alpha,n+1}$ and $\hC_{\alpha,n+2}$ such that 
\begin{itemize}
\item when $\hcS = \{1,2\}$, $\hC_{\alpha,n+1} = \RR, \hC_{\alpha,n+2} = \RR$;
\item when $\hcS = \{1\}$, $\hC_{\alpha,n+1} = \varnothing$;
\item when $\hcS = \{2\}$, $\hC_{\alpha,n+2} = \varnothing$.  
\end{itemize}
We can check that the weak selection-conditional coverage is satisfied for $j=1$:
\$ 
\PP(Y_{n+1} \in \hC_{\alpha,n+1} \given 1 \in \hcS) 
= \frac{\PP(Y_{n+1} \in \hC_\alpha(X_{n+1}), 1\in \hcS)}{\PP(1 \in \hcS)}
= \frac{1 - 2\alpha/(1+\alpha)}{1-2\alpha/(1+\alpha)+\alpha/(1+\alpha)} = 1-\alpha. 
\$ 
By symmetry, the weak selection-conditional coverage is also satisfied for $j=2$.
Meanwhile, one can check that
\$ 
\text{FCR} = \frac{\alpha}{1+\alpha} + \frac{\alpha}{1+\alpha} = \frac{2\alpha}{1+\alpha}
> \alpha.
\$
Therefore, the constructed prediction sets 
$\hC_{\alpha,n+1}$ and $\hC_{\alpha,n+2}$ 
satisfy the weak selection-conditional coverage
but violate the FCR control.

\subsection{Proof of Theorem~\ref{thm:PI_cond_set}}
\label{appd:proof_PI_cond_set}
  
Recall that $\cD_{j} = \cD_{\calib}\cup\{Z_{n+j}\}$ and
$\cD_{j}^c = \{X_{n+\ell}\}_{\ell \in [m] \backslash \{j\}}$. 
For notational convenience, 
we also denote the index sets 
$\cI_j = [n] \cup \{n+j\}$ and $\cI_j^c = \citest \backslash \{ j\}$. 
Let $[\cD_{j}] = [Z_1,\dots,Z_n,Z_{n+j}]$ be the unordered set of $\cD_{j}$, and 
we denote the realized value of the unordered set as $[d_{j}] = [z_1,\dots,z_n,z_{n+j}]$, 
i.e., $z_i$'s are fixed realized values of $Z_i$'s, but we do not know which values in 
$[d_{j}]$ correspond to the calibration set and which to the test point. We also 
write $z_i = (x_i,y_i)$ and $v_i = V(x_i,y_i)$, with $x_i$ and $y_i$ being the 
realized values of $X_i$ and $Y_i$. We then prove the two statements separately.\\

\noindent\textbf{Proof of (a).} 
    To prove~\eqref{eq:def_cond_cov_set_again}, it suffices to show that 
    \@\label{eq:equiv_cond_cov2}
    \PP\Big(  V_{n+j} \leq \quant 
    \big(1-\alpha; \{V_i\}_{i\in \hcR_{n+j}(Y_{n+j})} \cup \{+\infty\} \big) 
    \Biggiven j\in \hat\cS, \, \hcS \in \mfS
    \Big)\geq 1-\alpha.
    \@

    The exchangeability of data in $\cD_j$ given $\cD_j^c$ 
    implies that, given knowledge of the unordered realized data values in $[d_{j}]$, 
    the probability of any assignment of these values to 
    the data points is equal. That is,
    for any permutation $\pi$ of $\cI_j$, 
    \$
    \PP\big(Z_1 = z_{\pi(1)}, \dots, Z_{n} = z_{\pi(n)}, Z_{n+j} = z_{\pi(n+j)} 
    \biggiven [d_{j}], \cD_{j}^c  \big) = \frac{1}{(n+1)!}. 
    \$
    Next, given $[d_j]$ and $\cD_j^c$, {we define the following subset of 
    $\cI_j = [n]\cup \{n+j\}$:}
    \$
    \tilde{R} = \big\{ i\in\cI_j \colon j \in \cS([D_\calib^i], D_\test^i),\, 
    \cS([D_\calib^i],D_\test^i) \in \mfS\big\},
    \$
    where $[D_\calib^i] = [z_1,\dots,z_{i-1},z_{n+j},z_{i+1},\dots,z_n]$ 
    is an unordered set ($[d_j]$ with $z_i$ removed) 
    and  $D_\test^i = (X_{n+1},\dots, x_i,\dots,X_{n+m})$, 
    i.e., the test data with $X_{n+j}$ replaced by $x_i$. 
    Also, $[D_\calib^{n+j}]= [z_1,\dots,z_n]$.
    
    {We note that a set is always unordered. 
    By definition, the \emph{unordered} collection of $\{\cS([D_\calib^i], D_\test^i)\}_{i\in \cI_j}$ 
    is fully determined by $[d_j]$ and $\cD_j^c$. That is, 
    given $[d_j]$ and $\cD_j^c$, no matter which value  in $[z_1,\dots,z_{i-1},z_{n+j},z_{i+1},\dots,z_n]$  the test point $Z_{n+j}$ takes on, 
    the unordered set of  $\{\cS([D_\calib^i], D_\test^i)\}_{i\in \cI_j}$ remains the same. 
    As a result,  $\tilde{R}$ is also fully determined by $[d_j]$ and $\cD_j^c$. In addition, the set of scores $\{v_i\colon i\in \tilde{R}\}$ is fully determined by $[d_j]$ and $\cD_j^c$, which remains the same no matter which value  in $[z_1,\dots,z_{i-1},z_{n+j},z_{i+1},\dots,z_n]$  the test point $Z_{n+j}$ takes on.}  

    Recall that $\hat\cR_{n+j}(Y_{n+j}) = 
    \{i \in \cic: j \in \hcS^{\swap{i}{j}}(Y_{n+j}),\, \hcS^{\swap{i}{j}}(Y_{n+j})\in \mfS\}$. 
    We let $\hat\cR^+_{n+j} = \hat\cR_{n+j}(Y_{n+j})$ if $j\notin \hat\cS$ and $\hcS \in \mfS$ 
    and $\hat\cR^+_{n+j} = \hat\cR_{n+j}(Y_{n+j}) \cup\{n+j\}$ otherwise.

    We claim that given $[d_j]$ and $\cD_j^c$, 
    \@\label{eq:equiv_ref}
    \{V_i\colon i \in \hat\cR_{n+j}^+ \} = \{v_i\colon i \in \tilde{R}\},
    \@
    where repetition of elements is allowed on both sides. 
    {In words,~\eqref{eq:equiv_ref} means that no matter which value in $[z_1,\dots,z_n,z_{n+j}]$ is the test point $Z_{n+j}$, the left-handed side always equals $\{v_i\colon i\in \tilde{R}\}$, which is determined by $[d_j]$ and $\cD_j^c$.}
       To see this, let $\pi$ denote the permutation on $\cI_j$ such that 
       $V_i  = v_{\pi(i)}$ and $V_{n+j} = v_{\pi(n+j)}$. On this event, 
       \begin{equation}
       \begin{aligned}\label{eq:representation}
        &\hat\cS = \cS\big( \underbrace{[z_{\pi(1)},z_{\pi(2)},\dots,z_{\pi(n)}]}_{\textrm{calibration data}}, 
                           \underbrace{X_{n+1},\dots,x_{\pi(n+j)}, \dots, X_{n+m}}_{\textrm{test data}}    \big)
                           = \cS\big([D_\calib^{\pi(n+1)}], D_\test^{\pi(n+1)}\big), \\ 
        &\hat\cS^{\swap{i}{j}} = \cS\big( 
            \underbrace{[z_{\pi(1)},\dots,z_{\pi(i-1)},z_{\pi(n+j)},z_{\pi(i+1)},\dots,z_{\pi(n)}]}_{\textrm{calibration data after swap}}, 
            \underbrace{X_{n+1},\cdots, x_{\pi(i)}, \dots,X_{n+m}}_{\textrm{test data after swap}}    \big)\\
        & \qquad \qquad  
        = \cS\big( [D_\calib^{\pi(i)}], D_\test^{\pi(i)}\big)
       \end{aligned}
    \end{equation}
       for all $i\in\{1,\dots,n\}$.  
       Denoting $\hat\cS = \hat\cS^{\swap{n+j}{j}}(Y_{n+j})$ for notational simplicity, 
       we have,
       \$
       \{V_i\colon i \in \hat\cR_{n+j}^+ \} 
       &\step{a}{=} \big\{ V_i \colon i\in \cI_j,
       j \in \hat\cS^{\swap{i}{j}}(Y_{n+j}),\, 
       \hcS^{\swap{i}{j}}(Y_{n+j}) \in \mfS\big\} \\ 
       &\step{b}{=} \big\{ v_{\pi(i)} \colon i \in \cI_j,
       j\in \hat\cS^{\swap{i}{j}}(Y_{n+j}),\,
       \hcS^\swap{i}{j}(Y_{n+j}) \in \mfS \big\} \\ 
       &\step{c}{=} \big\{ v_{\pi(i)} \colon i\in \cI_j,
       j \in  \cS\big([D_{\calib}^{\pi(i)}], D_{\test}^{\pi(i)}\big),\,
       \cS\big([D_\calib^{\pi(i)}], D_\test^{\pi(i)} \big) \in \mfS \big\} \\ 
       &= \big\{ v_{\ell} \colon \ell\in \cI_j,
       j\in  \cS\big([D_{\calib}^{\ell}], D_{\test}^{\ell}\big),\, 
       \cS\big([D_\calib^{\ell}], D_\test^{\ell} \big) \in \mfS \big\} \\ 
       &\step{d}{=} \{v_i\colon i\in \tilde{R}\}.
       \$
       Above, step (a) follows from the definition of $\hat\cR_{n+j}^+$;
       steps (b) and (c) are due to~\eqref{eq:representation}, 
       and step (d) follows from the definition of $\tR$. 
       We thus completed the proof of Equation~\eqref{eq:equiv_ref}.

    Returning to the proof of~\eqref{eq:equiv_cond_cov2}, we have  
    that 
    \begin{equation}
        \begin{aligned}
    \label{eq:equiv_cond_cov3}
    &\PP\Big(V_{n+j} \leq \quant \big(1-\alpha; \{V_i\}_{i\in \hcR_{n+j}(Y_{n+j})}\cup \{+\infty\}  \big)
     \Biggiven j\in \hat\cS,~\hcS \in \mfS,~[d_j], \cD_j^c\Big)  \\ 
    = ~&\PP\Big(V_{n+j} \leq \quant \big(1-\alpha; \{V_i\}_{i\in \hcR_{n+j}^+(Y_{n+j})}  \big)
     \Biggiven j\in \hat\cS,~\hcS \in \mfS,~[d_j], \cD_j^c\Big)  \\ 
    =~&\frac{\PP\big(V_{n+j} \leq \quant \big(1-\alpha; \{V_i\}_{i\in \hcR_{n+j}^+(Y_{n+j})}\big), 
    j\in \hat\cS, 
    \hcS \in \mfS\biggiven [d_j], \cD_{j}^c\big)}
    {\PP\big(  j\in \hat\cS, \hcS \in \mfS \biggiven [d_j], \cD_{j}^c\big)},         
        \end{aligned}
    \end{equation}
    where the last step is due to Bayes' rule.
    Again, given $[d_j]$ and $\cD_{j}^c$, the only randomness is the assignment of 
    $[z_1,\dots,z_n,z_{n+j}]$ to $(Z_1,\dots,Z_{n},Z_{n+j})$. 
    By~\eqref{eq:equiv_ref}, 
    \@
    &\PP\Big( V_{n+j} \leq \quant \Big(1-\alpha; \{V_i\}_{i\in \cR_{n+j}^+(Y_{n+j})} \Big), 
    j\in \hat\cS, 
    \hcS \in \mfS \Biggiven [d_j], \cD_j^c\Big) \notag \\ 
    =~&\PP\Big(  V_{n+j} \leq \quant \big(1-\alpha; \{v_i  \colon i\in\tilde{R}\} \big), 
    j\in \hat\cS,
    \hcS \in \mfS
    \Biggiven [d_j], \cD_{j}^c\Big)\\
    =~& \sum_{\ell \in \cI_j} \PP\big(V_{n+j} = v_\ell, v_\ell 
    \leq \quant \big(1-\alpha; \{v_i  \colon i\in\tilde{R}\} 
    \big), \ell \in \tR  \biggiven [d_j], \cD^c_j\big). 
    \label{eq:prob_numer}
    \@
    Since $\tR$ is deterministic conditional on $[d_j]$ and $\cD_j^c$, 
    we know that 
    \@
    \eqref{eq:prob_numer}  
    =~& \sum_{\ell \in \tR} \PP\big(V_{n+1} = v_\ell \biggiven [d_j],\cD_j^c\big)
    \cdot \ind\{v_\ell \leq \quant \big(1-\alpha; \{v_i  \colon i\in\tilde{R}\}\}\\
\label{eq:numer_deter}
    = ~&\sum_{\ell \in \tR} 
    \frac{\ind\{v_\ell \leq \quant \big(1-\alpha; \{v_i  \colon i\in\tilde{R}\}\} }{n+1}, 
    \@
    where the last step is because $\PP(Z_{n+1}=z_\ell\given[D_{n+1}], \cD_{n+1}^c ) = 1/(n+1)$ for every $\ell$. 
    By the definition of the quantile function, 
    \$
    \eqref{eq:numer_deter} \geq \frac{(1-\alpha)|\tilde{R}|}{n+1}.
    \$
    On the other hand, it is straightforward that 
    \@
    \PP\big( j\in \hat\cS, \hcS \in \mfS
     \biggiven [d_j], \cD_{j}^c\big) 
    & = \sum_{i\in \cI_j}\PP \big( Z_{n+j} = z_i\biggiven 
      [d_j], \cD_{j}^c \big) \cdot 
    \ind\{i \in \tR\}\\
    \label{eq:prob_denom}
    & = \sum_{i\in \cI_j}
    \frac{\ind\{i\in \tilde{R}\}}{n+1} = \frac{|\tilde{R}|}{n+1}.
    \@
    Combining the above and~\eqref{eq:prob_denom}, we have thus verified~\eqref{eq:equiv_cond_cov3}. 
    Then by the tower property, we conclude~\eqref{eq:equiv_cond_cov2}, 
    and thus (11).
     
    When the elements in $\tR$ are distinct, we additionally have the following 
    deterministic inequality following from the definition of the quantile function:
    \$ 
    \eqref{eq:numer_deter} \le \frac{(1 - \alpha)|\tR| +1}{n+1}, 
    \$ 
    which then leads to 
    \$
    &\PP\Big(V_{n+j} \leq \quant \Big(1-\alpha; \{V_i\}_{i\in \cR_{n+j}^+(Y_{n+j})} \Big)
    \Biggiven j\in \hat\cS, 
    \hcS \in \mfS, [d_j], \cD_j^c\Big) \notag \\
    \le~& 1-\alpha + \frac{1}{|\tilde{R}|} = 1 - \alpha + \frac{1}{|\hcR_{n+j}^+(Y_{n+j})|}
    = 1- \alpha + \frac{1}{1+|\hcR_{n+j}(Y_{n+j})|}, 
    \$
    where the last step is because we are conditioning on the event $j\in \hat\cS$.
    Taking expectation with respect to $[d_j]$ and $\cD_j^c$ 
    on both sides, we obtain (12) and 
    complete the proof of statement (a).\\

\noindent\textbf{Proof of (b).} 
To conveniently denote repeated values in $[z_1,\dots,z_n,z_{n+j}]$
 and $[v_1,\dots,v_{n},v_{n+j}]$  
we let $[1],[2],\dots,[n],[n+j]$ be a permutation of $\cI_j$ so that 
$v_{[1]}\leq v_{[2]}\leq \cdots \leq v_{[n]} \leq v_{[n+j]}$. 
Then $[d_j]$ can be equivalently represented by $\{z_{[i]}\}_{i \in \cI_j}$, and
for any $i \in \cI_j$,  
\$
\PP\big(Z_{n+j} = z_{[i]} \biggiven [d_j],\cD_{j}^c\big) = \frac{1}{n+1}.
\$
Following the proof of part (a), we know that 
\$
\PP\big( j \in \hat\cS \biggiven [d_j], \cD_j^c\big) = \frac{|\tilde{R}|}{n+1}.
\$
On the other hand, due to~\eqref{eq:equiv_ref}
we have for any $t \in (0,1)$, 
\$
&~\PP\bigg( \frac{ \sum_{i \in \hcR_{n+j}(Y_{n+j})} \ind\{V_{n+j} < V_i\} + 
U_{j} (1+ \sum_{i\in \hcR_{n+j}(Y_{n+j})} 
\ind\{V_{n+j} = V_i\}) }{ 1 + |\hcR_{n+j}(Y_{n+j})| } \leq t,
~ j \in \hat\cS, ~\hcS \in \mfS  \Biggiven  [d_j], \cD_{j}^c \bigg) \\ 
=&~ \PP\bigg( \frac{ \sum_{i \in \hcR^+_{n+j}(Y_{n+j})} \ind\{V_{n+j} < V_i\} + 
U_{j} (\sum_{i\in \hcR_{n+j}^+(Y_{n+j})} 
\ind\{V_{n+j} = V_i\}) }{|\hcR^+_{n+j}(Y_{n+j})| } \leq t,
~ j \in \hat\cS, ~\hcS \in \mfS  \Biggiven  [d_j], \cD_{j}^c \bigg) \\ 
=&~\PP\bigg( \frac{ \sum_{i \in \tilde{R}} 
\ind\{V_{n+j} < v_{[i]}\} + U_{j}  \sum_{i\in \tilde{R}} 
\ind\{V_{n+j} = v_{[i]}\} }{ |\tilde{R}| } 
\leq t,~ j \in \hat\cS, \hcS \in \mfS  \Biggiven  [d_j], \cD_j^c \bigg) \\ 
=&~\sum_{\ell\in \cI_j}
\PP\bigg(Z_{n+j}=z_{[\ell]}, \frac{ \sum_{i \in \tilde{R}} \ind\{V_{n+j} < v_{[i]}\} 
+ U_j   \sum_{i\in \tilde{R}} \ind\{V_{n+j} = v_{[i]}\} }{ |\tilde{R}| } \leq t,
~ j \in \hat\cS, \hcS \in \mfS  \Biggiven  [d_j], \cD_j^c \bigg),
\$
where  
the last equality is a decomposition of disjoint events. 
When $Z_{n+j} = z_{[\ell]}$, note that $j\in \hat\cS$ and 
$\hcS \in \mfS$ is equivalent to $[\ell]\in \tilde{R}$. Thus, 
the last equation above equals to
\$
&\sum_{\ell \in \cI_j}
 \PP\bigg( Z_{n+j} = z_{[\ell]}, 
\frac{ \sum_{i \in \tilde{R}} \ind\{v_{[\ell]} < v_{[i]}\} +
 U_j \cdot   \sum_{i\in \tilde{R}} \ind\{v_{[\ell]} = v_{[i]}\} }
{ |\tilde{R}| } \leq t,~ [\ell] \in \tilde{R}  \Biggiven  [d_j], \cD_{j}^c \bigg) \\ 
=~&\frac{1}{n+1} \sum_{\ell \in \cI_j} 
\ind\big\{[\ell]\in \tilde{R}\big\} \cdot \PP\bigg(  \frac{ \sum_{[i] \in \tilde{R}} 
\ind\{v_{[\ell]} < v_{[i]}\} + U_j \cdot   \sum_{[i]\in \tilde{R}} \ind\{v_{[\ell]} = v_{[i]}\} }{ |\tilde{R}| } \leq t 
\Biggiven  [d_j], \cD_{j}^c\bigg) \\ 
=~& \frac{t|\tilde{R}|}{n+1},
\$
where we used the take-out property of conditional probability/expectation in the 
first step;
the probability in the second line is only with respect to $U_j\sim \textnormal{Unif}[0,1]$, 
and the last equality is a standard result that is straightforward to check with exchangeability
(see, e.g.,~\citet[Proposition 2.4]{vovk2005algorithmic}). Putting them together and 
letting $t=1-\alpha$, we have  
\$
\PP\big(Y_{n+j} \in \hC_{\alpha,n+j}^\rand  \biggiven j \in \hat\cS, [d_j], \cD_{d}^c\big)
= \frac{\PP\big( Y_{n+j} \in \hC_{\alpha,n+j},~j \in \hat\cS \biggiven  [d_j], \cD_{j}^c\big)}
{\PP\big(  j \in \hat\cS \biggiven [d_{j}], 
\cD_{j}^c\big)} = 1-\alpha,
\$
which concludes the proof of statement (b). 

\subsection{Proof of Proposition~\ref{prop:general_conf_pval}}
\label{app:proof_general_conf_pval}
Fix $j\in \hcS$. 
The proof proceeds as follows: 
we consider the cases of $y \le c_{n+j}$ and $y > c_{n+j}$ separately. 
In each case, we show $\hcR_{n+j}(y) = \hcR^\cp_{n+j}(y)$ by 
proving that $\hcR_{n+j}(y) \subseteq \hcR^\cp_{n+j}(y)$ and 
$\hcR^\cp_{n+j}(y) \subseteq \hcR_{n+j}(y)$.

For notational simplicity,
let $A^\swap{i}{j}_s(y)$ and $B^\swap{i}{j}_s(y)$ denote the partial sums resulting 
from swapping units $i$ and $n+j$ in $A_s$ and $B_s$, and $\tau^\swap{i}{j}(y)$ is the 
corresponding stopping time.

\subsubsection{Case 1: $y \le c_{n+j}$.}
\paragraph{Showing $\hcR_{n+j}(y) \subseteq \hcR^\cp_{n+j}(y)$.}
Suppose $i \in \hcR_{n+j}(y)$. We now show that $i \in \hcR^\cp_{n+j}(y)$.
By definition, we have $\hS_i \ge \tau^\swap{i}{j}(y)$ and 
\$
\{\ell \in [m]: \hS_{n+\ell}^\swap{i}{j}(y) \ge \tau^\swap{i}{j}\} 
\in \mfS.
\$
\begin{enumerate}
\item [(a)]
Suppose $Y_i \le c_i$. 
By the definition of $\tau^{(j)}(1,0)$, we have that 
\$ 
\ind\{\tau^{(j)}(1,0) \le \tau^\swap{i}{j}(y)\}
= f_{\tau^\swap{i}{j}(y)}(\{A_s^{(j)}(1,0)\}_{s \le \tau^\swap{i}{j}}, \{B_s^{(j)}\}_{s \le \tau^\swap{i}{j}}).
\$
Since $\hS_i \ge \tau^\swap{i}{j}(y)$, 
for any $s \le \tau^\swap{i}{j}(y)$, $\hS_i \ge s$ and 
\$ 
A_s^{(j)}(1,0) & = \ind\{\hS_{n+j} \ge s\} + \sum_{k \in [n]}\ind\{\hS_k \ge s, Y_k \le c_k\} \\
& = \ind\{\hS_{n+j} \ge s, y \le c_{n+j}\} + \sum_{k \neq i}\ind\{\hS_k \ge s, Y_k \le c_k\} + 1
= A_s^\swap{i}{j}(y),\\
B_s^{(j)} & = 1 + \sum_{\ell \neq j} \ind\{\hS_{n+\ell} \ge s\} =
\ind\{\hS_i \ge s\} + \sum_{\ell \neq j} \ind\{\hS_{n+\ell} \ge s\}
 = B_s^\swap{i}{j}(y).
\$
As a result, 
\begin{align} 
\begin{split}\label{eq:swap_tau}
\ind\{\tau^{(j)}(1,0) \le \tau^\swap{i}{j}(y)\}  
& = f_{\tau^\swap{i}{j}(y)}(\{A_s^{(j)}(1,0)\}_{s \le \tau^\swap{i}{j}(y)}, \{B_s^{(j)}\}_{s \le \tau^\swap{i}{j}(y) })\\
& = f_{\tau^\swap{i}{j}(y)}(\{A_s^\swap{i}{j}(y)\}_{s \le \tau^\swap{i}{j}}, \{B_s^\swap{i}{j}(y)\}_{s \le \tau^\swap{i}{j}})\\
& = \ind\{\tau^\swap{i}{j}(y) \le \tau^\swap{i}{j}(y)\} = 1,
\end{split}
\end{align} 
where we use the definition of $\tau^\swap{i}{j}(y)$.
Above, we have shown $\tau^{(j)}(1,0) \le \tau^\swap{i}{j}(y)$, 
and so that $\hS_i \ge \tau^{(j)}(1,0)$.
Conversely, by the definition of $\tau^\swap{i}{j}(y)$, there is 
\$ 
\ind\{\tau^\swap{i}{j}(y) \le \tau^{(j)}(1,0)\}
= f_{\tau^{(j)}(1,0)}(\{A_s^\swap{i}{j}(y)\}_{s \le \tau^{(j)}(1,0)}, 
\{B_s^\swap{i}{j}(y)\}_{s \le \tau^{(j)}(1,0)}).
\$
Now, observe that for any $s \le \tau^{(j)}(1,0)$,
\$ 
A_s^\swap{i}{j}(y) & = \ind\{\hS_{n+j} \ge s, y\le c_{n+j}\} + 
\sum_{k \neq i}\ind\{\hS _k \ge s, Y_k \le c_k\} + 1\\ 
& =  \ind\{\hS_{n+j} \ge s\} + 
\sum_{k \in [n]}\ind\{\hS_k \ge s, Y_k \le c_k\} = A_s^{(j)}(1,0),\\
B_s^\swap{i}{j}(y) &= \ind\{\hS_i \ge s\} + \sum_{\ell \neq j} \ind\{\hS_{n+\ell} \ge s\} 
= 1 + \sum_{\ell \neq j} \ind\{\hS_{n+\ell} \ge s\} = B_s^{(j)}.
\$
Using an argument similar to~\eqref{eq:swap_tau}, we can draw the conclusion that 
$\ind\{\tau^\swap{i}{j}(y) \le \tau^{(j)}(1,0)\} = 1$. 
Combining the two results, we have that $\tau^\swap{i}{j}(y) = \tau^{(j)}(1,0)$. 
Therefore, $\hS_i \ge \tau^\swap{i}{j}(y) = \tau^{(j)}(1,0)$ and 
\$ 
\{\ell \in [m]: \hS_{n+\ell}^\swap{i}{j} \ge \tau^{(j)}(1,0)\} 
= \{\ell \in [m]: \hS_{n+\ell}^\swap{i}{j}(y) \ge \tau^\swap{i}{j}\} 
\in \mfS.
\$
That is, $i \in \hcR^\cp_{n+j}$(y).

\item [(b)] Suppose that $Y_i > c_i$. Then by the 
definition of $\tau^{(j)}(1,1)$, we have that 
\$ 
\ind\{\tau^{(j)}(1,1) \le \tau^\swap{i}{j}(y)\} 
= f_{\tau^\swap{i}{j}(y)}(\{A_s^{(j)}(1,1)\}_{s \le \tau^\swap{i}{j}}, \{B_s^{(j)}\}_{s \le \tau^\swap{i}{j}}).
\$
For any $s \le \tau^\swap{i}{j}(y)$,
\$ 
A_s^{(j)}(1,1) & = 
1 + \sum_{k \in [n]} \ind\{\hS_k \ge s, Y_k \le c_k\} + 
\ind\{\hS_{n+j} \ge s\} \\
& = 1 + \sum_{k \neq i} \ind\{\hS_k \ge s, Y_k \le c_k\} + 
\ind\{\hS_{n+j} \ge s, y\le c_{n+j}\} 
= A_s^\swap{i}{j}(y),\\ 
B_s^{(j)} & = 1 + \sum_{\ell \neq j} \ind\{\hS_{n+\ell} \ge s\} = 
\ind\{\hS_i \ge s\} + \sum_{\ell \neq j} \ind\{\hS_{n+\ell} \ge s\} = B_s^\swap{i}{j}(y).
\$ 
As before, we can conclude that $\tau^{(j)}(1,1) \le \tau^\swap{i}{j}(y)$, 
and therefore $\hS_i \ge \tau^{(j)}(1,1)$. 
Conversely, by the definition of $\tau^\swap{i}{j}(y)$, there is
\$ 
\ind\{\tau^\swap{i}{j}(y) \le \tau^{(j)}(1,1)\}
= f_{\tau^{(j)}(1,1)}(\{A_s^\swap{i}{j}(y)\}_{s \le \tau^{(j)}(1,1)},
\{B_s^\swap{i}{j}(y)\}_{s \le \tau^{(j)}(1,1)}).
\$
For any $s \le \tau^{(j)}(1,1)$, 
\$
A_s^\swap{i}{j}(y) & = 1 + \sum_{k \neq i} \ind\{\hS_k \ge s, y_k\le c_k \}
+ \ind\{\hS_{n+j} \ge s, y \le c_{n+j}\}\\
& = 1+  \sum_{k \in [n]} \ind\{\hS_k \ge s, y_k\le c_k \}
+ \ind\{\hS_{n+j} \ge s\} = A^{(j)}(1,1),\\
B_s^\swap{i}{j}(y) & = \ind\{\hS_i \ge s\} + \sum_{\ell \neq j} \ind\{\hS_{n+\ell} \ge s\}
= 1 + \sum_{\ell \neq j} \ind\{\hS_{n+\ell} \ge s\} = B_s^{(j)}.
\$
Then $\tau^\swap{i}{j}(y) \le \tau^{(j)}(1,1)$, and therefore 
$\tau^\swap{i}{j}(y) = \tau^{(j)}(1,1)$.
Consequently, $\hS_i \ge \tau^\swap{i}{j}(y) = \tau^{(j)}(1,1)$ and 
\$ 
\{ \ell \in [m]: \hS_{n+\ell}^\swap{i}{j} \ge \tau^{(j)}(1,1)\}
= \{\ell \in [m]: \hS_{n+\ell}^\swap{i}{j} \ge \tau^\swap{i}{j}(y)\}
\in \mfS.
\$ 
\end{enumerate}
Combining (a) and (b), we have proved that $\hcR_{n+j}(y) \subseteq \hcR^\cp(y)$.\\

\paragraph{Showing $\hcR^\cp_{n+j}(y) \subseteq \hcR_{n+j}(y)$.}
Now suppose $i \in \hcR^\cp_{n+j}(y)$. We shall show that $i \in \hcR_{n+j}(y)$.

\begin{enumerate}
\item [(c)] Suppose$Y_i \le c_i$, and therefore $\hS_i \ge \tau^{(j)}(1,0)$. 
Again, by the definition of $\tau^\swap{i}{j}(y)$, we can write 
\$ 
\ind\{\tau^\swap{i}{j}(y)\le \tau^{(j)}(1,0)\} 
& = f_{\tau^{(j)}(1,0)}(\{A_s^\swap{i}{j}(y)\}_{s \le \tau^{(j)}(1,0)}, 
\{B_s^\swap{i}{j}(y)\}_{s \le \tau^{(j)}(1,0)}). 
\$
For any $s \le \tau^{(j)}(1,0)$, we have 
\$ 
A_s^\swap{i}{j}(y) & = 1 + \sum_{k \neq i} \ind\{\hS_k \ge s , Y_k \le c_k\} + 
\ind\{\hS_{n+j} \ge s, y \le c_{n+j}\}\\
& = \sum_{k \in [n]} \ind\{\hS_k \ge s , Y_k \le c_k\} + 
\ind\{\hS_{n+j} \ge s\} = A_s^{(j)}(1,0),\\ 
B_s^\swap{i}{j}(y) & = \ind\{\hS_i \ge s\} + \sum_{\ell \neq j} \ind\{\hS_{n+\ell} \ge s\}
= 1 + \sum_{\ell \neq j} \ind\{\hS_{n+\ell} \ge s\} = B_s^{(j)}.
\$
Then as in~\eqref{eq:swap_tau}, we have 
$\tau^\swap{i}{j}(y) \le \tau^{(j)}(1,0)$. 
We also have that $\hS_i \ge \tau^\swap{i}{j}(y)$. 
Conversely, using the definition of $\tau^{(j)}(1,0)$, we can write that 
\$ 
\ind\{\tau^{(j)}(1,0)\le \tau^\swap{i}{j}(y)\}
= f_{\swap{i}{j}(y)}(\{A_s^{(j)}(1,0)\}_{s \le \tau^\swap{i}{j}(y)}, 
\{B_s^{(j)}\}_{s \le \tau^\swap{i}{j}(y)}).
\$
For any $s \le \tau^\swap{i}{j}(y)$, we can check that 
\$ 
A_s^{(j)}(1,0) & = \ind\{\hS_{n+j} \ge s\} + \sum_{k \in [n]}\ind\{\hS_k \ge s, Y_k \le c_k\}\\ 
& = \ind\{\hS_{n+j} \ge s, y \le c_{n+j}\} + \sum_{k \neq i}\ind\{\hS_k \ge s, Y_k \le c_k\} + 1
= A_s^\swap{i}{j}(y),\\
B_s^{(j)} & = 1 + \sum_{\ell \neq j} \ind\{\hS_{n+\ell} \ge s\} = 
\ind\{\hS_i \ge s\} + \sum_{\ell \neq j} \ind\{\hS_{n+\ell} \ge s\} = B_s^\swap{i}{j}(y).
\$
Then $\ind\{\tau^{(j)}(1,0) \le \tau^\swap{i}{j}(y)\} = 
\ind\{\tau^\swap{i}{j}(y) \le \tau^\swap{i}{j}(y)\}=1$.
Collectively, we arrive at $\tau^{(j)}(1,0) = \tau^\swap{i}{j}(y)$, implying that 
$\hS_i \ge \tau^{(j)}(1,0) = \tau^\swap{i}{j}(y)$ and 
\$
\{\ell \in [m]: \hS_{n+\ell}^\swap{i}{j} \ge \tau^\swap{i}{j}(y)\}
= \{\ell \in [m]: \hS_{n+\ell}^\swap{i}{j} \ge \tau^{(j)}(1,0)\} \in \mfS,
\$
which implies that $i \in \hcR_{n+j}(y)$.

\item [(d)] Suppose $Y_i > c_i$. In this case, we have $\hS_i \ge \tau^{(j)}(1,1)$. 
The definition of $\tau^\swap{i}{j}(y)$ leads to 
\$ 
\ind\{\tau^\swap{i}{j}(y) \le \tau^{(j)}(1,1) \}
= f_{\tau^{(j)}(1,1)}(\{A_s^\swap{i}{j}(y)\}_{s \le \tau^{(j)}(1,1)},
\{B_s^\swap{i}{j}(y)\}_{s \le \tau^{(j)}(1,1)}).
\$
For any $s \le \tau^{(j)}(1,1)$, we have 
\$ 
A_s^\swap{i}{j}(y) & = 1 + \sum_{k \neq i} \ind\{\hS_k \ge s, Y_k \le c_k\} + \ind\{\hS_{n+j} \ge s, y\le c_{n+j}\}\\
& = 1 + \sum_{k \in [n]} \ind\{\hS_k \ge s, Y_k \le c_k\} + \ind\{\hS_{n+j} \ge s\} = A_s^{(j)}(1,1),\\
B_s^\swap{i}{j}(y) & = \ind\{\hS_i \ge s\} + \sum_{\ell \neq j} \ind\{\hS_{n+\ell} \ge s\}
= 1 + \sum_{\ell \neq j} \ind\{\hS_{n+\ell} \ge s\} = B_s^{(j)}. 
\$
Then we have $\ind\{\tau^\swap{i}{j}(y) \le \tau^{(j)}(1,1)\} = 1$, 
and therefore $\hS_i \ge \tau^\swap{i}{j}(y)$.
Conversely, by the definition of $\tau^{(j)}(1,1)$, we can write that 
\$ 
\ind\{\tau^{(j)}(1,1) \le \tau^\swap{i}{j}(y)\} 
= f_{\tau^\swap{i}{j}(y)}(\{A_s^{(j)}(1,1)\}_{s \le \tau^\swap{i}{j}(y)}, 
\{B_s^{(j)}\}_{s \le \tau^\swap{i}{j}(y)}).
\$
For any $s \le \tau^\swap{i}{j}(y)$, there is 
\$ 
A_s^{(j)}(1,1) & = 1 + \sum_{k \in [n]} \ind\{\hS_k \ge s, Y_k \le c_k\} + \ind\{\hS_{n+j} \ge s\}\\
& = 1 + \sum_{k \neq i} \ind\{\hS_k \ge s, Y_k \le c_k\} + \ind\{\hS_{n+j} \ge s, y\le c_{n+j}\} = A_s^\swap{i}{j}(y),\\
B_s^{(j)} & = 1 + \sum_{\ell \neq j} \ind\{\hS_{n+\ell} \ge s\} = 
\ind\{\hS_i \ge s\} + \sum_{\ell \neq j} \ind\{\hS_{n+\ell} \ge s\} = B_s^\swap{i}{j}(y).
\$
Then we have $\ind\{\tau^{(j)}(1,1) \le \tau^\swap{i}{j}(y)\} = 1$.
Combining the two results, we have that $\tau^{(j)}(1,1) = \tau^\swap{i}{j}(y)$, and
therefore $\hS_i \ge \tau^{(j)}(1,1) = \tau^\swap{i}{j}(y)$ and 
\$
\{\ell \in [m]: \hS_{n+\ell}^\swap{i}{j} \ge \tau^\swap{i}{j}(y)\}
= \{\ell \in [m]: \hS_{n+\ell}^\swap{i}{j} \ge \tau^{(j)}(1,1)\} \in \mfS,
\$
which implies that $i \in \hcR_{n+j}(y)$.
\end{enumerate}
Combining (c) and (d), we have proved that $\hcR_{n+j}^\cp(y) \subseteq \hcR_{n+j}(y)$, 
which finishes the proof for case 1.

\subsubsection{Case 2: $y > c_{n+j}$.}
As in case 1, we prove that $\hcR_{n+j}(y) = \hcR_{n+j}^\cp(y)$ by showing that
$\hcR_{n+j}(y) \subseteq \hcR_{n+j}^\cp(y)$ and $\hcR_{n+j}^\cp(y) \subseteq \hcR_{n+j}(y)$. 

\paragraph{Showing $\hcR_{n+j}(y) \subseteq \hcR_{n+j}^\cp(y)$.}
Suppose $i \in \hcR_{n+j}(y)$. We proceed to show that $i \in \hcR_{n+j}^\cp(y)$.

\begin{enumerate}
\item [(a)] Suppose $Y_i \le c_i$. 
Using the definition of $\tau^{(j)}(0,0)$, we write
\$ 
\ind\{\tau^{(j)}(0,0) \le \tau^\swap{i}{j}(y)\}
= f_{\tau^\swap{i}{j}(y)}(\{A_s^{(j)}(0,0)\}_{s \le \tau^\swap{i}{j}(y)},
\{B_s^{(j)}\}_{s \le \tau^\swap{i}{j}(y)}).
\$ 
We have for any $s \le \tau^\swap{i}{j}(y)$ that 
\$ 
A_s^{(j)}(0,0) & = \sum_{k \in [n]} \ind\{\hS_k \ge s, Y_k \le c_k\} \\
& = 1 +  \sum_{k\neq i} \ind\{\hS_k \ge s, Y_k \le c_k\} + \ind\{\hS_{n+j}\ge s, y \le c_{n+j}\}
= A_s^\swap{i}{j}(y),\\
B_s^{(j)} & = 1 + \sum_{\ell \neq j} \ind\{\hS_{n+\ell} \ge s\} =
\ind\{\hS_i \ge s\} + \sum_{\ell \neq j} \ind\{\hS_{n+\ell} \ge s\} = B_s^\swap{i}{j}(y).
\$
We then have $\ind\{\tau^{(j)}(0,0) \le \tau^\swap{i}{j}(y)\} = \ind\{\tau^\swap{i}{j}(y)
\le \tau^\swap{i}{j}(y)\} = 1$.
Similarly, we can use the definition of $\tau^\swap{i}{j}(y)$ and obtain that  
\$ 
\ind\{\tau^\swap{i}{j}(y) \le \tau^{(j)}(0,0)\} = 
f_{\tau^{(j)}(0,0)}(\{A_s^\swap{i}{j}(y)\}_{s \le \tau^{(j)}(0,0)},
\{B_s^\swap{i}{j}(y)\}_{s \le \tau^{(j)}(0,0)}).
\$
Again for any $s \le \tau^{(j)}(0,0)$, we have 
\$ 
A_s^\swap{i}{j}(y) & = 1 + \sum_{k \neq i} \ind\{\hS_k \ge s, Y_k \le c_k\} +
\ind\{\hS_{n+j} \ge s, y \le c_{n+j}\} \\ 
& =\sum_{k \in [m]} \ind\{\hS_k \ge s, Y_k \le c_k\}
= A_s^{(j)}(0,0),\\
B_s^\swap{i}{j}(y) & = \ind\{\hS_i \ge s\} + \sum_{\ell \neq j} \ind\{\hS_{n+\ell} \ge s\}
= 1 + \sum_{\ell \neq j} \ind\{\hS_{n+\ell} \ge s\} = B_s^{(j)}.
\$
Then we have $\ind\{\tau^\swap{i}{j}(y) \le \tau^{(j)}(0,0)\} = \ind\{\tau^{(j)}(0,0) 
\le \tau^{(j)}(0,0)\} = 1$.
Combining the two results, we have that $\tau^\swap{i}{j}(y) = \tau^{(j)}(0,0)$, and
therefore $\hS_i \ge \tau^\swap{i}{j}(y) = \tau^{(j)}(0,0)$ and
\$
\{\ell \in [m]: \hS_{n+\ell}^\swap{i}{j} \ge \tau^{(j)}(0,0)\}
= \{\ell \in [m]: \hS_{n+\ell}^\swap{i}{j} \ge \tau^\swap{i}{j}(y)\} \in \mfS,
\$
which implies that $i \in \hcR_{n+j}^\cp(y)$.
\item [(b)] Suppose $Y_i > c_i$. 
We use the definition of $\tau^{(j)}(0,1)$:
\$ 
\ind\{\tau^{(j)}(0,1) \le \tau^\swap{i}{j}(y)\}
= f_{\tau^\swap{i}{j}(y)}(\{A_s^{(j)}(0,1)\}_{s \le \tau^\swap{i}{j}(y)},
\{B_s^{(j)}\}_{s \le \tau^\swap{i}{j}(y)}).
\$ 
For any $s \le \tau^\swap{i}{j}(y)$, we have
\$ 
A_s^{(j)}(0,1) & =  1 + \sum_{k \in [n]} \ind\{\hS_k \ge s, Y_k \le c_k\} \\
& = 1 + \sum_{k \neq i} \ind\{\hS_k \ge s, Y_k \le c_k\} + \ind\{\hS_{n+j} \ge s, y \le c_{n+j}\}
= A_s^\swap{i}{j}(y),\\
B_s^{(j)} & = 1 + \sum_{\ell \neq j} \ind\{\hS_{n+\ell} \ge s\} =
\ind\{\hS_i \ge s\} + \sum_{\ell \neq j} \ind\{\hS_{n+\ell} \ge s\} = B_s^\swap{i}{j}(y).
\$
Then we have $\ind\{\tau^{(j)}(0,1) \le \tau^\swap{i}{j}(y)\} = 
\ind\{\tau^\swap{i}{j}(y) \le \tau^\swap{i}{j}(y)\} = 1$. 
Similarly, we can use the definition of $\tau^\swap{i}{j}(y)$ and obtain that 
\$
\ind\{\tau^\swap{i}{j}(y) \le \tau^{(j)}(0,1)\} =
f_{\tau^{(j)}(0,1)}(\{A_s^\swap{i}{j}(y)\}_{s \le \tau^{(j)}(0,1)},
\{B_s^\swap{i}{j}(y)\}_{s \le \tau^{(j)}(0,1)}).
\$
For any $s \le \tau^{(j)}(0,1)$, we have $\hS_i \ge s$ and 
\$
A_s^\swap{i}{j}(y) & = 1 + \sum_{k \neq i} \ind\{\hS_k \ge s, Y_k \le c_k\} +
\ind\{\hS_{n+j} \ge s, y \le c_{n+j}\} \\
& = 1 + \sum_{k \in [m]} \ind\{\hS_k \ge s, Y_k \le c_k\}
= A_s^{(j)}(0,1),\\
B_s^\swap{i}{j}(y) & = \ind\{\hS_i \ge s\} + \sum_{\ell \neq j} \ind\{\hS_{n+\ell} \ge s\}
= 1 + \sum_{\ell \neq j} \ind\{\hS_{n+\ell} \ge s\} = B_s^{(j)}.
\$
Then we have $\ind\{\tau^\swap{i}{j}(y) \le \tau^{(j)}(0,1)\} = 1$, 
and therefore $\tau^\swap{i}{j}(y) = \tau^{(j)}(0,1)$. 
This implies that $\hS_i \ge \tau^\swap{i}{j}(y) = \tau^{(j)}(0,1)$ and
\$ 
\{\ell \in [m]: \hS_{n+\ell}^\swap{i}{j} \ge \tau^{(j)}(0,1)\} 
= \{\ell \in [m]: \hS_{n+\ell}^\swap{i}{j} \ge \tau^\swap{i}{j}(y)\}
\in \mfS.
\$ 
Then we have $i \in \hcR_{n+j}^\cp(y)$.

\end{enumerate}
Combining (a) and (b), we have proved that $\hcR_{n+j}(y) \subseteq \hcR_{n+j}^\cp(y)$ 
when $y > c_{n+j}$. We proceed to show the reverse inclusion.

\paragraph{Showing $\hcR_{n+j}^\cp(y) \subseteq \hcR_{n+j}(y)$.}
Suppose that $i \in \hcR^\cp_{n+j}(y)$.
\begin{enumerate}
\item [(c)] When $Y_i \le c_i$, $\hS_i \ge \tau^{(j)}(0,0)$.
As before, we have 
\$ 
\ind\{\tau^\swap{i}{j}(y) \le \tau^{(j)}(0,0)\}
= f_{\tau^{(j)}(0,0)}(\{A_s^\swap{i}{j}(y)\}_{s \le \tau^{(j)}(0,0)},
\{B_s^\swap{i}{j}(y)\}_{s \le \tau^{(j)}(0,0)}).
\$
For any $s \le \tau^{(j)}(0,0)$, observe that 
\$ 
A_s^\swap{i}{j}(y) & = 1 + \sum_{k \neq i} \ind\{\hS_k \ \ge s, Y_k \le c_k\} + 
\ind\{\hS_{n+j} \ge s, y \le c_{n+j}\}\\
& = \sum_{k \in [m]} \ind\{\hS_k \ \ge s, Y_k \le c_k\}
= A_s^{(j)}(0,0),\\
B_s^\swap{i}{j}(y) & = \ind\{\hS_i \ge s\} + \sum_{\ell \neq j} \ind\{\hS_{n+\ell} \ge s\}
= 1 + \sum_{\ell \neq j} \ind\{\hS_{n+\ell} \ge s\} = B_s^{(j)},
\$
which leads to $\ind\{\tau^\swap{i}{j}(y) \le \tau^{(j)}(0,0)\} = 1$.
Conversely, we have 
\$
\ind\{\tau^{(j)}(0,0) \le \tau^\swap{i}{j}(y)\}
= f_{\tau^{\swap{i}{j}}(y)}(\{A_s^{(j)}(0,0)\}_{s \le \tau^\swap{i}{j}(y)},
\{B_s^{(j)}\}_{s \le \tau^\swap{i}{j}(y)}),
\$
and for any $s \le \tau^\swap{i}{j}(y)$, we have
\$
A_s^{(j)}(0,0) & = \sum_{k \in [m]} \ind\{\hS_k \ge s, Y_k \le c_k\}\\
& = 1 + \sum_{k \neq i} \ind\{\hS_k \ \ge s, Y_k \le c_k\} + \ind\{\hS_{n+j} \ge s, y \le c_{n+j}\}
= A_s^\swap{i}{j}(y),\\
B_s^{(j)} & = 1 + \sum_{\ell \neq j} \ind\{\hS_{n+\ell} \ge s\} =
\ind\{\hS_i \ge s\} + \sum_{\ell \neq j} \ind\{\hS_{n+\ell} \ge s\} = B_s^\swap{i}{j}(y).
\$
Then we have $\ind\{\tau^{(j)}(0,0) \le \tau^\swap{i}{j}(y)\} = 1$, 
which implies that $\tau^{(j)}(0,0) = \tau^\swap{i}{j}(y)$. We then have
$\hS_i \ge \tau^{(j)}(0,0) = \tau^\swap{i}{j}(y)$ and
\$
\{\ell \in [m]: \hS_{n+\ell}^\swap{i}{j} \ge \tau^\swap{i}{j}(y)\}
= \{\ell \in [m]: \hS_{n+\ell}^\swap{i}{j} \ge \tau^{(j)}(0,0)\}
\in \mfS,
\$
implying that $i \in \hcR_{n+j}(y)$.

\item [(d)] When $Y_i > c_i$, $\hS_i \ge \tau^{(j)}(0,1)$.
We have
\$ 
\ind\{\tau^\swap{i}{j}(y) \le \tau^{(j)}(0,1)\}
= f_{\tau^{(j)}(0,1)}(\{A_s^\swap{i}{j}(y)\}_{s \le \tau^{(j)}(0,1)},
\{B_s^\swap{i}{j}(y)\}_{s \le \tau^{(j)}(0,1)}).
\$ 
For any $s \le \tau^{(j)}(0,1)$, there is 
\$ 
A_s^\swap{i}{j}(y) & = 1 + \sum_{k \neq i} \ind\{\hS_k \ \ge s, Y_k \le c_k\} +
\ind\{\hS_{n+j} \ge s, y \le c_{n+j}\}\\
& = 1 + \sum_{k \in [m]} \ind\{\hS_k \ \ge s, Y_k \le c_k\} = A_s^{(j)}(0,1),\\
B_s^\swap{i}{j}(y) & = \ind\{\hS_i \ge s\} + \sum_{\ell \neq j} \ind\{\hS_{n+\ell} \ge s\} =
1 + \sum_{\ell \neq j} \ind\{\hS_{n+\ell} \ge s\} = B_s^{(j)}.
\$
Therefore, we have $\ind\{\tau^\swap{i}{j}(y) \le \tau^{(j)}(0,1)\} = 1$.
Conversely, we have 
\$
\ind\{\tau^{(j)}(0,1) \le \tau^\swap{i}{j}(y)\}
= f_{\tau^\swap{i}{j}(y)}(\{A_s^{(j)}(0,1)\}_{s \le \tau^\swap{i}{j}(y)},
\{B_s^{(j)}\}_{s \le \tau^\swap{i}{j}(y)}),
\$
where for any $s \le \tau^\swap{i}{j}(y)$, we have
\$
A_s^{(j)}(0,1) & = 1 + \sum_{k \in [m]} \ind\{\hS_k \ge s, Y_k \le c_k\}\\
& = 1 + \sum_{k \neq i} \ind\{\hS_k \ge s, Y_k \le c_k\} + \ind\{\hS_{n+j} \ge s, y \le c_{n+j}\}
= A_s^\swap{i}{j}(y),\\
B_s^{(j)} & = 1 + \sum_{\ell \neq j} \ind\{\hS_{n+\ell} \ge s\} =
\ind\{\hS_i \ge s\} + \sum_{\ell \neq j} \ind\{\hS_{n+\ell} \ge s\} = B_s^\swap{i}{j}(y).
\$
Then we have $\ind\{\tau^{(j)}(0,1) \le \tau^\swap{i}{j}(y)\} = 1$, 
which implies that $\tau^{(j)}(0,1) = \tau^\swap{i}{j}(y)$. 
Therefore, $\hS_i \ge \tau^{(j)}(0,1) = \tau^\swap{i}{j}(y)$ and
\$
\{\ell \in [m]: \hS_{n+\ell}^\swap{i}{j} \ge \tau^\swap{i}{j}(y)\}
= \{\ell \in [m]: \hS_{n+\ell}^\swap{i}{j} \ge \tau^{(j)}(0,1)\}
\in \mfS,
\$
implying that $i \in \hcR_{n+j}(y)$.
\end{enumerate}
Combining (c) and (d), we have proved that $\hcR_{n+j}^\cp(y) \subseteq \hcR_{n+j}(y)$ 
when $y > c_{n+j}$. The whole proof is now complete.

\subsection{Proof of Proposition~\ref{prop:conf_pred}}
\label{appd:proof_conf_pred}
Recall that the generic form of the JOMI prediction set is 
\$ 
\hC_{\alpha,n+j} = \Big\{y\in \cY: V(X_{n+j},y) \le  
\quant\big(\{V_i\}_{i \in \hcR_{n+j}(y)}\cup\{\infty\}, 1-\alpha \big)\Big\},
\$
with $\hcR_{n+j}(y)$ being the reference set.
We categorize $y$ into three cases according to the value of $S(X_{n+j},y)$:
(1) $S(X_{n+j},y) < \eta^-$, 
(2) $S(X_{n+j},y) > \eta^+$, and (3) $\eta^-\le S(X_{n+j},y) \le \eta^+$.
The case where $\eta^- \le S(X_{n+j},y) \le \eta^+$ is obvious, 
so we discuss the remaining two cases in the following.
In what follows, we let $\eta^\swap{i}{j}(y)$ denote the $K$-th smallest 
elements in $\{S^\swap{i}{j}_{\ell}(y)\}_{\ell \in \cic}$.

\paragraph{Case 1: $S(X_{n+j},y) < \eta^-$.}
Consider $y\in \cY$ such that $S(X_{n+j},y) < \eta^-$ and $y \in \hC_{\alpha,n+j}$. 
We claim that in this case 
\$
\hcR_{n+j}(y) =  \tR_{n+j}^{(1)}\,:=&\, \{i\in\cic: S_i \le \eta^-,\cL(X_i, \eta) = 1, 
\{\ell \in [m]: \cL(X^\swap{i}{j}_{n+\ell}, \eta)=1\} \in \mfS\} \\
& \cup \{i\in\cic: S_i > \eta^-, \cL(X_i, \eta^-) = 1, 
\{\ell \in [m]: \cL(X_{n+\ell}^\swap{i}{j},\eta^-) = 1\} \in \mfS\}.
\$
To see this, note that when $S_i \le \eta^-$, we have $\etas = \eta$. Then 
$\cL(X_i, \etas) = \cL(X_i,\eta)$ and 
\$ 
\{\ell \in  [m]: \cL(X_{n+\ell}^\swap{i}{j}, \eta) = 1\}
= 
\{\ell \in  [m]: \cL(X_{n+\ell}^\swap{i}{j}, \etas) = 1\}
= \hcS^\swap{i}{j}(y).
\$ 
When $S_i > \eta^-$,  $\etas = \eta^-$, $\cL(X_i,\eta^-) = \cL(X_i,\etas)$
and 
\$ 
\{\ell \in  [m]: \cL(X_{n+\ell}^\swap{i}{j}, \eta^-) = 1\}
= 
\{\ell \in  [m]: \cL(X_{n+\ell}^\swap{i}{j}, \etas) = 1\}
= \hcS^\swap{i}{j}(y).
\$ 
As a result, we have that $\hcR_{n+j}(y) = \tR_{n+j}^{(1)}$ and therefore 
\$ 
V(X_{n+j},y) \le \quant\big(1-\alpha;\{V_i: i \in \tR_{n+j}^{(1)}\}\cup\{\infty\}\big)
\$

\paragraph{Case 2: $S(X_{n+j},y) > \eta^+$.}
We now consider $y\in \cY$ such that $S(X_{n+j},y)> \eta^+$ and $y \in \hC_{\alpha,n+j}$.
In this case, we instead have
\$
\hcR_{n+j}(y) = \tR_{n+j}^{(2)}\,:=&\, \{i\in\cic: S_i \le \eta,\cL(X_i, \eta^+) = 1, 
\{\ell \in [m]: \cL(X_{n+\ell}^\swap{i}{j}, \eta^+) = 1\}\in \mfS\}\\ 
& \cup \{i\in\cic: S_i > \eta, \cL(X_i, \eta) = 1, 
\{\ell \in [m]: \cL(X_{n+\ell}^\swap{i}{j}, \eta) = 1\} \in \mfS\}.
\$
This is because when $S_i\le \eta$,  $\etas = \eta^+$ and therefore $\cL(X_i,\etas) = \cL(X_i,\eta^+)$,
as well as 
\$ 
\hcS^\swap{i}{j}(y) 
= \{\ell \in [m]: \cL(X_{n+\ell}^\swap{i}{j}, \etas) = 1\}
= \{\ell \in [m]: \cL(X_{n+\ell}^\swap{i}{j}, \eta^+) = 1\}.
\$
When $S_i > \eta$, then $\etas = \eta$. Therefore $\cL(X_i,\etas) = \cL(X_i,\eta)$, and
\$ 
\hcS^\swap{i}{j}(y) 
= \{\ell \in [m]: \cL(X_{n+\ell}^\swap{i}{j}, \etas) = 1\}
= \{\ell \in [m]: \cL(X_{n+\ell}^\swap{i}{j}, \eta) = 1\}.
\$
Consequently, 
\$ 
V(X_{n+j},y) \le \quant\big(1-\alpha;\{V_i:i \in \tR_{n+j}^{(2)}\}\cup\{\infty\}\big).
\$
With the above, we have completed the proof of Proposition~\ref{prop:conf_pred}.

\subsection{Proof of Proposition~\ref{prop:topK}}
\label{appd:prood_of_topk}
We start by defining an alternative stopping time that is determined solely by $\cD_j^c$: 
\$ 
\tT_\topk = \inf\Big\{t \in \cT: \sum_{\ell \in [m]\backslash\{j\}} 
\ind\{S_{n+\ell} \le t\} \ge m-K \Big\}.
\$
Abusing the notations a bit, we call $T^\swap{i}{j}_\topk$ and $\tT_\topk^\swap{i}{j}$ the 
corresponding stopping times defined with $(\cdc^\swap{i}{j}$, $\cdte^\swap{i}{j})$.
By definition, it is straightfoward to check that $\tT^\swap{i}{j}_\topk = \tT_\topk$.

We claim that  $S_{n+j} > T_\topk$ implies $T_\topk =  \tT_\topk$.
To see this, note first that for any $t \in \cT$,
\$ 
\sum_{\ell \in [m]} \ind\{S_{n+\ell} \le t\} \ge 
\sum_{\ell \in [m]\backslash \{j\} } \ind\{S_{n+\ell} \le t\}
~\Rightarrow ~ \tT_\topk \ge T_\topk. 
\$
On the other hand, one can check that 
\@ \label{eq:topk_ttj}
\sum_{\ell \in [m]\backslash\{j\}} \ind\{S_{n+\ell} \le T_\topk\} \stepa{=} 
\sum_{\ell \in [m]\backslash \{j\}} \ind\{S_{n+\ell} \le T_\topk\} + \ind\{S_{n+j} \le T_\topk\} 
\stepb{\ge}    m-K. 
\@
Above, step (a)  is because $S_{n+j} > T_\topk$ and step (b)
follows 
from the choice of $T_\topk$. With the definition of $\tT_\topk$, we can then conclude 
from~\eqref{eq:topk_ttj} that $T_\topk \ge \tT_\topk$, and  hence $T_\topk = \tT_\topk$. 

We can similarly show that $T^\swap{i}{j}_\topk = \tT_\topk$ when $S_i > T_\topk^\swap{i}{j}$.
By construction, we have $\tT_\topk \ge T_\topk^\swap{i}{j}$. Next, 
\$ 
\sum_{\ell \in [m]\backslash \{j\}} 
\ind\{S_{n+\ell} \le T_\topk^\swap{i}{j}\} 
=\sum_{\ell \in [m]\backslash\{j\}} 
\ind\{S_{n+\ell} \le T_\topk^\swap{i}{j}\} 
+ \ind\{S_i \le T_\topk^\swap{i}{j}\}
\ge m-K.
\$
By the definition of $\tT_\topk$, we then have $T_\topk^\swap{i}{j} \ge \tT_\topk$,
and therefore $T_\topk^\swap{i}{j} = \tT_\topk$.

We now proceed to show that $\hcR_\fixed = \hcR_{n+j}$.
For any $i \in \hcR_\topk$, 
it holds that  $S_i > T_\topk$. Then 
since $T_\topk = \tT_\topk \ge T^\swap{i}{j}_\topk$, we know that 
$S_i > T^\swap{i}{j}_\topk$. 
Therefore, $T^\swap{i}{j}_\topk = \tT_\topk$, and
\$
\hcS^\swap{i}{j} 
& = \big\{\ell \in[m]: S_{n+\ell}^{\swap{i}{j}} > T^\swap{i}{j}_\topk\big\}
 = \big\{\ell \in [m]: S_{n+\ell}^\swap{i}{j} > \tT_\topk \big\}\\
& \stepa{=}  \{j\} \cup\big\{\ell \in [m]\backslash \{j\}: S_{n+\ell} > \tT_\topk \big\} \\
& = \{j\} \cup\big\{\ell \in [m]\backslash \{j\}: S_{n+\ell} > T_\topk \big\} 
= \hcS, 
\$ 
wher step (a) follows from the fact that $S_{n+j}^\swap{i}{j} = S_i$ and $S_i > \tT_\topk$.
As a result, $j\in \hcS^\swap{i}{j}$ and $\hcS^\swap{i}{j} \in \mfS$, i.e., $i \in \hcR_{n+j}$.
On the other hand, for any $i \in \hcR_{n+j}$, it holds that $j \in \hcS^\swap{i}{j}$.
Then,
\@
S_i > T^\swap{i}{j}_\topk ~\Rightarrow~ T^\swap{i}{j}_\topk = \tT_\topk  = T_\topk 
~\Rightarrow ~S_i > T_\topk.
\@  
Combining all the above, we complete the proof.

\subsection{Proof of Proposition~\ref{prop:cq}}
\label{appd:proof_cq}
Fix $j \in \hcS_{\cq}$.
Define the proxy stopping time that is agnostic to the swapping of units $i$ and $n+j$:
\$ 
\tT_{\cq} = \inf\bigg\{t \in \cT: 
\sum_{i\in \cI_j} \ind\{S_i \le t\} \ge qn \bigg\}.
\$
By construction, $\tT_\cq \le T_\cq$. We claim that when $S_{n+j} > \tT_\cq$, $T_\cq = \tT_\cq$. 
To see this, note that 
\@ \label{eq:cq_ttj}
\sum_{i\in \cic} \ind\{S_i \le \tT_\cq\}
= \sum_{i\in \cic} \ind\{S_i \le \tT_\cq\} + \ind\{S_{n+j} \le \tT_\cq\} \ge qn,
\@
where the first step is because $S_{n+j} > \tT_\cq$ and
the second step follows from the definition of $\tT_\cq$. 
By the definition of $T_\cq$,~\eqref{eq:cq_ttj} implies that 
$\tT_\cq \ge T_\cq$ and therefore $T_\cq = \tT_\cq$.

Similarly, we claim that $T^\swap{i}{j}_\cq = \tT_\cq$ when $S_i > \tT_\cq$.
By construction, we have $\tT_\cq \le  T^\swap{i}{j}_\cq$.
When $S_i > \tT_\cq$, 
\$ 
\sum_{\ell \in \cic\backslash\{i\}} \ind\{S_\ell \le \tT_\cq\} + 
\ind\{S_{n+j} \le \tT_\cq\} = \sum_{\ell \in \cic} \ind\{S_\ell \le \tT_\cq\} + 
\ind\{S_{n+j}\le \tT_\cq\}\ge qn, 
\$
with the first step following from $S_i > \tT_\cq$  and 
the second  step from the definition of $\tT_\cq$.
We then have $\tT_\cq \ge T_\cq^\swap{i}{j}$ 
and therefore $T^\swap{i}{j}_\cq = \tT_\cq$. 

Next, we proceed to show that $\hcR_{n+j} = \hcR_\cq$.
For any $i\in \hcR_\cq$, there is $S_i > T_\cq= \tT_\cq$. 
Then, $T^\swap{i}{j}_\cq = \tT_\cq$, and
\$ 
\hcS^\swap{i}{j} & = \{\ell \in [m]: S_{n+\ell}^\swap{i}{j} > T^\swap{i}{j}_\cq\}
{=} \{\ell \in [m]: S_{n+\ell}^\swap{i}{j} > \tT_\cq\}\\
& \stepa{=} \{j\} \cup \{\ell \in [m]\backslash \{j\}: S_{n+\ell} > T_\cq\} = \hcS_\cq,
\$  
where step (a) follows from the fact that $S^\swap{i}{j}_{n+j} = S_i$
and $S_i > \tT_\cq$. The above implies that $j \in  \hcS^\swap{i}{j}$ and 
that 
$\hcS^\swap{i}{j} \in \mfS$. Conversely, for $i \in \cic$ such that $j \in \hcS^\swap{i}{j}$, 
\$ 
S_i > T^\swap{i}{j}_\cq ~\Rightarrow ~
T^\swap{i}{j}_\cq = \tT_\cq = T_\cq  ~\Rightarrow~ S_i > T_\cq,
\$ 
which completes the proof.

\section{Auxiliary results}
\subsection{Connection between selection-conditional coverage and FCR control}

Lemma~\ref{lem:scc_fcr} establishes 
the asympototic equivalence 
between the selection-conditional coverage and the false coverage rate (FCR) 
when the selection rule is asymptotically close to a threshold function. 

\begin{lemma}\label{lem:scc_fcr}
Suppose $\{Z_i\}_{i=1}^{m+n}$ are i.i.d., and there exists some fixed function 
$f: \cX \mapsto \RR$ and fixed constant $\tau \in \RR$ such that the selection set $\hat\cS = \cS(\cD_\calib,\cD_\test)$ 
obeys $ |\hat\cS\Delta \hat\cS^*|/|\hat\cS|  {\to} 0$ in probability as $n,m\rightarrow \infty$,  
where $\hat\cS^* = \{j\colon f(X_{n+j})\geq \tau\}$ and $\PP(f(X_{n+1})\ge \tau)>0$.
Assume that $\hat S$ is equivariant to the permutation of the units in $\cD_\test$ and that 
the prediction set takes the form $\hat{C}_{\alpha,n+j} = \{y\colon V(X_{n+j},y)\leq \hat{q}_{m,n}\}$, 
where $\hat{q}_{m,n} \stackrel{\textnormal{p}}{\to} q^*\in \RR$, as $m,n \to \infty$. 
Assume also that the function $G(t) := \PP(f(X_{n+1})\ge \tau, V_{n+1} > t)$ is continuous in $t$. 
Then,  
\$ 
\lim_{m,n \to \infty} \textnormal{FCR} 
& = \PP\big( V(X_{n+1},Y_{n+1} ) > q^* \given f(X_{n+1})\ge \tau\big)
= \lim_{m,n\to \infty} \PP(Y_{n+1}\notin \hat{C}_{\alpha,n+1}\given  1\in \hat\cS).
\$
That is, the asymptotic FCR coincides with the asymptotic selection-conditional miscoverage.
\end{lemma}
\begin{proof}
By the definition of FCR, we have
\@\label{eq:fcr_decomp}
\textnormal{FCR} & = \EE\bigg[\frac{\sum_{j=1}^m \ind\{j\in \hat \cS, V_{n+j}> \hat q_{m,n}\}}
{1 \vee |\hat \cS| }\bigg] \nonumber \\ 
& = \EE\bigg[\frac{\sum_{j=1}^m \ind\{j\in \hat \cS^*, V_{n+j}> \hat q_{m,n}\}}
{1 \vee |\hat \cS| }\bigg]  \nonumber\\ 
& \qquad  - \EE\bigg[\frac{\sum_{j=1}^m \ind\{j \in \hat \cS^* \setminus \hat \cS, V_{n+j}> \hat q_{m,n}\}}
{1 \vee |\hat \cS| }\bigg] 
+  \EE\bigg[\frac{\sum_{j=1}^m \ind\{j \in \hat \cS \setminus \hat \cS^*, V_{n+j}> \hat q_{m,n}\}}
{1 \vee |\hat \cS| }\bigg]. 
\@
Note that 
\$ 
& \Bigg|- \EE\bigg[\frac{\sum_{j=1}^m \ind\{j \in \hat \cS^* \setminus \hat \cS, V_{n+j}> \hat q_{m,n}\}}
{1 \vee |\hat \cS| }\Bigg] 
+  \EE\bigg[\frac{\sum_{j=1}^m \ind\{j \in \hat \cS \setminus \hat \cS^*, V_{n+j}> \hat q_{m,n}\}}
{1 \vee |\hat \cS| }\bigg]\bigg|\\
\le ~& \EE\bigg[\frac{\sum_{j=1}^m \ind\{j \in \hat \cS^* \setminus \hat \cS, V_{n+j}> \hat q_{m,n}\}}
{1 \vee |\hat \cS| }\bigg] 
+  \EE\bigg[\frac{\sum_{j=1}^m \ind\{j \in \hat \cS \setminus \hat \cS^*, V_{n+j}> \hat q_{m,n}\}}
{1 \vee |\hat \cS| }\bigg]\\
\le ~& \EE\Big[\frac{|\hcS^* \Delta \hcS|}{|\hcS|}\Big],
\$
which goes to zero by assumption and the dominated convergence theorem.
For the first term in Equation~\eqref{eq:fcr_decomp}, we further have
\$ 
& \EE\bigg[\frac{\sum_{j=1}^m \ind\{j\in \hat \cS^*, V_{n+j}> \hat q_{m,n}\}}
{1 \vee |\hat \cS| }\bigg] \\
= ~& \EE\bigg[\frac{\sum_{j=1}^m \ind\{j\in \hat \cS^*, V_{n+j}> \hat q_{m,n}\}}
{1 \vee |\hat \cS^*| } \frac{|\hat \cS^*|\vee 1}{|\hat \cS| \vee 1}\bigg] \\
= ~& \EE\bigg[\frac{\sum_{j=1}^m \ind\{j\in \hat \cS^*, V_{n+j}> \hat q_{m,n}\}}
{1 \vee |\hat \cS^*| }\bigg] + \EE\bigg[\frac{|\hat \cS^*| \vee 1 - |\hat \cS| \vee 1}
{1 \vee |\hat \cS| } \frac{\sum_{j=1}^m \ind\{j\in \hat \cS^*, V_{n+j}> \hat q_{m,n}\}}
{1 \vee |\hat \cS^*| }\bigg]. 
\$
The absolute value of the second term above is bounded by $\EE[|\hat \cS^* \Delta \hat \cS|/|\hat \cS|]$,
and therefore goes to zero by assumption. Next, by the strong law of large numbers,
\$ 
\frac{1 \vee |\hat \cS^*|}{m} = \frac{1}{m} \vee \frac{1}{m}\sum_{j=1}^m \ind\{f(X_{n+j}) \ge \tau\}
\to \PP(f(X_{n+1}) \ge \tau),\text{ a.s.}
\$
Meanwhile, 
\$ 
\bigg|\frac{1}{m}\sum_{j=1}^m \ind\{j\in \hat \cS^*, V_{n+j}> \hat q_{m,n}\}
- G(\hat q_{m,n})\bigg|
& \le 
\sup_{t\in \RR} ~\bigg|\frac{1}{m}\sum_{j=1}^m \ind\{j\in \hat \cS^*, V_{n+j}> t\}
- G(t)\bigg|\\ 
& = \sup_{t \in\RR} ~ \bigg|\frac{1}{m}\sum_{j=1}^m \ind\{f(X_{n+j}) \ge \tau, V_{n+j}> t\} - G(t)\bigg|
\to 0, \text{ a.s.}
\$
The last step follows from the uniform law of large numbers.
Since $G(t)$ is continuous in $t$ and that $\hat q_{m,n} \stackrel{p}{\to} q^*$, 
there is 
\@\label{eq:asymp_11}
\frac{1}{m}\sum_{j=1}^m \ind\{j\in \hat \cS^*, V_{n+j}> \hat q_{m,n}\}
\stackrel{\text{p}}{\to} G(q^*). 
\@
Then, applying the continuous mapping theorem and the dominated convergence theorem yields 
\$ 
 \EE\bigg[\frac{\sum_{j=1}^m \ind\{j\in \hat \cS^*, V_{n+j}> \hat q_{m,n}\}}
{1 \vee |\hat \cS^*| }\bigg]  \rightarrow \PP\big( V(X_{n+1},Y_{n+1} ) > q^* \given f(X_{n+1})\ge \tau\big).
\$
This establishes the expression of the asymptotic FCR. 
In addition, due to the exchangeability of $\{Z_{n+j}\}_{j=1}^m$
and that $\hat S$ is equivariant to the permutation of the test points, we have 
\$ 
\PP(Y_{n+1}\notin \hat{C}_{\alpha,n+1}\given  1\in \hat\cS)
&= \frac{\frac{1}{m}\sum_{j=1}^m \PP(Y_{n+j}\notin \hat{C}_{\alpha,n+j}, j\in \hat\cS)}{\frac{1}{m}\sum_{j=1}^m \PP( j\in \hat\cS)} \\ 
&= \frac{ \sum_{j=1}^m\EE[ \ind\{V(X_{n+j},Y_{n+j})> \hat{q}_{m,n}, j\in \hat\cS\}]}{ \sum_{j=1}^m \EE[\ind\{ j\in \hat\cS\}]}
\$
For the numerator, we have 
\$ 
& \Big|\frac{1}{m}\sum_{j=1}^m\EE\big[ \ind\{V_{n+j} > \hat{q}_{m,n}, j\in \hat\cS\}\big] 
- \frac{1}{m}\sum_{j=1}^m\EE \big[ \ind\{V_{n+j} > \hat{q}_{m,n}, j\in \hat\cS^*\}\big]\Big| \\
&\le \frac{1}{m} \EE\big[|\hcS \Delta \hcS^*|\big] 
= \EE\bigg[\frac{|\hcS \Delta \hcS^*|}{|\hcS^*|} \frac{|\hcS^*|}{m}\bigg]
\to 0 
\$
where the last step is because 
$|\hat\cS\Delta \hat\cS^*|/|\hat\cS|  {\to} 0$ in probability and the dominated convergence theorem.
Using~\eqref{eq:asymp_11}, we then have 
\$ 
\frac{1}{m}\sum_{j=1}^m\EE\big[ \ind\{V_{n+j} > \hat{q}_{m,n}, j\in \hat\cS\}\big] 
\to G(q^*).
\$ 
Similarly, for the denominator,  
\$  
\frac{1}{m}\Big| \sum_{j=1}^m \EE[\ind\{ j\in \hat\cS\}]  -\frac{1}{m}\sum_{j=1}^m \EE[\ind\{j \in \hcS^*\}]\Big| 
\le \frac{1}{m} \EE\bigg[\frac{|\hcS \Delta \hcS^*|}{|\hcS^*|} |\hcS^*|\bigg]
\to 0.
\$
Therefore, we obtain the same expression for the asymptotic selection-conditional coverage. 
\end{proof}

\subsection{The conformal p-values}
\label{subsec:conf_pval}
Recall that in Section~\ref{sec:conf_pval} of the main text, we introduce the conformal p-value for 
each $j \in [m]$ as 
\$
p_j = \frac{1+\sum_{i \in \cic} \ind\{Y_i\le c_i, \hS_i \ge \hS_{n+j}\}}{n+1},
\$ 
where $\hat{S}_i = S(X_i,c_i)$, and $S\colon \cX\times\RR\to \RR$ is 
the pre-specified score function. 
We can construct sharper p-values by introducing extra randomness. Letting 
$U_{n+1}, U_{n+2},\ldots,U_{n+m}$ be i.i.d.~uniform random variables on $[0,1]$ that 
are independent of everything else, we can construct the randomized conformal p-values as
\$ 
p_j^\rand = \frac{U_{n+j}(1+\sum_{i\in \cic}\ind\{\hS_i = \hS_{n+j},Y_i \le c_i\} 
+ \sum_{i \in \cic}\ind\{\hS_i > \hS_{n+j}, Y_i \le c_i\})}{n+1}. 
\$
The conformal p-values and 
their random counterparts defined in~\cite{jin2022selection} are 
\$ 
& \bar{p}_j = \frac{1+\sum_{i \in \cic}\ind\{S_i \ge \hS_{n+j}\}}{n+1}\text{ and }\\
& \bar{p}^\rand_j = \frac{U_{n+j}(1+\sum_{i\in\cic}\ind\{S_i = \hS_{n+j}\})+
\sum_{i \in \cic}\ind\{S_i > \hS_{n+j}\}}{n+1}.
\$

The validity of $p_j$ and $p_j^\rand$, as well as their relation to $\bar{p}_j$ and 
$\bar{p}_j^\rand$, is established in the following lemma.

\vspace{0.5em}
\begin{lemma}\label{lem:pvals_compare}
Under the condition of Theorem~\ref{thm:PI_cond_set}, 
for any $j\in[m]$, the conformal p-values $p_j$ and $p_j^\rand$ 
defined in (15) satisfy that 
\$ 
& \PP\big(p_j \le t, Y_{n+j} \le c_{n+j}\big) \le t, \quad \text{and} \quad 
\PP\big(p_j^\rand \le t, Y_{n+j} \le c_{n+j}\big) \le t, \quad \forall t\in[0,1].
\$ 

Furthermore, there is  $p_j \le \bar{p}_j$ and 
$p_j^\rand \le \bar{p}_j^\rand$ deterministically.
If the selection score $S(x,y) = \hat{\mu}(X) - \big(\sup_x 2|\hat{\mu}(x)|\big)\cdot y$, 
where $\hat{\mu}$ is regression function fitted on $\cdc$, then 
the equalities hold.
\end{lemma}

\begin{proof}[Proof of Lemma~\ref{lem:pvals_compare}]
For any $j\in[m]$ and $t \in (0,1)$, it can be checked that 
when $Y_{n+j} \le c_{n+j}$, 
\$
p_j\le t \iff \hcS_{n+j} \text{ is among the }k^*\text{-th largest element in }
\cN_j := \{\hS_i: i\in \cI_j, Y_i \le c_i\},
\$
where $k^* = \min(|\cN_j|, \lfloor t(n+1) \rfloor)$. 
Let $T_{k^*}$ denote the value of the $k^*$-th largest element in $\cN_j$. 
Apparently, both $k^*$ and $T_{k^*}$ are invariant 
to the permutation on $\cI_j$. Consequently,  
\$ 
\PP(p_j \le t, Y_{n+j} \le c_{n+j}) & = 
\EE\big[\ind\{Y_{n+j} \le c_{n+j}, \hS_{n+j} \ge T_{k^*} \}\big]
\stepa{=} \frac{1}{n+1}\sum_{i\in \cI_j} \EE\Big[\ind\{Y_i \le c_i, \hS_i \ge T_{k^*}\}\Big]\\
& = \frac{1}{n+1} \EE\Big[\sum_{i \in \cN_j} \ind\{\hS_i \ge T_{k^*}\}\Big] 
\stepb{\le} \frac{k^*}{n+1} \le t, 
\$
where step (a) is due to the exchangeability of $\cD_j$ and step (b) follows from 
the fact that the number of elements greater than or equal to $T_{k^*}$ in $\cN_j$ is at most $k^*$.

Next, we consider the randomized version. For any $i \in \cI_j$,  let $E_i$, $G_i$ and 
$G^+_i$ denote the number of elements in $\cN_j$ equal to, no less than, and greater than 
$\hS_i$, respectively, i.e.,
\$ 
E_i \,:=\, \sum_{\ell \in \cI_j} \ind\{\hS_i =\hS_\ell, Y_\ell \le c_\ell\},
~ G_i \,:=\, \sum_{\ell \in \cI_j} \ind \{\hS_i \le \hS_\ell, Y_\ell \le c_\ell\},
~ G_i^+ \,:=\, \sum_{\ell \in \cI_j} \ind \{\hS_i < \hS_\ell, Y_\ell \le c_\ell\}.
\$
Using these notations, we can write  
\@ \label{eq:pval_1}
& \PP\big(p_j^\rand \le t, Y_{n+j} \le c_{n+j}\big) = 
\PP\big(U_{n+j} E_{n+j} + G_{n+j}^+ \le (n+1)t, Y_{n+j}\le c_{n+j} \big) \notag\\
=~& \frac{1}{n+1}\sum_{i\in\cI_j} 
\PP\big(U_{n+j}E_i + G_i^+ \le (n+1)t, Y_i \le c_i\big) \notag \\
=~& \frac{1}{n+1}\EE\bigg[\sum_{i\in \cN_j} 
\ind\Big\{U_{n+j}E_i + G_i^+ \le (n+1)t \Big\}\bigg],
\@ 
where the first step follows from exchangeability.
Let $i^*$ denote the index of the largest element in $\{\hS_i:i \in \cN_j, G_i^+ < (n+1)t\}$ 
(with ties broken arbitrarily). Then
\$
\eqref{eq:pval_1}=~& \frac{1}{n+1}\EE\Bigg[\ind\big\{|\cN_j| \le (n+1)t \big\} \cdot |\cN_j| 
+ \ind\{|\cN_j| > (n+1)t\}\cdot  \bigg(G_{i^*}^+  + E_{i^*} \cdot \ind
\Big\{U_{n+j} \le \frac{(n+1)t- G_{i^*}^+}{E_{i^*}}\Big\}\bigg) \Bigg]\\
\le ~& \frac{1}{n+1}\EE\bigg[\ind\Big\{|\cN_j| \le (n+1)t \big\} (n+1)t
+ \ind\{|\cN_j| > (n+1)t\}\cdot (n+1)t \Bigg] 
=  t.
\$
We have therefore proved the validity of $p_j$'s and $p_j^\rand$'s.

Next, we have 
\$ 
p_j = \frac{1+\sum_{i\in \cic}\ind\{Y_i \le c_i, \hS_i \ge \hS_{n+j}\}}{n+1} 
& \stepa{\le} \frac{1+\sum_{i\in \cic}\ind\{Y_i \le c_i, S_{i} \ge \hS_{n+j}\}}{n+1} \\
& \stepb{\le} \frac{1+\sum_{i\in \cic}\ind\{S_i \ge \hS_{n+j}\}}{n+1}  = \bar{p}_j.
\$
Above, the step (a) is because when $Y_i \le c_i$, $S_i \ge \hS_i$; 
when $S(x,y) = \hat{\mu}(x) - (\sup_x 2|\hat{\mu}(x)|) \cdot\ind\{y>c\}$, 
the ``$\le$'' is replaced by ``$=$'' in step (a) and (b).
A similar argument leads to that results corresponding to $p_j^\rand$. 
\end{proof}

\subsection{Additional results on the conformal p-value-based procedures}
\label{app:equiv}
Below, Lemma~\ref{lem:fixed_pval} and Lemma~\ref{lem:bf_ref} show 
that that fixed-threshold and BH-threshold selection rules are 
special instances of the general selection rule described in Section~\ref{sec:conf_pval} of the main text.
\label{subsec:conf_pval_procedures}
\begin{lemma}\label{lem:fixed_pval}
Let $\hcS_\fixed = \{j\in[m]: p_j\le q\}$. Then 
$\hcS_\fixed$ is equivalent to $\tS_\fixed = \{j\in [m]: \hcS_{n+j} \ge T\}$, where 
\$
T = \inf\bigg\{t \in \cT: 1 + \sum_{i\in \cic}\ind\{\hS_i \ge t,Y_i \le c_i\}
\ge q(n+1) \bigg\}.
\$ 
\end{lemma} 

\begin{proof}[Proof of Lemma~\ref{lem:fixed_pval}]
Suppose $j \in \hcS_\fixed$, then $p_j \le q$. 
By the definition of $p_j$, 
\$ 
\frac{1}{n+1}
\Big(1 + \sum_{i\in \cic}\ind\{\hS_i \ge \hS_{n+j},Y_i \le c_i\}\Big) = p_j \le q.
\$
By the definition of $T$, $\hS_{n+j} \ge T$, and therefore $j \in \tS_\fixed$. 
Conversely, if $\hS_{n+j} \ge T$, 
then 
\$ 
p_j = \frac{1}{n+1}\Big(1 + \sum_{i\in \cic}\ind\{\hS_i \ge \hS_{n+j},Y_i \le c_i\}\Big) 
\le \Big(1 + \sum_{i\in \cic}\ind\{\hS_i \ge T,Y_i \le c_i\}\Big) \le q,
\$
where the last step is by the definition of $T$. As a result, $j \in \hcS_\fixed$.   
We then have $\hcS_\fixed = \tS_\fixed$. 
\end{proof}

\begin{lemma}\label{lem:bf_ref}
Let $\hcS_\bh$ denote the selection set obtained by the BH procedure applied to the 
conformal p-values. 
Then $\hcS_\bh$ is equivalent to $\tS_\bh := \{j\in [m]: \hS_{n+j} 
\ge T^\bh\}$, where 
\$
T^\bh = \inf\bigg\{t \in \cT: \frac{1 + \sum_{i \in \cic}\{\hS_i \ge t, Y_i \le c_i\}}
{(\sum_{j\in [m]} \ind\{\hS_{n+j} \ge t\})\vee 1}\frac{m}{1+n} \le q\bigg\}.
\$
\end{lemma}

\begin{proof}[Proof of Lemma~\ref{lem:bf_ref}]
Let $j^* = \sum^m_{j=1}\ind\{\hS_{n+j} \ge T^\bh\}$ and
$\hS_{n+(1)},\hS_{n+(2)},\ldots,\hS_{n+(m)}$ denote the order statistics of
$\{\hS_{n+j}\}_{j\in [m]}$ in a descending order. It then suffices to 
show that $j^* = k^*$.
When $k^*\ge 1$, 
\@ \label{eq:bh_equivalence}
\frac{1 + \sum_{i \in \cic}\{\hS_i \ge \hS_{n+(k^*)}, Y_i \le c_i\}}
{(\sum_{j\in [m]} \ind\{\hS_{n+j} \ge \hS_{n+(k^*)}\})\vee 1}\frac{m}{1+n}
= p_{(k^*)} \frac{m}{k^*} \le q,
\@
where the last inequality follows from the definition of $k^*$.
From the definition of $T^\bh$, we can conclude that
$\hS_{n+(k^*)} \ge T^\bh$; since $\hS_{n+(j^*)}$ is the smallest 
element in $\{\hS_{n+j}\}_{j\in[m]}$ that is no less than $T^\bh$, 
we also have that $\hS_{n+(k^*)} \ge 
\hS_{n+(j^*)}\ge T^\bh$. Conversely, 
\$ 
p_{(j^*)} \frac{m}{j^*} & \stepa{=} 
\frac{1+ \sum_{i\in \cic} \ind\{\hS_i \ge \hS_{n+(j^*)}, Y_i \le c_i\}}{n+1} 
\cdot \frac{m}{\sum^m_{j=1}\ind\{\hS_{n+j} \ge T^\bh\}} \\
& \stepb{\le} 
\frac{1+ \sum_{i\in \cic} \ind\{\hS_i \ge T^\bh, Y_i \le c_i\}}
{\sum^m_{j=1}\ind\{\hS_{n+j} \ge T^\bh\}}
\cdot \frac{m}{n+1} \stepc{\le} q.
\$
Above, step (a) uses the definition of $j^*$, step (b)
is because $\hS_{n+(j^*)} \ge T^\bh$  and step (c)
is by the definition of $T^\bh$. From the form of 
the BH procedure, we have that $j^* \ge k^*$, and therefore $j^* = k^*$.

When $k^* = 0$, for any $t\in \cT$, we let $j(t) = \sum^m_{j=1}\ind\{\hS_{n+j}\ge t\}$.
If $j(t)>0$, then
\$ 
& \frac{1+ \sum_{i\in \cic} \ind\{\hS_i \ge t,Y_i \le c_i\}}{(\sum^m_{j=1}\ind\{\hS_{n+j} \ge t\}) \vee 1}
\frac{m}{n+1} 
 =   \frac{1+ \sum_{i\in \cic} \ind\{\hS_i \ge t, Y_i \le c_i\}}{j(t)} \frac{m}{n+1}\\
\ge~& \frac{1+ \sum_{i\in \cic} \ind\{\hS_i \ge \hS_{n+(j(t))},Y_i \le c_i\}}{j(t)} \frac{m}{n+1}
= p_{(j(t))} \cdot \frac{m}{j(t)}> q,
\$ 
where the last step is because $k^* = 0$. The above implies that 
$T^\bh > \hS_{n+(1)} \Rightarrow j^* = 0$.
The equivalence is therefore established.
\end{proof}

\subsection{Prediction set construction by further sample splitting} 
\label{app:splitting}
In this section, we discuss the possibility of simplifying the 
prediction set construction by further splitting the calibration 
set. Since we are primarily interested in the case where 
we do not have the flexibiility to change the selection rule, 
it turns out that---in most of the cases---the sample-splitting 
idea does not help with finding the exchangeable reference set  
without substantially changing the nature of selection. 

There is one class of selection rules, however, that can 
benefit from the sample-splitting idea with only minor modifications 
to the original selection rule.
To be concrete, consider the selection rule such that 
$\ind\{j \in \hat \cS\} = f(\cD_\calib, X_{n+j})$ for 
some function $f$, i.e., whether $j$ is selected depends solely on 
the calibration set $\cD_\calib$ and the feature $X_{n+j}$. 
Examples of such selection rules include thresholding conformal 
p-values with a fixed threshold, and selecting units based on properties of the 
preliminary prediction set.

For such selection rules, we could split the calibration set 
into two parts, $\cD_\calib = \cD_\calib^{(1)} \cup \cD_\calib^{(2)}$.
If we allow ourselves to determine the selection set via 
$\ind\{j \in \hat \cS\} = f(\cD_\calib^{(1)}, X_{n+j})$, 
then we can construct the prediction set as 
\@\label{eq:splitting_ps}
\begin{split}
& \hC_{\alpha,n+j} = \Big\{y\in \cY: V(X_{n+j},y) \le
\quant\big(\{V_i\}_{i \in \hcR}\cup\{\infty\}, 1-\alpha \big)\Big\},\\
& \text{where }\hcR := \{Z_i \in \cD_\calib^{(2)}: f(\cD_\calib^{(1)}, X_i) = 1\}.
\end{split}
\@
The following proposition shows that the prediction set constructed 
in~\eqref{eq:splitting_ps} is a valid $(1-\alpha)$ selection-conditional 
prediction set.
\begin{proposition}
Under the condition of Theorem~\ref{thm:PI_cond_set}, and suppose that for $\forall j \in[m]$, 
the selection rule is such that $\ind\{\ell \in \hat \cS\} = f(\cD_\calib, X_{n+\ell})$
for some function $f$. Then the prediction set $\hC_{\alpha,n+j}$ 
constructed in~\eqref{eq:splitting_ps} is a valid $(1-\alpha)$ prediction set.
\end{proposition}
\begin{proof}

Treating $\cD_\calib^{(2)}$ as the calibration set in Theorem~\ref{thm:PI_cond_set},
it is straight forward to check that $\hcR_{n+j}(y) = \hcR$, and therefore 
by Theorem~\ref{thm:PI_cond_set}, the prediction set $\hC_{\alpha,n+j}$ 
is a valid $(1-\alpha)$ prediction set.
\end{proof}

\subsection{Analysis of calibration sample size}
A practical concern of our proposed method is that accounting for the 
selection event may reduce the number of calibration samples, which 
has a direct impact on the size and stability of the prediction set.
The degree of effect is high dependent on the selection rule, which 
determines the fundamental difficulty of the problem: for example, 
if the selection rule picks out only the test unit with the largest 
predicted value, then conditional inference essentially becomes 
inference on the tail of a distribution; it is not surprising to 
see that only a limited number of calibration samples can be obtained. 

In this section, we provide an asymptotic analysis of 
a class of selection rules that are relatively ``mild''. To be 
concrete, we consider the selection rule yields a reference set 
in the form of $\{i\in [n]: h(X_i) \ge \hat q\}$, where $h$ is a function
of the features and $\hat q$ is a threshold that potentially dependent 
on $(\cD_\calib,\cD_\test)$ but converges to a constant as $n,m$ goes to infinity.
The following proposition characterizes the size of such reference set 
asymptotically.

\begin{proposition}
Suppose the reference set $\hcR = \{i \in [n]: h(X_i) \ge \hat q\}$, 
where $h$ is a fixed function of the features and $\hat q$ is a data-dependent 
threshold such that $\lim_{n,m \rightarrow} \hat q = q^*$ a.s. For $\forall t \in \RR$, let 
$G(t) = \PP(h(X) \ge t)$. Assume that $G(t)$ is continuous in $t$. 
Then the size of the reference set satisfies that 
\$ 
\lim_{n,m \rightarrow \infty} \frac{1}{n}|\hcR| = G(q^*) \text{ a.s.}
\$
\end{proposition}

\begin{proof}
By definition, we can write $|\hcR| = \sum_{i=1}^n \ind\{h(X_i) \ge \hat q\}$. 
By the strong law of large numbers, we have that 
\$
\Big|\frac{1}{n} \sum^n_{i=1} \ind\{h(X_i) \ge \hat q\} - G(\hat q)\Big| 
\le \sup_{t \in \RR}
\Big|\frac{1}{n} \sum^n_{i=1} \ind\{h(X_i) \ge t\} - G(t)\Big| 
\rightarrow 0 \text{ a.s.}
\$
Since $\lim_{n,m \rightarrow \infty} \hat q = q^*$,
and $G(t)$ is continuous in $t$, we have by the continuous mapping theorem that 
$\lim_{n,m \rightarrow \infty} |G(\hat q) - G(q^*)| = 0$.
By triangular inequality, we complete the proof. 
\end{proof}

Many selection rules we have considered in this paper fall into the class.
For example, for the quantile based selection rule, the reference set 
is $\{i\in [n]: S(X_i) \ge \hat q\}$, where $\hat q$ is either the 
joint quantile or the calibration quantile, which converges to the 
popuplation quantile as $n,m$ goes to infinity.  
The class of conformal p-value-based selection rules considered 
in Section~\ref{sec:conf_pval} can also be characterized 
by the proposition. To see this, consider the fixed threshold selection rule 
or the BH-threshold selection rule.  For simplicity 
suppose that $y \le c_{n+j}$ and $\mfS = 2^{[m]}$, 
then 
\$
\hcR_{n+j}(y) = \big\{i\in [n]: Y_i \le c_i, \hS_i \ge \tau^{(j)}(1,0) \big\}
\cup  \big\{i\in [n]: Y_i > c_i, \hS_i \ge \tau^{(j)}(1,1) \big\}.
\$ 
Above, $\tau^{(j)}(1,0)$ and $\tau^{(j)}(1,1)$ converge to the same constant 
for both variants (c.f.~\citet{jin2022selection}).

\end{document}